\documentclass{article}

\usepackage{arxiv}

\usepackage[utf8]{inputenc} 
\usepackage[T1]{fontenc}    
\usepackage{hyperref}       
\usepackage{url}            
\usepackage{booktabs}       
\usepackage{amsfonts}       
\usepackage{nicefrac}       
\usepackage{graphicx}
\usepackage{natbib}         
\usepackage{doi}            
\usepackage{amsmath,amssymb,amsthm} 
\usepackage{caption}
\usepackage{subcaption}
\usepackage{calrsfs}
\usepackage{tikz}
\usetikzlibrary{arrows.meta}
\usetikzlibrary{shapes.multipart}
\usepackage{listings}
\usepackage{multicol} 
\usepackage{enumitem} 
\usepackage{array} 
\usepackage{makecell} 
\usepackage[table]{xcolor} 
\usepackage[dvipsnames]{xcolor}
\usepackage{algorithmicx}
\usepackage{algcompatible}
\usepackage{algpseudocode}
\usepackage[english, ruled, linesnumbered]{algorithm2e} 
\usepackage{float}
\usepackage{tikz-cd}
\usepackage[table]{xcolor}
\usepackage{booktabs}
\usepackage{multirow}
\usepackage{lscape}
\usepackage{todonotes} 

\SetKwFunction{Return}{Return}
\newcommand{\variable}{\omega}
\newcommand{\inputvar}{\xi}

\def\operationsize{9mm}
\def\operationcolor{SkyBlue!40}
\def\variablecolor{SpringGreen!30}
\def\layercolor{Apricot!30}  
\def\inLayerComment{Bittersweet}
\def\arrowTip{Triangle}
\def\arrrowbend{10}

\theoremstyle{plain}
\newtheorem{theorem}{Theorem}[section]
\newtheorem{lemma}{Lemma}[section]

\newtheorem{corollary}{Corollary}[section]

\theoremstyle{definition}

\theoremstyle{remark}

\usepackage[titletoc,title]{appendix}

\DeclareMathAlphabet{\pazocal}{OMS}{zplm}{m}{n}

\SetCommentSty{mycommfont}

\title{Hybrid DeepONet Surrogates for Multiphase Flow in Porous Media}

\author{ 
        \href{https://orcid.org/0009-0006-0347-0309}{\includegraphics[scale=0.06]{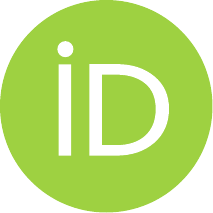}\hspace{1mm}Ezequiel S. Santos}\\
	Department of Civil Engineering\\
	COPPE - Federal University of Rio de Janeiro\\
	Rio de Janeiro, RJ \\
 	\texttt{ezequiel.souza@coc.ufrj.br} \\
         \And
        \href{https://orcid.org/0009-0000-1259-510X}{\includegraphics[scale=0.06]{orcid.pdf}\hspace{1mm}Gabriel F. Barros}\\
	Department of Civil Engineering\\
	COPPE - Federal University of Rio de Janeiro\\
	Rio de Janeiro, RJ \\
	\texttt{gabriel.barros@coc.ufrj.br} \\
        \And
        \href{https://orcid.org/0000-0003-0711-5970}{\includegraphics[scale=0.06]{orcid.pdf}\hspace{1mm}Amanda C. N. Oliveira}\\
	Department of Systems Engineering and \\Computer Sciences \\
	COPPE - Federal University of Rio de Janeiro\\
	Rio de Janeiro, RJ \\
	\texttt{amandacno@cos.ufrj.br} \\
	\And
	\href{https://orcid.org/0000-0003-1755-2605}{\includegraphics[scale=0.06]{orcid.pdf}\hspace{1mm}Rômulo M. Silva}\\
	Department of Civil Engineering\\
	COPPE - Federal University of Rio de Janeiro\\
	Rio de Janeiro, RJ \\
	\texttt{romulo.silva@coc.ufrj.br} \\
         \And
        \href{https://orcid.org/0000-0001-6036-8534}{\includegraphics[scale=0.06]{orcid.pdf}\hspace{1mm}Rodolfo S. M. Freitas}\\
	Department of Mechanical Engineering\\
	COPPE - Federal University of Rio de Janeiro\\
	Rio de Janeiro, RJ \\
	\texttt{rodolfosmfreitas@mecanica.coppe.ufrj.br} \\
        \And
        \href{https://orcid.org/0000-0003-4044-4266}
        {\includegraphics[scale=0.06]{orcid.pdf}\hspace{1mm}Dakshina M. Valiveti}\\ExonMobil Technology and Engineering Company\\
	Spring, TX, USA \\
    \texttt{dakshina.m.valiveti@exxonmobil.com}
        \And
        \href{https://orcid.org/0000-0002-8917-6760}
        {\includegraphics[scale=0.06]{orcid.pdf}\hspace{1mm}Xiao-Hui Wu}\\ExxonMobil Technology and Engineering Company\\
	Spring, TX, USA \\
    \texttt{xiao-hui.wu@exxonmobil.com }
        \And
        \href{https://orcid.org/0000-0001-8035-9651}{\includegraphics[scale=0.06]{orcid.pdf}\hspace{1mm}Fernando A. Rochinha}\\
	Department of Mechanical Engineering\\
	COPPE - Federal University of Rio de Janeiro\\
	Rio de Janeiro, RJ \\
	\texttt{faro@mecanica.ufrj.br} \\ 
        \And 
        \href{https://orcid.org/0000-0002-4764-1142}{\includegraphics[scale=0.06]{orcid.pdf}\hspace{1mm}Alvaro L. G. A. Coutinho}\\
	Department of Civil Engineering\\
	COPPE - Federal University of Rio de Janeiro\\
	Rio de Janeiro, RJ \\
	\texttt{alvaro@nacad.ufrj.br} \\ 
}

\date{}

\renewcommand{\headeright}{}
\renewcommand{\undertitle}{}
\renewcommand{\shorttitle}{Hybrid DeepONet Surrogates for Multiphase Flow in Porous Media}
\begin{document}
\maketitle

\begin{abstract}
The solution of partial differential equations (PDEs) plays a central role in numerous applications in science and engineering, particularly those involving multiphase flow in porous media. Complex, nonlinear systems govern these problems and are notoriously computationally intensive, especially in real-world applications and reservoirs. Recent advances in deep learning have spurred the development of data-driven surrogate models that approximate PDE solutions with reduced computational cost. Among these, Neural Operators such as Fourier Neural Operator (FNO) and Deep Operator Networks (DeepONet) have shown strong potential for learning parameter-to-solution mappings, enabling the generalization across families of PDEs. However, both methods face challenges when applied independently to complex porous media flows, including high memory requirements and difficulty handling the time dimension. To address these limitations, this work introduces hybrid neural operator surrogates based on DeepONet models that integrate Fourier Neural Operators, Multi-Layer Perceptrons (MLPs), and Kolmogorov-Arnold Networks (KANs) within their branch and trunk networks. The proposed framework decouples spatial and temporal learning tasks by splitting these structures into the branch and trunk networks, respectively. We evaluate these hybrid models on multiphase flow in porous media problems ranging in complexity from the steady 2D Darcy flow to the 2D and 3D problems belonging to the $10$th Comparative Solution Project from the Society of Petroleum Engineers. Results demonstrate that hybrid schemes achieve accurate surrogate modeling with significantly fewer parameters while maintaining strong predictive performance on large-scale reservoir simulations.
\end{abstract}

\keywords{Neural Operators \and DeepONets \and Fourier Neural Operators \and Kolmogorov-Arnold Networks \and Scientific Machine Learning \and Reservoir Engineering}

\section{Introduction}

\par 
Scientific Machine Learning (SciML) has been revolutionizing Computational Science and Engineering (CSE) over the past decade. The vast availability of data, the advent of specialized hardware, and the continuous development of machine learning algorithms are reshaping how learning-based methods are increasingly employed to solve complex problems traditionally tackled by the numerical approximation of partial differential equations (PDEs). By integrating physics and data, SciML enables the construction of surrogate models that accelerate numerical simulations while retaining the essential physical features of the underlying systems. Specifically in reservoir engineering applications, the use of SciML surrogate models  (proxy, metamodels) \cite{LATRACH2024212938} has been widely used in a plethora of applications such as accelerating reservoir simulations \cite{samnioti2023applicationsa, samnioti2023applicationsb, Bocoum2023} and carbon capture and storage (CCS) \cite{wen2021ccsnet, wen2022u, Witte2023}. Furthermore, surrogates are important in many query computations, such as parameter exploration, optimization, and uncertainty quantification, and also play an increasing role in digital twins \cite{10.2118/218461-MS, gahlot2024digital, willcox2023foundational} and digital shadows \cite{gahlot2025uncertainty}.

\par 
Among the several SciML techniques studied for a broad range of applications, methods based on deep neural networks (DNNs) have been widely applied for physical systems governed by PDEs. Multilayer Perceptrons (MLPs) are the foundation of DNN-based models \cite{Goodfellow-et-al-2016} and are mathematically guaranteed to be universal approximators of any nonlinear function \cite{augustine2024surveyuniversalapproximationtheorems, Funahashi1989, Hornik1989}. Their applications in physics-based systems include physics-informed neural networks (PINNs) \cite{raissi2019physics}, model order reduction \cite{solera2024beta, velho2025}, and surrogate modeling \cite{fresca2021pod, velho2025}. 
On the other hand, inspired by the Kolmogorov-Arnold representation theorem \cite{kolmogorov1957representation, arnold2009functions}, Kolmogorov-Arnold Networks (KANs) are capable of representing any continuous multivariate function as a sum of univariate functions \cite{liu2024kan}. Although this approach is effective for function representation, its application to high-dimensional data poses significant challenges, particularly in terms of memory usage and processing efficiency \cite{alter2024robustness, wang2024cest}. To increase flexibility, the function combinations in KANs use splines and radial basis functions, although reconciling interpretability with computational performance remains necessary \cite{vaca2024kolmogorov}. Unlike MLPs, KANs employ learnable activation functions on the network's edges, allowing for a more flexible, potentially more powerful representation via a series of univariate transformations of the input values. 
Both KANs and MLPs act as universal function approximators capable of representing complex nonlinear relationships between inputs and outputs. While this property is powerful for regression and inverse problems, standard MLPs often struggle to capture multiscale dynamics, nonlocal dependencies, and strong nonlinearities. 

\par 
A breakthrough in the SciML field came with the development of Neural Operators (NOs) \cite{kovachki2023neural, kovachki2024operator}, a class of architectures designed to learn mappings between infinite-dimensional function spaces rather than finite-dimensional functions. This operator-learning paradigm enables efficient and generalizable solution surrogates for parametric PDEs, where the objective is not merely to approximate a single trajectory of a physical system, but to learn the entire solution operator across a range of input conditions. Among these methods, Fourier Neural Operator (FNO) \cite{li2021fourierneuraloperatorparametric, grady2023model}, and the Deep Operator Network (DeepONet) \cite{Lu_2021} have gained significant attention for their robustness, generalization, and efficiency in learning complex PDE-based mappings. In reservoir engineering, both FNOs and DeepONets have been successfully applied to a variety of problems, including carbon capture and storage (CCS) \cite{wen2022u, Witte2023} and multiphase flow prediction \cite{badawi2024neural}. Although DeepONets have been proposed in this context \cite{udeeponet}, the use of FNOs as surrogate models for reservoir simulation has demonstrated strong generalization and efficient learning of spatially correlated structures. However, FNOs naturally operate in the spectral domain and therefore assume box-bounded or periodic geometries for the efficient application of the Fast Fourier Transform (FFT). Although progress has been made in this direction \cite{li2023geometryinformedneuraloperatorlargescale, LI2025117732}, this poses limitations when dealing with irregular domains and complex boundary conditions, which are common in realistic reservoir configurations. Another limitation of FNOs is their treatment of time as a fixed dimension \cite{diab2025}. For time-bounded surrogate modeling, FNOs lack the autonomy to extrapolate in time. In contrast, in autoregressive schemes, the time step sizes must be preserved, and the rollout sizes chosen for training largely influence performance \cite{mccabe2023stabilityautoregressiveneuraloperators}. Finally, their reliance on high-dimensional tensor representations results in significant memory consumption and computational expense during training, especially when scaling to three-dimensional space dimensions and high-resolution problems \cite{grady2023model, Liu2023domain}. 

\par 
To address these limitations, several hybrid extensions of standard Neural Operator architectures have been proposed in the recent literature \cite{diab2025, jiang2024fourier, udeeponet, CHO2025114430, HU2025114272}, each introducing distinct strategies and demonstrating varying degrees of success. In \cite{diab2025}, the authors focus on improving the temporal treatment of time-dependent PDEs and propose a DeepONet-based alternative to enhance temporal representation. In \cite{jiang2024fourier}, a FNO-based model with multiple-input capabilities is developed to improve efficiency in modeling multiphase CO2 flow. In the present study, we investigate hybrid Neural Operator surrogates for multiphase flow in porous media applications in oil reservoirs. Unlike previous approaches, we decouple space-time structures from the data and feed it into a DeepONet-based architecture. Here, the branch network encodes spatial features, while the trunk network represents the temporal domain. Owing to their modular structure, DeepONets naturally support the integration of complementary components such as FNOs, KANs, and MLPs. We leverage these hybrid configurations to enhance data efficiency and expressive power in complex reservoirs characterized by heterogeneous permeability fields. We design different combinations of hybrid variants within this structure for systematic evaluation to quantify their trade-offs in terms of predictive accuracy, generalization capability, and computational performance.

\par 
This paper is structured as follows: Section \ref{sec:nn_pdes} deals with how partial differential equations (PDEs) can be approximated and solved using neural network architectures. We cover the state-of-the-art of different strategies, such as DeepONets and FNOs, with applications focused on reservoir engineering and porous media flow. Throughout this work, we use a notation based on that introduced by Lu et al. \cite{lu2022comprehensive}. We also describe and elaborate on the relationship between KANs and MLPs, and discuss how these architectures can be merged to improve model performance in Section \ref{sec:hybrid}. We delve into a few hybrid architectures regarding DeepONets, FNOs, MLPs, and KANs. In Section \ref{Numerical_Experiments}, we propose experiments on multiphase flow in porous media and test the capabilities of the hybrid approaches. Finally, in Section \ref{sec:conclusions}, we present our concluding remarks.
\section{Neural Networks for PDEs}
\label{sec:nn_pdes}


\par 
The application of neural networks to solve PDEs has become a cornerstone of Scientific Machine Learning. Initial studies using MLPs to solve initial and boundary value problems date back to 1998 \cite{lagaris1998}. Since then, advances in techniques, architectures, and hardware have made neural networks an increasingly powerful tool for approximating spatio-temporal coherent structures arising from PDEs. Unlike classical numerical methods, which rely on the PDE approximation via spatial and temporal discretization, DNN-based approaches learn continuous representations of the solution space. 

\par 
 The introduction of physics-informed neural networks (PINNs) \cite{raissi2019physics} marked an important milestone by embedding the governing equations directly into the loss function, ensuring that the learned solutions satisfy the underlying physical laws. PINNs paved the way for broader integration of deep learning into computational physics, enabling data-driven models to complement conventional solvers in scenarios with noisy data, incomplete measurements, or high-dimensional parameter spaces. For reservoir engineering applications, PINNs have been extensively used to solve problems such as the Darcy equation for porous media flow \cite{Pu2022} and the Buckley-Leverett problem \cite{zhang2024physicsinformedneuralnetworksmultiphase, COUTINHO2023112265, Rodriguez-Torrado2022}. For large-scale reservoirs, due to the significant computational effort seen on fully connected networks, strategies such as domain decomposition \cite{HAN20233450}, convolutional PINNs \cite{math12203281, ZHANG2023111919}, and graph neural networks \cite{Zhao_2025} are employed to circumvent scalability issues.
 
PINNs are mostly built on MLPs. However, a different strategy for DNN-based models was presented for the resolution of physical systems \cite{liu2024kan}. Kolmogorov-Arnold Networks (KANs) are a neural architecture inspired by the Kolmogorov-Arnold superposition theorem \cite{kolmogorov1957representation, arnold2009functions, guilhoto2025} where conventional weight parameters are replaced by learnable univariate functions (often splines) \cite{liu2024kan}. In porous media flow applications, KAN-based PINNs have been applied to solve single-phase Darcy flow in porous media \cite{rao2025} and two-phase Buckley-Leverett flow \cite{kalesh2025}. These works have shown that KANs outperform traditional MLP architectures in performance. However, it is known that the total number of learnable parameters in KANs grows quadratically with network width \cite{wu2025pockanphysicsinformeddeepoperator}. Nevertheless, progress is being made toward improving the scalability of KANs \cite{LI2026109108, MOSTAJERAN2025114116}. In the following, we provide a brief description of MLPs and KANs.

\par
Let $\boldsymbol{\variable}\in \mathbb{R}^n$, be an input vector, and  $w_{ij}, a_i, d_i, \in \mathbb{R}$ be scalar parameters, with $i = 1, \ldots, m$ and $j = 1, \ldots, n$. An MLP defines a function $f_{\text{mlp}}(\boldsymbol{\variable}): \mathbb{R}^n \xrightarrow{} \mathbb{R}$ 
and fixed nonlinear activation functions $\sigma$~\cite{krose1996introduction,nielsen2015neural,gurney2018introduction}. For a single-layer MLP, we have
\begin{equation}\label{eq:mlp}
f_{\text{mlp}}(\boldsymbol{\variable}) = \sum_{i=1}^{m} a_i \sigma\left(\sum_{j=1}^{n} w_{ij} \variable_j + d_i\right)
\end{equation}

\noindent
On the other hand, a KAN defines its function $f_{\text{KAN}}(\boldsymbol{\variable}): \mathbb{R}^n \xrightarrow{} \mathbb{R}$ as a sum and composition of learnable univariate functions on the edges
~\cite{liu2024kan}. For a single layer KAN, we have
\begin{equation}\label{eq:kan}
f_{\text{kan}}(\boldsymbol{\variable}) = \sum_{i=1}^{2n+1} \psi_i \left( \sum_{j=1}^{n} \varphi_{ij}(\variable_j) \right) 
\end{equation}
\noindent where $\varphi_{ij}: \mathbb{R} \to \mathbb{R}$ are the inner univariate functions (with learnable parameters) that process each input component $\variable_j$ independently, and $\psi_i: \mathbb{R} \to \mathbb{R}$ are the outer univariate functions (with learnable parameters) that combine the results from the inner functions. The index $i$ ranges from $1$ to $2n+1$, following the Kolmogorov-Arnold theorem, though in practice this can be adjusted based on the desired network capacity. 

In this study, we further demonstrate that one-layered MLPs and KANs are mathematically equivalent in their representational capacity of a scalar function, as detailed in Appendix \ref{sec:app_a}. We show that, although an MLP is constructed as a sequence of alternating linear transformations and nonlinear activation functions, and a KAN is constructed by combining operations within a single functional block inspired by the Kolmogorov-Arnold representation, both architectures are capable of representing the same class of functions under appropriate parameterizations. In particular, this means that we can make the MLP as close to the KAN as desired (the reverse is also true) - for further details, see Appendix \ref{sec:app_a}.  MLPs and KANs can be extended to $L$ layers through proper operator compositions. For more details, we refer to \cite{liu2024kan}. 

\par 
Although both KANs and MLPs are used to learn how to approximate PDEs, a new paradigm was recently introduced in the SciML community. 
The advent of deep learning architectures capable of learning mappings between infinite-dimensional function spaces, known as Neural Operators (NOs), represents a pivotal advancement in computational science. Unlike MLPs, KANs, and other DNN-based networks, which map finite-dimensional vectors to finite-dimensional vectors, NOs are specifically designed to approximate solution operators $\mathcal{G}: \mathcal{V} \to \mathcal{U}$, where $\mathcal{V}$ and $\mathcal{U}$ are Banach spaces of functions \cite{kovachki2024operator}. 
This inherent capability allows NOs to be discretization-invariant, meaning they can be trained on data generated at one resolution and deployed to accurately predict solutions at arbitrary, often higher, resolutions without re-training (zero-shot super-resolution). 

\par 
The motivation for NOs reside on a typical problem in scientific computing, which involves finding the solution $u = \mathcal{G}_\theta (v)$, where $v$ is a function representing parameters (e.g., initial conditions, boundary conditions, or source terms), $u$ is the solution function (e.g., velocity, temperature field), and $\theta$ are the learnable parameters of the model. In order to simplify the notation, $\theta$ will be suppressed throughout the text, and the operator will be displayed as $u = \mathcal{G}(v)$.%
The operator $\mathcal{G}$ might be the inverse of a differential operator defined by a PDE. The general formulation of a Neural Operator layer can be expressed as an iterative composition of linear integral operators and non-linear activation functions, mirroring the structure of standard Neural Networks, that is,

\begin{equation}
z_{n+1}(\boldsymbol{\variable}) = \sigma \left( \mathcal{K}(z_n)(\boldsymbol{\variable}) + \textbf{W} z_n(\boldsymbol{\variable}) \right),
\label{eq:neural_operator_layer}
\end{equation}

where $z_n$ is the feature representation at layer $n$, $\sigma$ is a non-linear activation function (such as ReLU), 
$\textbf{W}$ is a local linear operator and $\mathcal{K}$ is a global integral operator, 

\begin{equation}
\mathcal{K}(z_n)(\boldsymbol{\variable}) = \int_\Omega \kappa(\boldsymbol{\variable}, y) z_n(y) dy,
\label{eq:integral_operator}
\end{equation}

where $\Omega$ is the spatial domain, and $\kappa(\boldsymbol{\variable}, y)$ is a learnable kernel function. The fundamental difference lies in how $\mathcal{K}$ is parameterized and calculated. 

Two major NO architectures are fundamental in the SciML community: Deep Operator Networks (DeepONets) and Fourier Neural Operators (FNOs). DeepONet \cite{lu2019deeponet} is an architecture motivated by the universal approximation theorem for operators, which suggests that a combination of two sub-networks can approximate any continuous non-linear operator. The output of a DeepONet is structured as a generalized inner product of two components: the branch network and the trunk network. The core objective of DeepONet is to approximate the output function $u(\inputvar) = \mathcal{G}(v)(\inputvar)$ at a specific query location $\inputvar \in \Omega_\inputvar$. In the DeepONet setting, the branch network takes the input function $v$ as a finite-dimensional vector of sensor readings $\mathbf{v} = \{v(\variable_1), v(\variable_2), \dots, v(\variable_m)\}$ at a fixed set of $m$ locations and maps this discrete input to a latent vector $\mathbf{b}(\textbf{v}) \in \mathbb{R}^r$, with $r$ being the size of the output dimension. 
The trunk network, on the other hand, takes the coordinates of the output query location $\inputvar$ as its input. It maps the coordinates $\inputvar \in \Omega_\inputvar$ to another latent vector $\mathbf{t}(\inputvar) \in \mathbb{R}^r$. 
The final output approximation $\mathcal{G}(\mathbf{v})(\inputvar)$ is computed by the inner product of the outputs from the two networks, that is,


\begin{equation} \label{eq:deeponet_ouput}
\mathcal{G} (\mathbf{v})(\inputvar) = \sum_{k=1}^r b_k(\textbf{v}) t_k(\inputvar) \,,
\end{equation}

where $b_k$ and $t_k$ are the $k$-th components of the latent vectors $\mathbf{b}$ and $\mathbf{t}$, respectively. The full function $\mathcal{G}(\mathbf{v})$ is then approximated by evaluating this expression over a set of desired query points $\boldsymbol{\variable}$. The DeepONet scheme described in this section is also represented in Figure \ref{fig:Architecture_DeepOnet}. 

On the other hand, the Fourier Neural Operator (FNO) \cite{li2021fourierneuraloperatorparametric} is a mesh-free architecture that parameterizes the integral kernel $\mathcal{K}$ directly in the Fourier, or frequency, domain, leveraging the computational efficiency of the Fast Fourier Transform (FFT). FNO is built upon the idea that, when the kernel $\kappa(\boldsymbol{\variable}, y)$ in the integral operator described in Equation \ref{eq:integral_operator} is translation-invariant, the integral operation becomes a convolution

\begin{equation}
\mathcal{K}(z_n)(\boldsymbol{\variable}) = (z_n * \kappa)(\boldsymbol{\variable}).
\label{eq:kernel_fno}
\end{equation}

Convolutions in the physical domain correspond to element-wise multiplication in the frequency (Fourier) domain \cite{li2020fourier, kovachki2023neural}. This principle allows an integral operator to be implemented efficiently through three key steps. First, the input feature map is transformed from the spatial domain to the frequency domain using the (Discrete) Fourier Transform, resulting in $\hat{z}_n = \mathcal{F}(z_n)$. Next, a learned, parameterized linear operator $\mathbf{R}$ is applied to the lower-frequency modes of $\hat{z}_n$. This operation is typically sparse, since higher-frequency modes are truncated both to reduce computational cost and to act as a low-pass filter, yielding $\hat{z}'_n = \mathbf{R} \cdot \hat{z}_n$. Finally, the modified spectral representation is transformed back into the spatial domain through the Inverse Fourier Transform, recovering the updated feature map $z'_n = \mathcal{F}^{-1}(\hat{z}'_n)$. The total FNO layer then combines this spectral convolution with a local linear transformation $\textbf{W}$ (analogous to the term in Equation \ref{eq:neural_operator_layer})

\begin{equation}
z_{n+1}(\boldsymbol{\variable}) = \sigma \left( \mathcal{F}^{-1}(\mathbf{R} \cdot \mathcal{F}(z_n))(\boldsymbol{\variable}) + \textbf{W} z_n(\boldsymbol{\variable}) \right) \,.
\label{eq:fno_output}
\end{equation}

The FNO's ability to perform global convolution efficiently (by operating on the entire domain simultaneously via the Fourier transform) makes it highly effective at capturing long-range dependencies, which are critical in many PDE solutions, such as those governed by advection or diffusion. This makes FNO computationally superior to CNN-based methods, which require numerous layers to achieve a comparable receptive field \cite{li2021fourierneuraloperatorparametric}. The overall structure of a FNO is represented in Figure \ref{fig:Architeture_FNO}.

\par 
In porous media applications, the use of neural operators to construct surrogate models for parameter exploration and temporal prediction is abundant in the literature. For instance, in the reservoir simulation context, NO-based models have been applied with good results in carbon capture and storage \citep{wen2021ccsnet, udeeponet, wen2022u}, reservoir engineering \cite{badawi2024neural, Tang2024, chandra2024fourier, Liu2023domain}, and other porous media applications \cite{XING2025114385}. Although many variations have been proposed to tackle the natural limitations of this family of methods in porous media applications, such as memory usage \cite{Liu2023} and temporal treatment \cite{diab2025}, other implementations aim to improve existing architectures. For instance, Wavelet Neural Operators \cite{Witte2023} have been proposed to improve the efficiency and performance of vanilla FNOs for CCS in large-scale (around $2$M cells) reservoirs.

\begin{figure}
            \centering
            \begin{tikzpicture}
        \tikzstyle{varstyle}=[circle, minimum size=\operationsize, fill=\variablecolor];
        \tikzstyle{operation}=[circle, minimum size=\operationsize, fill=\operationcolor];
        \tikzstyle{layer} = [rounded corners, minimum width=3*\operationsize, minimum height=\operationsize, fill=\layercolor]

        \node[varstyle] (v) at (0,\operationsize) { $\textbf{v}$ };
        \node[varstyle] (xi) at (0,-\operationsize) { $\inputvar$ };

        \node[layer] (branch) at (3*\operationsize, \operationsize) {Branch Net};
        \node[layer] (trunk) at (3*\operationsize, -\operationsize) {Trunk Net};

        \coordinate (bedge) at (6*\operationsize, \operationsize);
        \coordinate (tedge) at (6*\operationsize, -\operationsize);

        \node[operation] (merge) at (6*\operationsize, 0) {$\odot$};

        \node[layer, fill=\variablecolor, minimum width=4.6*\operationsize] (out) at (10*\operationsize,0) { $\mathcal{G}(\textbf{v})(\inputvar) = \displaystyle\sum_{k=1}^r b_k(\textbf{v}) t_k(\inputvar)$ };

        
        \path[-]
                (branch) edge (bedge)
                (trunk) edge (tedge);
                
        \path[-{\arrowTip}] 
    			(v) edge (branch)
    			(xi) edge (trunk)
                (bedge) edge (merge)
                (tedge) edge (merge)
                (merge) edge (out);

\end{tikzpicture}
            \caption{
                DeepONet Architecture. 
                The model is comprised of two subnetworks: a \textit{branch} and a \textit{trunk} network. The branch network takes as input the function $v$, represented by the vector $\textbf{v}$, and the trunk network takes as input the query point $\inputvar$. The outputs of the networks are combined through a merge operation (inner product, Hadamard product, linear/non-linear transformation), and the model outputs the result of the operator $\mathcal{G}(\textbf{v})$ in query point $\inputvar$, that is, $\mathcal{G}(\textbf{v})(\inputvar)$. 
            }
            \label{fig:Architecture_DeepOnet}
\end{figure}
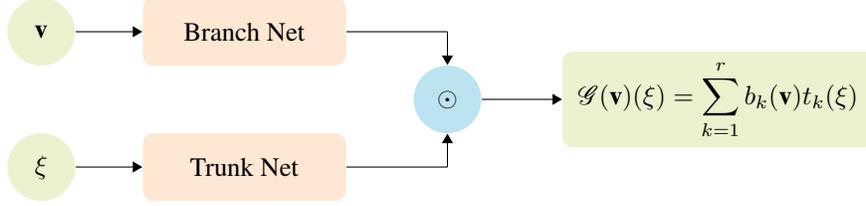

\begin{figure}
            \centering
            {

    \small
    \subfloat[FNO layer structure]{
        \begin{tikzpicture}
            \tikzstyle{varstyle}=[circle, minimum size=\operationsize, fill=\variablecolor];
            \tikzstyle{operation}=[circle, minimum size=\operationsize, fill=\operationcolor];
            \tikzstyle{layer} = [rounded corners, minimum width=3*\operationsize, minimum height=\operationsize, fill=\layercolor]
        
            \node[varstyle] (input) at (0,0) { $v(\variable)$ };
            \node[operation] (p) at (1.5*\operationsize, 0) { $\textbf{P}$ };
            \node[layer] (fl1) at (4*\operationsize, 0) { Fourier Layer 1 };
            \node[layer] (fl2) at (7.5*\operationsize, 0) { Fourier Layer 2 };
            \node[circle, minimum size=\operationsize] (etc) at (10*\operationsize, 0) { $\cdots$ };
            \node[layer] (fln) at (12.5*\operationsize, 0) { Fourier Layer $T$ };
            \node[operation] (q) at (15*\operationsize, 0) { $\textbf{Q}$ };
            \node[varstyle] (output) at (16.5*\operationsize, 0) { $u(\variable)$ };
        
            \path[-{\arrowTip}] 
        			(input) edge (p)
        			(p) edge (fl1)
                    (fl1) edge (fl2)
                    (fl2) edge (etc)
                    (etc) edge (fln)
                    (fln) edge (q)
                    (q) edge (output);
        \end{tikzpicture}
    } \\
    \vspace{3mm}
    \subfloat[Fourier layer $n$]{
        \begin{tikzpicture}
            \tikzstyle{varstyle}=[circle, minimum size=\operationsize, fill=\variablecolor];
            \tikzstyle{operation}=[circle, minimum size=\operationsize, fill=\operationcolor];
            \tikzstyle{layer} = [rounded corners, minimum width=3*\operationsize, minimum height=\operationsize, fill=\layercolor]
    
            \node[varstyle] (input) at (0,0) { $z_{n-1}$ };  
            \node[layer, minimum width=4.5*\operationsize, minimum height=1.5*\operationsize] (fourier) at (4*\operationsize, 1*\operationsize) {
                \begin{tikzpicture}
                    \tikzstyle{operation}=[circle, minimum size=\operationsize, fill=\operationcolor];
                    
                    \node[operation] (fft) at (0, 0) { $\mathcal{F}$ };
                    \node[operation] (r) at (1.5*\operationsize, 0) { $\textbf{R}$ };
                    \node[operation] (invfft) at (3*\operationsize, 0) { $\mathcal{F}^{-1}$ };
    
                    \path[-{\arrowTip}] 
                        (fft) edge (r)
                        (r) edge (invfft);
                \end{tikzpicture}
            };
            \node[operation] (w) at (4*\operationsize, -1*\operationsize) { $\textbf{W}$ };
            \node[operation] (sum) at (8*\operationsize, 0) { $+$ };
            \node[operation] (sigma) at (10*\operationsize, 0) { $\sigma$ };
            \node[varstyle] (output) at (12*\operationsize,0) { $z_{n}$ };  
    
            \path[-{\arrowTip}] 
                (input) edge[bend left=\arrrowbend] (fourier)
                (input) edge[bend right=\arrrowbend] (w)
                (fourier) edge[bend left=\arrrowbend] (sum)
                (w) edge[bend right=\arrrowbend] (sum)
                (sum) edge (sigma)
                (sigma) edge (output);
            
        \end{tikzpicture}
    }
}
            \caption{
                Fourier Neural Operator architecture. 
                Figure (a) shows the overall structure, in which the input is first lifted to a higher dimension through operator $\textbf{P}$, then passed through Fourier layers, and lastly projected back into the output space dimension through operation $\textbf{Q}$. 
                Figure (b) shows the structure of a Fourier layer, which comprises of the application of the FFT ($\mathcal{F}$), a linear transformation $\textbf{R}$ and the reverse FFT ($\mathcal{F}^{-1}$), then the result is summed with a local linear transformation $\textbf{W}$, and lastly a non-linear activation function $\sigma$ is applied, generating the layer output.
            }
            \label{fig:Architeture_FNO}
\end{figure}
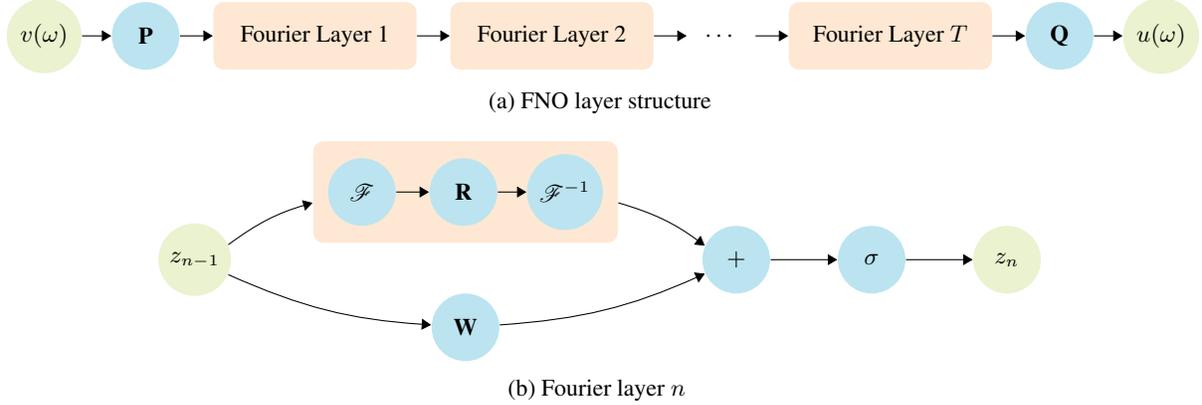

\section{Hybrid Deep Operator Networks}
\label{sec:hybrid}

\par
Although Deep Operator Networks (DeepONets) and Fourier Neural Operators (FNOs) are grounded in operator learning, they exhibit distinct behaviors when applied to spatio-temporal data. For purely spatial partial differential equations (PDEs), both architectures have demonstrated robust performance \cite{HU2025114272}. However, when addressing transient PDEs, their treatment of the temporal dimension diverges significantly. In the case of FNOs, two main formulations are commonly adopted in the literature: Markovian or autoregressive approaches \cite{diab2025, mccabe2023stabilityautoregressiveneuraloperators}, and direct time-mapping strategies \cite{li2020neural, Witte2023}. Despite their success in porous media applications \cite{badawi2024neural, wen2021ccsnet, Liu2023}, both formulations exhibit limited flexibility when extrapolating to unseen time horizons or varying time-step resolutions.

\par
Besides the time dimension modeling challenges, standard FNO implementations are also characterized by a substantial memory footprint \cite{george2025tensorgalore}, a large number of learnable parameters \cite{LI2025117732, tran2021factorized}, and an inherent dependence on rectangular computational domains \cite{Liu2023domain}. In large-scale porous media simulations, such constraints become particularly critical. It is not uncommon for FNO-based surrogate models representing reservoirs with approximately $10^6$ grid cells to require on the order of $O(10^8 \sim 10^9)$ trainable parameters. For instance, \cite{Witte2023} reported an FNO-based model for a carbon capture and storage (CCS) application with two million grid cells comprising approximately $226$ million parameters. Similarly, \cite{Liu2023} proposed a domain decomposition strategy to train an FNO model for a reservoir with $819$k cells, fitting the model across two NVIDIA GeForce RTX 3090 Ti GPUs with 24 GB of VRAM each. In another example, \cite{chandra2024fourier} required eight NVIDIA A100 GPUs to accommodate an FNO model for a reservoir comprising 428k grid cells. 

\par 
Hybrid schemes have recently emerged in the literature to address these challenges. In \cite{seabra2024ai}, the authors proposed hybrid data assimilation frameworks that combine FNO and Transformer U-Net surrogates to accelerate and improve uncertainty quantification in CO2 storage simulations. To overcome the geometric constraints of box-bounded reservoirs, the Domain-Agnostic FNO \cite{liu2023domainagnosticfourierneural} and Geometry-Informed Neural Operator \cite{li2023fourier} were introduced. Several other studies have also explored strategies to improve temporal modeling in time-dependent PDEs using different neural operator architectures \cite{jiang2024fourier, udeeponet, diab2025}.

\par 
In this section, we present a hybrid Neural Operator architecture that integrates conventional neural network components with operator-based learning techniques. This hybridization strategy enables more efficient scaling for the numerical approximation of large-scale transient problems by explicitly decoupling spatial and temporal learning within the branch and trunk networks. The main objective is to enhance the representation of temporal dynamics while reducing the overall memory footprint of operator learning for high-dimensional porous media flow data. Unlike previous studies that rely on multi-GPU training \cite{Witte2023, Tang2024}, domain-decomposition strategies \cite{Liu2023domain}, or spatial slicing techniques to fit the problem into hardware constraints \cite{Witte2023, 10898911, 10.1145/3534678.3539045}, the proposed approach aims to reduce the intrinsic complexity of Neural Operator architectures. This design facilitates the handling of coupled spatio-temporal PDEs without the need for excessively large parameter counts, thereby improving computational efficiency and scalability.

\par
The hybridization of DeepOnet with FNOs, MLPs, or KANs is based on the universal approximation theorem for neural operators \cite{lu2019deeponet} and guarantees that such hybrid architectures can approximate a broad class of nonlinear operators. 
The hybrid scheme is depicted in Figure \ref{fig:hybrid_scheme}. The base model is a DeepONet, and we test different configurations for both the branch and trunk networks. The branch net, which handles spatially coherent structures in the data, can be structured as an MLP, KAN, or FNO. For temporal treatment, the trunk net is either a KAN or an MLP model. The model is constructed using NVIDIA PhysicsNeMo \cite{physicsnemo2023}, an open-source deep-learning framework for SciML models built on top of PyTorch \cite{paszke2019pytorch}. PhysicsNeMo has built-in implementations of the DeepONet, FNO, and MLP architectures, which are employed in this work. For the KAN implementation, we adapt our own algorithm based on the Pykan \footnote{\url{https://github.com/KindXiaoming/pykan}} \cite{liu2024kan} and efficient-kan \footnote{\url{https://github.com/Blealtan/efficient-kan}} repositories. The model is merged into the PhysicsNeMo environment by inheriting the respective base model classes. For all DeepONet configurations, the loss function used to optimize the weights is the Mean Square Error (MSE), that is,

\begin{equation}
    MSE = \frac{1}{N_{s}} \sum_{j=1}^{N_{s}} \frac{1}{n_{\inputvar_{(j)}}} \sum_{i=1}^{n_{\inputvar_{(j)}}} \left[ \mathcal{G}(\mathbf{v}_j)(\mathbf{\inputvar}_i) - G_{gt}(\mathbf{v}_j)(\mathbf{\inputvar}_i) \right]^2  \,,
\end{equation}

where $G_{gt}(\cdot)$ corresponds to the ground truth the operator wishes to deduce from the data, $N_s$ is the number of samples, and $n_{\xi_{(j)}}$ is the number of query points, $\xi_i$, in $j$-th simulation in the dataset.

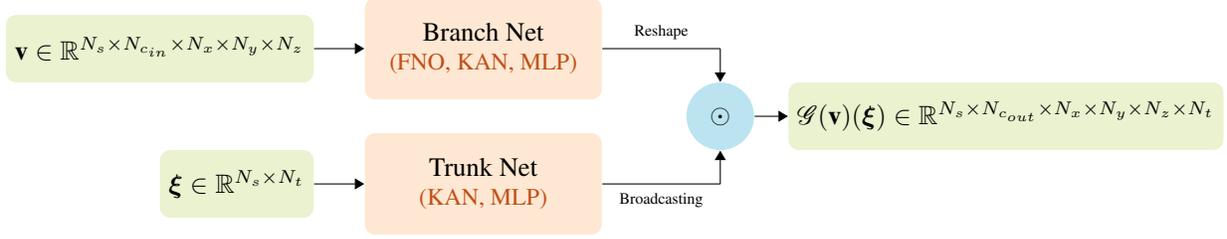
\begin{figure}
    \centering

\begin{tikzpicture}
        \tikzstyle{varstyle}=[rounded corners, minimum width=2*\operationsize, minimum height=\operationsize, fill=\variablecolor, anchor=east];
        \tikzstyle{operation}=[circle, minimum size=\operationsize, fill=\operationcolor];
        \tikzstyle{layer} = [rounded corners, minimum width=3.5*\operationsize, minimum height=1.5*\operationsize, fill=\layercolor, align=center]

        \node[varstyle] (v) at (0,\operationsize) { $\textbf{v} \in \mathbb{R}^{N_s \times N_{c_{in}} \times N_x \times N_y \times N_z}$ };
        \node[varstyle] (xi) at (0,-\operationsize) { $\boldsymbol{\inputvar} \in \mathbb{R}^{N_s \times N_t}$ };

        \node[layer] (branch) at (2.5*\operationsize, \operationsize) {
            Branch Net \\ {\small\color{\inLayerComment} (FNO, KAN, MLP)}
        };
        \node[layer] (trunk) at (2.5*\operationsize, -\operationsize) {
            Trunk Net \\ {\small\color{\inLayerComment} (KAN, MLP)}
        };

        \coordinate (bedge) at (6*\operationsize, \operationsize);
        \coordinate (tedge) at (6*\operationsize, -\operationsize);

        \node[operation] (merge) at (6*\operationsize, 0) {$\odot$};

        \node[varstyle, anchor=west] (out) at (7*\operationsize,0) { $\mathcal{G}(\textbf{v})(\boldsymbol{\inputvar}) \in \mathbb{R}^{N_s \times N_{c_{out}} \times N_x \times N_y \times N_z \times N_t} $ };

        
        \path[-]
                (branch) edge[anchor=south] node {{\tiny Reshape}} (bedge)
                (trunk) edge[anchor=north] node {{\tiny 
                Broadcasting}} (tedge);
                
        \path[-{\arrowTip}] 
    			(v) edge (branch)
    			(xi) edge (trunk)
                (bedge) edge (merge)
                (tedge) edge (merge)
                (merge) edge (out);

\end{tikzpicture}
    \caption{Hybrid DeepONet model. The input features are subdivided into spatial $\textbf{v}$ and temporal $\boldsymbol{\inputvar}$ information. The spatial information $\textbf{v}$ is processed by the branch network, which can encapsulate an FNO, a KAN, or an MLP model, while the temporal information $\boldsymbol{\inputvar}$ is processed by the trunk network, which comprises a KAN or an MLP model. The outputs of each subnetwork are then combined through a Hadamard product to generate the model output $\mathcal{G}(\mathbf{v})(\boldsymbol{\inputvar})$.}
    \label{fig:hybrid_scheme}
\end{figure}


\par 
The forward pass algorithms of a DeepOnet and its hybrid versions are shown respectively in Algorithms~\ref{alg:deepONet} and \ref{alg:deepONet_FNO}. DeepONet comprises two main networks: the Branch network, which encodes the input functions, and the Trunk network, which encodes the operator's evaluation locations. The outputs of these networks are combined through a Hadamard product to generate the final approximation of the target operator. These algorithms structure this approach, providing a methodology for constructing and training DeepONet models. In the algorithms $N_s$ is the number of samples, $N_t$ is the number of snapshots, $N_x$, $N_y$, and $N_z$ are the number of grid points in the $x$, $y$, and $z$ directions, $N_{c_{in}}$ is the number of input channels, and $N_{c_{out}}$ is the number of output channels. For the hybrid schemes that use either MLP or KAN in their branch network, the algorithm used in described in Algorithm \ref{alg:deepONet}, as for combinations using FNO in their branch network, we refer to Algorithm \ref{alg:deepONet_FNO}. An important step in Algorithm \ref{alg:deepONet} is the permutation used to process all channels in the branch net when using a KAN or MLP.

\begin{algorithm}
\footnotesize
\caption{The forward-pass of a DeepONet model}
\label{alg:deepONet}
\SetAlgoLined

\KwIn{$\mathbf{v} = v(\boldsymbol{\variable})\in\mathbb{R}^{N_s\times N_{c_{in}}\times N_x\times N_y\times N_z}$, 
      $\boldsymbol{\inputvar}\in\mathbb{R}^{N_s\times N_t}$} 

\KwOut{$\mathcal{G}(\textbf{v})(\boldsymbol{\inputvar})\in\mathbb{R}^{N_s\times N_{c_{out}}\times N_t\times N_x\times N_y\times N_z}$}

Initialize models $\texttt{branch\_model}\in\{\texttt{MLP},\texttt{KAN}\}$ and $\texttt{trunk\_model}\in\{\texttt{MLP},\texttt{KAN}\}$; \\
\tcp{Output shapes: \texttt{branch\_model} $\to N_t\cdot N_{c_{out}}$, \ \texttt{trunk\_model} $\to N_t$}

\For{$i\leftarrow 1$ \KwTo $N_s$}{
    Permute: $\textbf{v}[i] \in\mathbb{R}^{N_{c_{in}}\times N_x\times N_y\times N_z}\;\leftarrow\; \textbf{v}[i] \in\mathbb{R}^{N_z\times N_x\times N_y\times N_{c_{in}}}$\;
    
    Compute: $\textbf{b}[i]\leftarrow \texttt{branch\_model}\big(\textbf{v}[i]\big)$\;
    \tcp{shape $=(N_z,\;N_x,\;N_y,\;N_t\cdot N_{c_{out}})$}
    
    Permute: $\textbf{b}[i]\leftarrow \textbf{b}[i]\in\mathbb{R}^{\,N_t\cdot N_{c_{out}}\times N_x\times N_y\times N_z}$\;
    
    Reshape: $\textbf{b}[i]\leftarrow \textbf{b}[i]\in\mathbb{R}^{\,N_{c_{out}}\times N_t\times N_x\times N_y\times N_z}$\;
    
    Compute: $\textbf{t}[i]\leftarrow \texttt{trunk\_model}\big(\boldsymbol{\inputvar}[i]\big)$\;
    \tcp{shape $=(1,\;N_t)$}
    
    Broadcast: $\textbf{t}[i]\leftarrow \textbf{t}[i]\in\mathbb{R}^{\,N_{c_{out}}\times N_t\times N_x\times N_y\times N_z}$\;
    
    Compute: $\; \mathcal{G}(\textbf{v})(\boldsymbol{\inputvar})[i]\leftarrow \textbf{t}[i]\odot \textbf{b}[i]\;$ \tcp{Hadamard product}
}

\Return{$\mathcal{G}(\textbf{v})(\boldsymbol{\inputvar})$}\; \tcp{shape = $(N_s,N_{c_{out}},N_t,N_x,N_y,N_z)$}
\end{algorithm}

\begin{algorithm}
\footnotesize
\caption{The forward-pass of a Hybrid DeepONet-FNO model}
\label{alg:deepONet_FNO}
\SetAlgoLined

\KwIn{$\mathbf{v} = v(\boldsymbol{\variable}) \in \mathbb{R}^{N_s \times N_{c_{in}} \times N_x \times N_y \times N_z}$, 
      $\boldsymbol{\inputvar} \in \mathbb{R}^{N_s \times N_t}$}

\KwOut{$\mathcal{G}(\textbf{v})(\boldsymbol{\inputvar}) \in \mathbb{R}^{N_s \times N_{c_{out}} \times N_t \times N_x \times N_y \times N_z}$}

Initialize the models: $\texttt{branch\_model} \in \{ \texttt{FNO} \}$ and $\texttt{trunk\_model} \in \{\texttt{MLP}, \texttt{KAN}\}$\;
\tcp{Output shapes: \texttt{branch\_model} = $N_t \cdot N_{c_{out}}$ and \texttt{trunk\_model} = $N_t$}

\For{$i = 1$ \KwTo $N_s$}{
    Compute: $\textbf{b}[i] \gets \texttt{branch\_model}(\textbf{v}[i])$\;
    \tcp{shape = $(N_t \cdot N_{c_{out}}, N_x, N_y, N_z)$}
    
    Reshape: $\textbf{b}[i] \gets \textbf{b}[i] \in \mathbb{R}^{N_{c_{out}} \times N_t \times N_x \times N_y \times N_z}$\;
    
    Compute: $\textbf{t}[i] \gets \texttt{trunk\_model}(\boldsymbol{\inputvar}[i])$\;
    \tcp{shape = $(1, N_t)$}
    
    Broadcast: $\textbf{t}[i] \gets \textbf{t}[i] \in \mathbb{R}^{N_{c_{out}} \times N_t \times N_x \times N_y \times N_z}$\;
    
    Compute: $\mathcal{G}(\textbf{v})(\boldsymbol{\inputvar})[i] \gets \textbf{t}[i] \odot \textbf{b}[i]$\;
    \tcp{Hadamard product}
}
\Return{$\mathcal{G}(\textbf{v})(\boldsymbol{\inputvar})$}\; \tcp{shape = $(N_s, N_{c_{out}}, N_t, N_x, N_y, N_z)$}
\end{algorithm}



\section{Numerical Experiments}
\label{Numerical_Experiments}

\par 
In this section, we cover our numerical experiments validating our implementation of the hybrid model and assessing its capabilities. We begin with the 2D steady Darcy flow problem, which serves as a fundamental test case and baseline for our methodology. This two-dimensional elliptic problem models the pressure field in porous media governed by Darcy's law, with permeability fields sampled from a Gaussian process. Despite its simplified setting, it is widely adopted in the literature as a canonical benchmark to validate operator learning methods. Then, we test the model's ability to generalize to reservoir simulations (i.e., a time-dependent, nonlinear, coupled system of PDEs). We select the 10th Comparative Solution Project (CSP) from the Society of Petroleum Engineers (SPE) \cite{Christie2001} as our benchmark problem. This benchmark, also known as SPE10, comprises two reservoir models. The original purpose of SPE10 was to test contestants' ability to perform upscaling in reservoirs with complex permeability fields. Here, we take advantage of the natural complexity progression of the two models. SPE10 Model 1 is a relatively small 2D domain saturated with oil with one gas injection well and one production well. SPE10 Model 2 is a complex 3D reservoir with $1{,}122{,}000$ cells, one water injection well, and four production wells. Although SPE10 was proposed back in 2001, the complexity of the SPE10 Model 2 is still addressed to validate state-of-the-art solvers and SciML models.  We test different combinations for the branch and trunk networks, as shown in Table \ref{trunk_branch_exemplo2}. We test the proposed hybrid model's ability to serve as a surrogate, generalizing and predicting unseen scenarios beyond the provided training dataset. All numerical experiments were trained on a NVIDIA H100 GPU with $94$ GB of VRAM.

\subsection{2D Darcy Flow}
The Darcy flow problem is modeled by a second-order elliptic PDE, given by:
\begin{align}
    -\nabla \cdot (k(\mathbf{x}) \nabla p(\mathbf{x})) = 1, \quad \Omega \in [0, 1]^2, \\
    p(\mathbf{x}) = 0 \quad \text{on } \partial (0,1)^2,
\end{align}
where $p(\mathbf{x})$ is the pressure field, $k(\mathbf{x})$ is the permeability field, sampled from a Gaussian process:
\begin{equation}
    k(\mathbf{x}) \sim \mathcal{N}(0,(-\Delta + 9I)^{-2}).
\end{equation}

\begin{table}
    \centering
    \caption{
        Configurations of the hybrid DeepONet architectures. The models used in the branch and in the trunk networks are specified for each proposed architecture. 
    }
        \begin{tabular}{lcc}
            \hline
            \textbf{Architecture} & \textbf{Branch} & \textbf{Trunk} \\
            \hline
            DeepONet (FNO + KAN) & FNO & KAN \\
            DeepONet (FNO + MLP) & FNO & MLP \\
            DeepONet (KAN) & KAN & KAN \\
            DeepONet (MLP) & MLP & MLP \\
            \hline
        \end{tabular}
    \label{trunk_branch_exemplo2}
\end{table}

The objective is to learn the mapping from the permeability field $ k(\mathbf{x}) $ to the corresponding pressure field $ p(\mathbf{x})$. Given that this problem is steady and two-dimensional, the dimensions of the problem are $N_x = 240$, $N_y = 240$, $N_z = 0$, and $N_t = 0$. A total of 1500 samples are used for training and 300 for testing. Solutions $p(\mathbf{x})$ used as ground truth for training are obtained by using a second-order finite difference solver, which can be obtained at the PhysicsNeMo repository\footnote{\url{https://github.com/NVIDIA/physicsnemo/tree/main/examples/cfd/darcy_transolver}}. For this problem, the branch network input is the permeability field $N_{c_{in}} = 1$ and the trunk network input are the $\mathbf{x} = (x, y)$ coordinates. The model's output is the predicted pressure field through the surrogate model, that is, $N_{c_{out}} = 1$. Notice that although the proposed hybrid schemes are initially intended to deal with spatio-temporal PDEs, this canonical example is a steady-state case. In this first experiment, we test the use of the hybrid neural operator to learn how to approximate purely spatial PDEs. Another aspect that differs in this example is the data normalization, which is done by subtracting the mean and dividing by the standard deviation of the dataset. We choose the SiLU activation function for all configurations in the trunk, while the branch network uses the Tanh activation function. All models are trained using $500$ epochs, and the Adam optimizer is used. The initial learning rate is $0.9$ with the cosine learning rate decay. For hybrid setups that use FNO in the branch network, the lifting and projection layers are of size $32$ and $1$ FNO layer with eight modes is used. The output channel is a single channel. 
For MLP networks, regardless of branch or trunk networks, the setup used 
has $240$ input features, two layers of size $32$ and $240$ output features. 
KANs are built similarly to MLPs, with a spline order of $4$. 
The hyperparameters were chosen after an ablation study measuring the relative error $2$-norm for the test set, that is, $||p||_{rel}=\frac{||p-p^*||_2}{||p^*||_2}$,  where $p^*$ represents the ground truth.

\par 
In Figure \ref{exemplo_2_training}, we can see the training (left) and the test (right) loss evolution for all hybrid configurations. While they show small discrepancies, the comparison of pressure predictions for the test set, shown in Figure \ref{exemplo_2_val}, reveals that the hybrid models combining KAN and FNOs present smaller pointwise absolute errors, and smaller errors in both infinity, $||p||_{\infty}$, and relative error $2$-norm.

\begin{figure}
    \centering
    \begin{subfigure}{0.45\textwidth}
        \centering
        \includegraphics[width=0.95\textwidth]{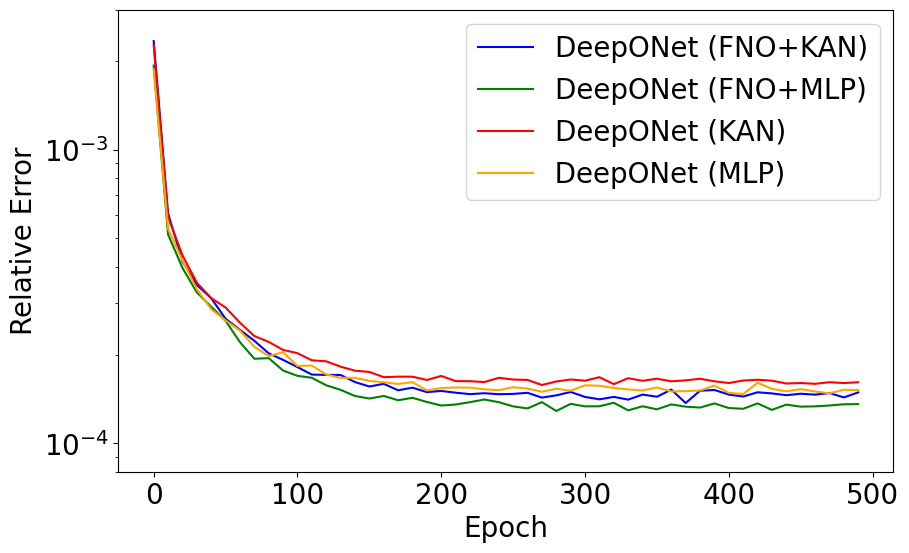}
        \caption{Training Loss.}
        \label{exemplo_2_train_1_0}
    \end{subfigure}
    \begin{subfigure}{0.45\textwidth}
        \centering
        \includegraphics[width=\textwidth]{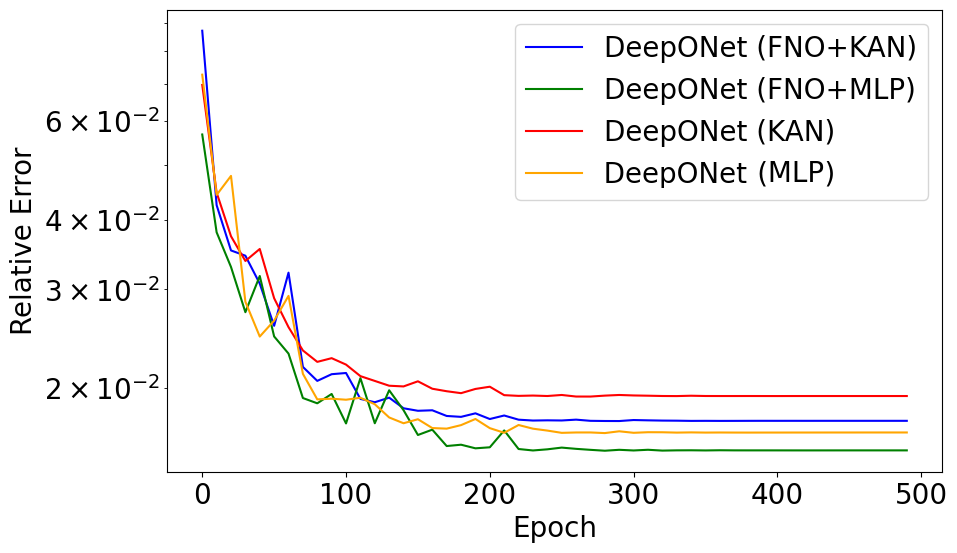}
        \caption{Test Loss.}
        \label{exemplo_2_train_1_1}
    \end{subfigure}
    \caption{
        Training (a) and test (b) relative error $2-$norms of the proposed hybrid models throughout 500 training epochs in the Darcy flow problem.
    }
    \label{exemplo_2_training}
\end{figure}

\begin{figure}
    \centering
    \begin{subfigure}{1.0\textwidth}
        \includegraphics[width=\linewidth]{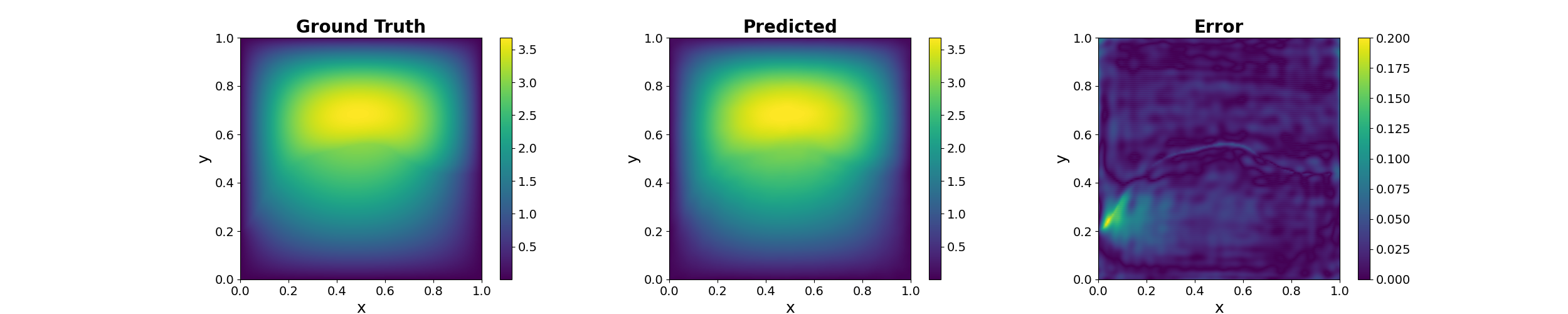}
        \caption{DeepONet (FNO+KAN).}
        \label{exemplo_2_val_1_0}
    \end{subfigure}
    \hfill
    \begin{subfigure}{1.0\textwidth}
        \includegraphics[width=\linewidth]{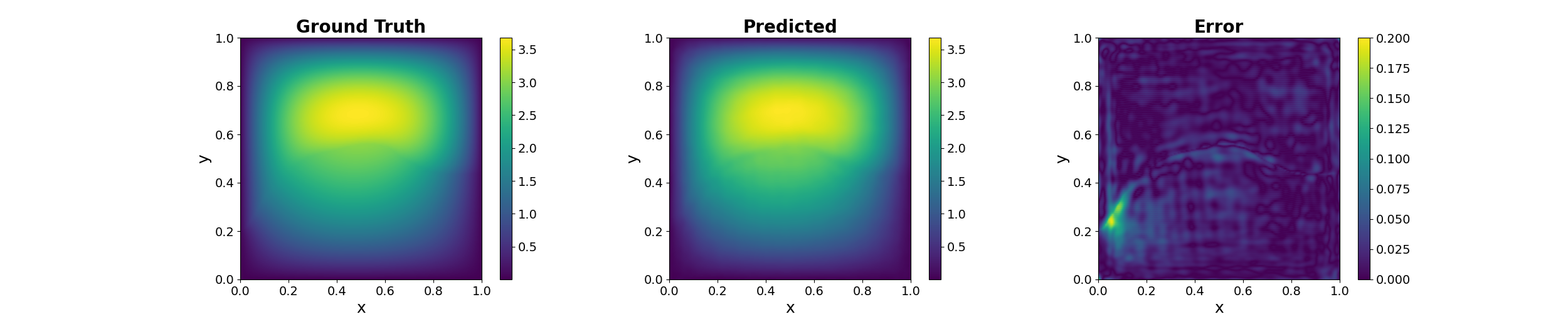}
        \caption{DeepONet (FNO+MLP).}
        \label{exemplo_2_val_1_1}
    \end{subfigure}
    \hfill
    \begin{subfigure}{1.0\textwidth}
        \includegraphics[width=\linewidth]{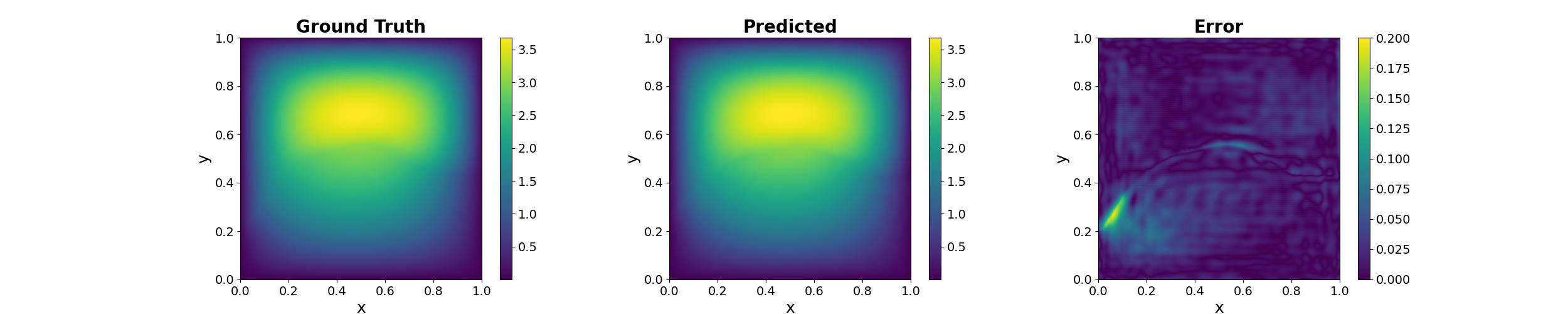}
        \caption{DeepONet (KAN).}
        \label{exemplo_2_val_1_2}
    \end{subfigure}
    \hfill
    \begin{subfigure}{1.0\textwidth}
        \includegraphics[width=\linewidth]{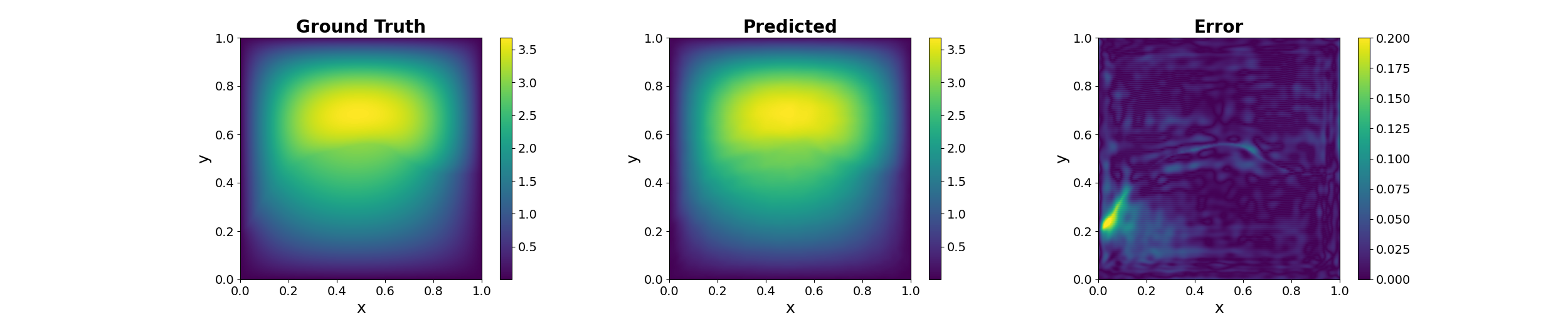}
        \caption{DeepONet (MLP).}
        \label{exemplo_2_val_1_3}
    \end{subfigure}
    \caption{Test results for 2D Darcy flow. The column on the left contains the ground truth pressure for comparison, the middle column the predicted pressure fields, and the right column the pointwise abolute error for: (a) DeepONet (FNO + KAN), $||p||_{\infty} = 0.201$ and $||p||_{rel} = 0.0495$; (b) DeepONet (FNO + MLP),  $||p||_{\infty} = 0.184$ and $||p||_{rel} = 0.0476$; (c) DeepONet (KAN),  $||p||_{\infty} = 0.190$ and $||p||_{rel} = 0.0545$; and (d) DeepONet (MLP),  $||p||_{\infty}= 0.234$ and $||p||_{rel} = 0.0635$.}
    \label{exemplo_2_val}
\end{figure}

\subsection{Reservoir simulation}

\par 
The steady Darcy flow is often used to validate SciML models, but their application in real-world problems is limited. Reservoir simulation is dictated by the physical phenomenon of unsteady multiphase fluid flow in porous media. The governing equations are obtained from the integration of mass, momentum, and energy conservation together with thermodynamic equilibrium \cite{chen2007reservoir}. Typically, reservoir simulation is handled using two major model families: black-oil models and compositional fluid models. Black-oil models are simpler models used to express simple phase behavior phenomena, based on the assumption that any point along the flow in the reservoir to the surface facilities is considered as a binary fluid composed of stock tank oil and surface gas. The phase behavior phenomena are quantified using PVT properties that depend only on pressure and temperature, disregarding effects due to the exact fluid composition. Compositional models, on the other hand, monitor changes in the fluid composition at every time instant and throughout the domain. In this case, phase behavior calculations require the solution of stability and flash calculations based on an Equation of State (EoS) model. Compositional models are more accurate and substantially more expensive than black-oil models. For a more descriptive comparison between black-oil models and compositional fluid models, we refer to \cite{Fevang2000, Coats1998}.

\par 
In this study, we focus on the former. The mathematical formulation of the black-oil model is derived from the principle of mass conservation for each pseudo-component $ \alpha \in \{w, o, g\} $, which represents water, oil, and gas, respectively. This results in a system of coupled, non-linear PDEs
\begin{equation} \label{eq:blackoil}
    \frac{\partial}{\partial t} \left( \phi_{\text{ref}} A_\alpha \right) + \nabla \cdot \textbf{u}_\alpha + q_\alpha = 0 \,\,,
\end{equation}
where the first term represents mass accumulation and the second term describes the flux. The accumulation terms, $A_\alpha$,  as well as their respective component velocities $\mathbf{u}_\alpha$, are defined as: 
\begin{align}
    &A_w = m_\phi b_w s_w \,,   &\textbf{u}_w = b_w \textbf{v}_w \,, \\
    &A_o = m_\phi (b_o s_o + r_{og} b_g s_g) \,,   &\textbf{u}_o = b_o \textbf{v}_o + r_{og} b_g \textbf{v}_g \,, \\
    &A_g = m_\phi (b_g s_g + r_{go} b_o s_o) \,,   &\textbf{u}_g = b_g \textbf{v}_g + r_{go} b_o \textbf{v}_o \,, 
\end{align}
in which $\phi_\text{ref}$ is the reference porosity, $m_\phi$ is a pressure-dependent multiplier, $b_\alpha$ is the phase shrinkage/expansion factor, $s_\alpha$ is the phase saturation, and $r_{go}, r_{og}$ are the mass ratios of dissolved gas in oil and vaporized oil in gas, respectively. The component velocities, $\textbf{u}_\alpha$, are related to the phase fluxes, $\textbf{v}_\alpha$, which are governed by the multiphase extension of Darcy's law,
\begin{equation}
    \textbf{v}_\alpha = - \lambda_\alpha \textbf{K}\left(\nabla p_\alpha - \rho_\alpha \textbf{g} \right)
\end{equation}
with $\textbf{K}$ being the absolute permeability tensor, $\lambda_\alpha$ the phase mobility (relative permeability divided by viscosity), $p_\alpha$ the phase pressure, $\rho_\alpha$ the phase density, and $\textbf{g}$ the gravitational acceleration vector. 

Aside from the governing PDEs, two additional physics constraints are required. The first is that the phase saturations must sum to unity within the pore volume, 
\begin{equation}
    s_w + s_o + s_g = 1 \,\,,
\end{equation}
and the second relates the phase pressures through capillary pressure, $p_c$, which is a function of saturation, that is, 
\begin{align}
    p_{c,ow}(s_w) &= p_o - p_w \,\,, \\
    p_{c,og}(s_g) &= p_g - p_o \,\,. \label{eq:closure}
\end{align}

\par 
To solve the black-oil model, we choose the OPM Flow simulator \cite{opmflow}, an open-source library that solves this system of equations using a fully implicit numerical scheme. The spatial domain is discretized with an upwind finite-volume method, while time is discretized using an implicit (backward) Euler scheme. The resulting large system of non-linear algebraic equations is solved simultaneously at each time step using a Newton-Raphson linearization method coupled with a preconditioned iterative linear solver. OPM Flow uses input decks -- plain-text files with a defined structure -- to set simulation parameters. More details on how OPM Flow approximates the black-oil model can be found in \cite{opmflow}. 

\par 
We chose the 10th SPE Comparative Solution Project (or simply SPE10) \cite{Christie2001} as our problem of choice, given that this application is well-studied, and our results can be compared with the literature. This benchmark includes two reservoir models, Model 1 and Model 2, which are further described in this section. 

\subsubsection{SPE10 Model 1}
The SPE10 Model 1 is a two-phase (oil and gas) model with no dipping or faults. The dimensions of the model are $2500 \text{ft long} \times 25 \text{ft wide} \times 50 \text{ft}$ thick on a $100 \times 1 \times 20$ mesh. Initially, the model is fully saturated with oil, as gas is injected at a constant rate. Figure \ref{fig:perm} shows the isotropic and heterogeneous permeability field provided for the problem, as well as the injection and production wells locations. Although SPE10 Model 1 is a low-dimensional problem, the permeability field ranges over $6$ orders of magnitude.
The benchmark injection rate is set to $0.30$ Mscf/day, the constant porosity is $0.2$, and the production well bottomhole pressure (BHP) is set to 95 psia. For all parameters used in the SPE10 model, we refer to \cite{Christie2001}.

\begin{figure}
    \centering
    \includegraphics[width=1\linewidth]{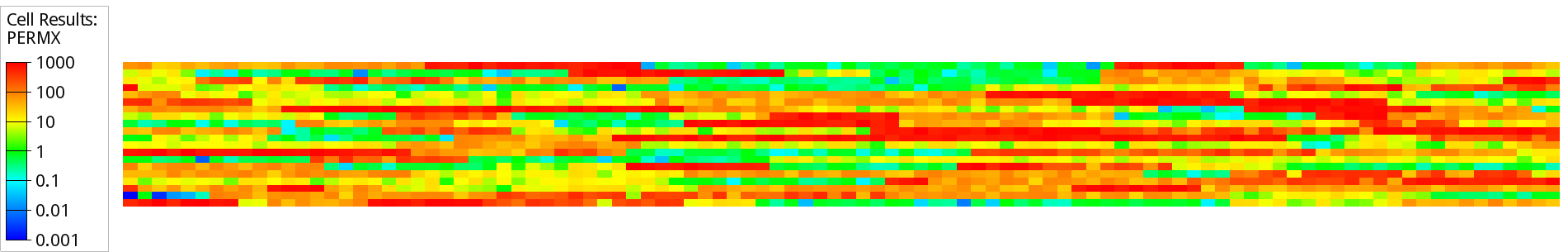}
    \caption{Correlated permeability field in the $x$ direction. Notice that, for this benchmark, the logarithmic scale reveals a permeability variation of $6$ orders of magnitude. The $z$-axis (vertical) has been exaggerated by a factor of $5$.}
    \label{fig:perm}
\end{figure}

\par 
For this first experiment, we generate $100$ simulations of the SPE10 Model 1 setup by varying the gas injection rates as,
\begin{equation}
    q_g = \text{Unif}(0.23, 0.37) \,,
\end{equation}
where $\text{Unif}(a_1, a_2)$ is the uniform distribution. For each generated sample of $q_g$, a new simulation is evaluated, generating saturation and pressure fields as well as oil curves for each value of gas injection rate. Figure \ref{fig:spe10_data_gen} shows the envelope of cumulative oil production for all simulations, as well as some individual curves. 

\begin{figure}
     \centering
     \begin{subfigure}[b]{0.49\textwidth}
            \centering
            \includegraphics[width=\linewidth]{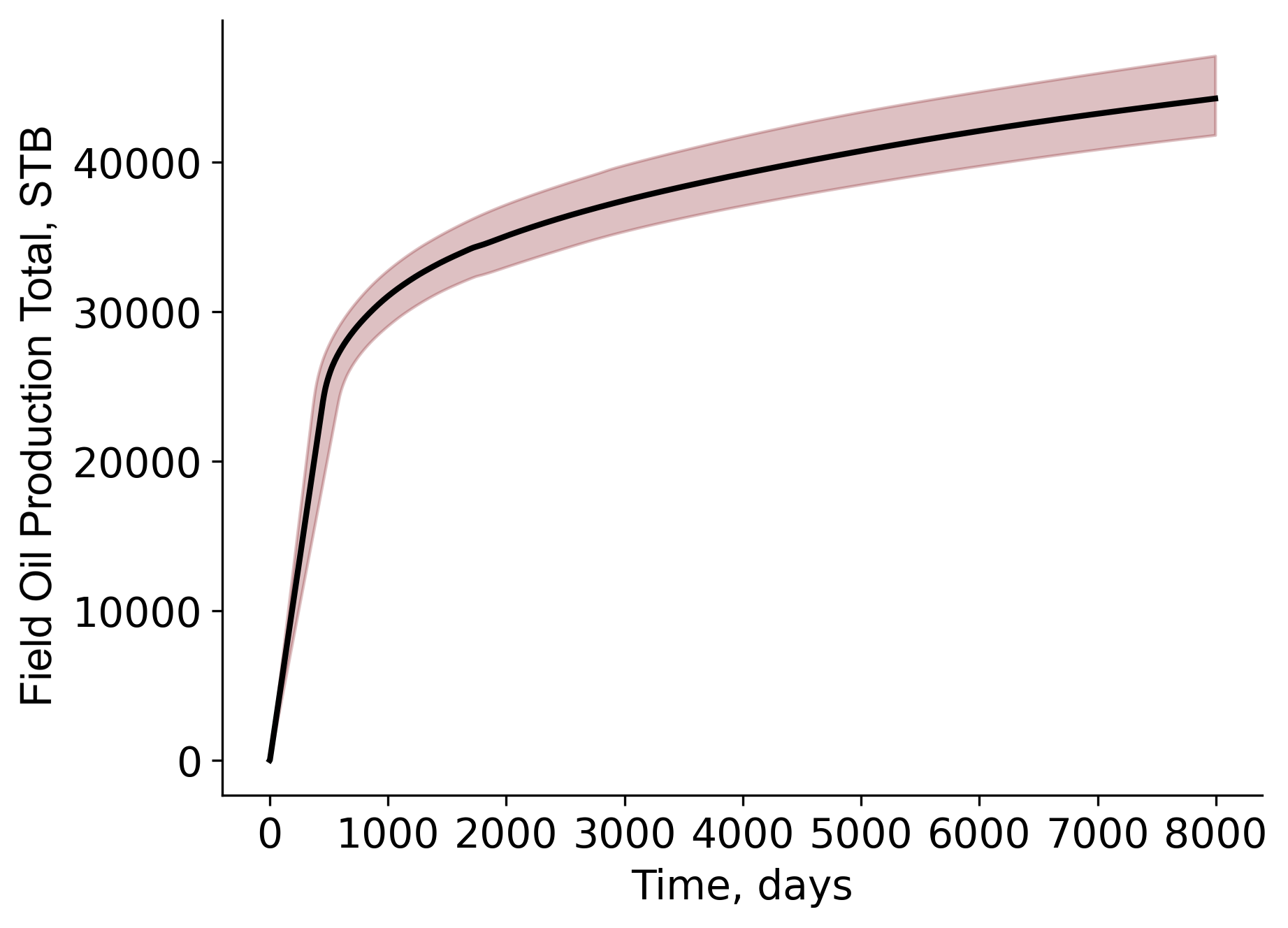}
            \subcaption{Field oil production total envelope for $100$ simulations}
            \label{fig:envelope_spe10}
     \end{subfigure}
     \hfill
     \begin{subfigure}[b]{0.49\textwidth}
            \centering
            \includegraphics[width=\linewidth]{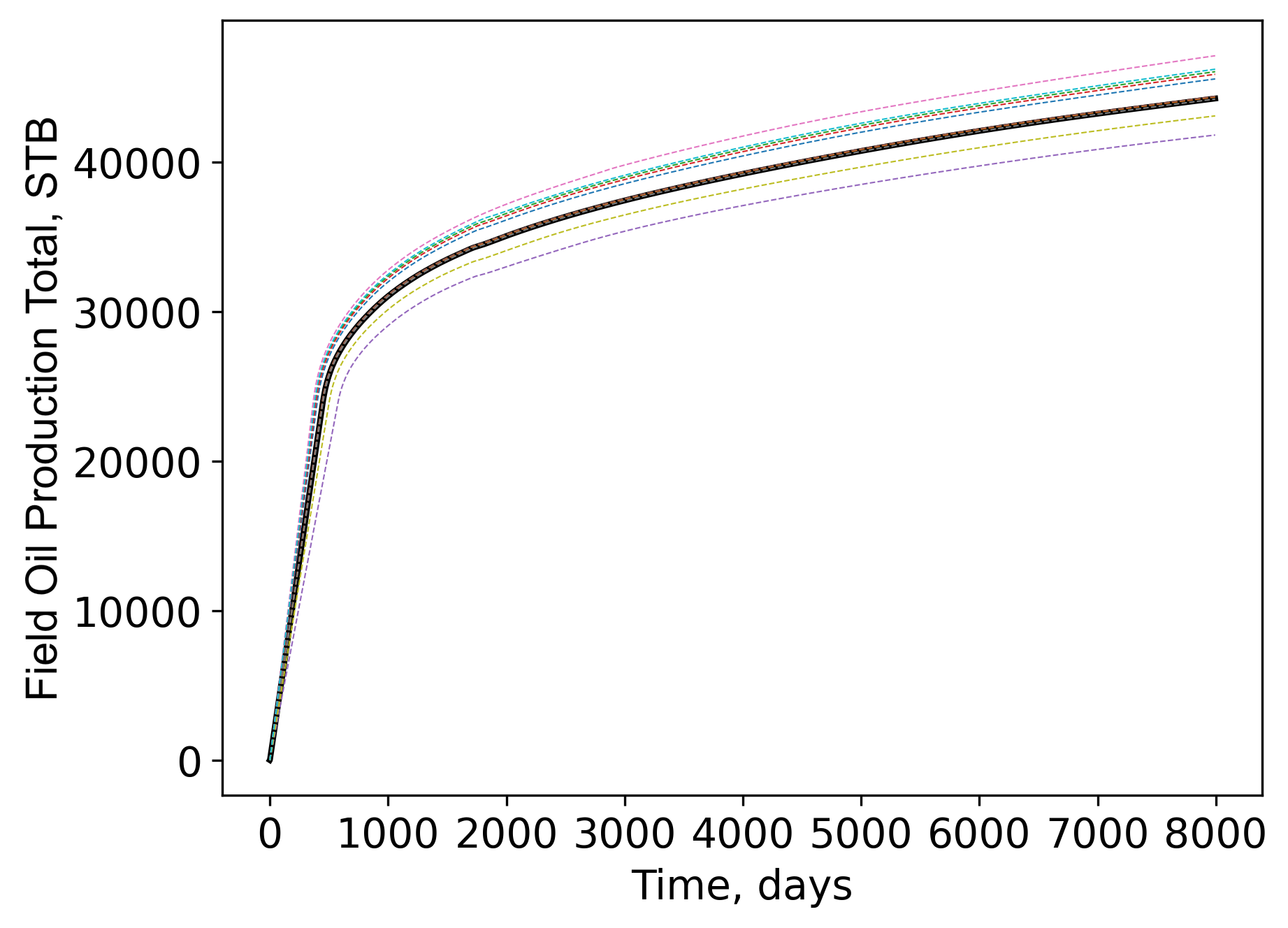}
            \subcaption{Field oil production total curves for a sample of $10$ simulations}
            \label{fig:cases_spe10}
     \end{subfigure} \\ 
     \begin{subfigure}[b]{0.99\textwidth}
            \centering
            \includegraphics[width=0.75\linewidth]{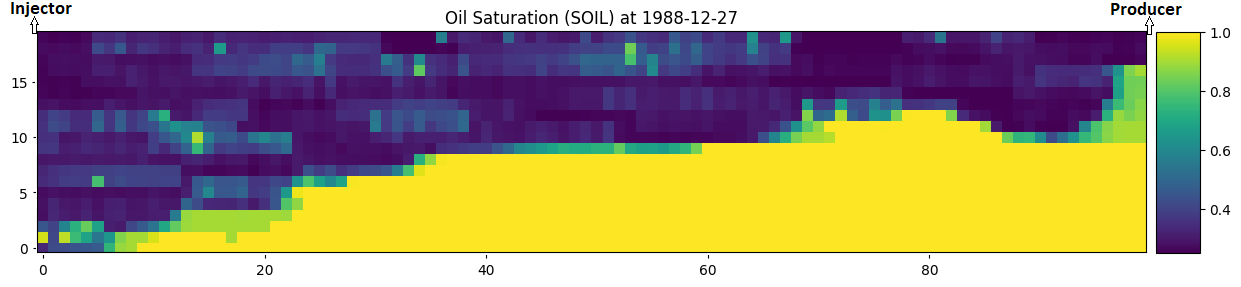}
            \subcaption{Oil saturation field for the last report step of the SPE10 Model 1 benchmark (using injection rate of $0.30$ Mscf/day).}
            \label{fig:spe10_sat_field}
     \end{subfigure}
    \caption{Outputs of the numerical simulation. The Figure in the top left shows the envelope of field oil production curves for $100$ simulations on SPE10 Model 1, while the top right shows a sample of $10$ curves. The black line denotes the original SPE10 Model 1 solution in both cases. The bottom Figure shows the oil saturation field at the last report step for the benchmark using an injection rate of $0.30$ Mscf/day. The injection rates are generated using a uniform distribution between $0.23$ Mscf/day and $0.37$ Mscf/day.}
    \label{fig:spe10_data_gen}
\end{figure}

\par 
Unlike the steady 2D Darcy flow, the SPE10 Model 1 benchmark is a transient, two-phase, nonlinear problem, and the dynamics of the components should be properly accounted for. For each simulation, OPM Flow generates grid fields and summary values at every report or time step. Grid fields are primary spatio-temporal variables that are described by each cell of the mesh at all report steps, such as the component's saturations and pressure, while summary values are post-processed scalar quantities computed at every solver time step, such as each well's oil production rates and bottomhole pressures. For instance, the integration of oil saturation, a grid field, through the existing well model logic in OPM Flow generates, at every solver time step, a summary value, which is the post-processed scalar plotted in each curve in Figure \ref{fig:cases_spe10}. OPM Flow reads and writes binary files in formats used by the ECLIPSE simulator from Schlumberger, since these are the dominant file formats supported by most pre- and post-processing tools in reservoir engineering \cite{opmflow}. So, in order to extract and structure grid and scalar data from the simulations, the res2df package \footnote{\url{https://equinor.github.io/res2df/index.html}} from Statoil/Equinor is used. 

\par 
Both grid and post-processed variables are crucial in reservoir engineering and are important in our attempt to generate proxy models for the numerical simulations. However, their fundamental differences in data structure may affect the ingestion of both types of information in a machine learning model. For example, geological input data, such as permeabilities and porosity, require steady spatial information, whereas operational parameters, such as variations in well BHPs and injection rates during oil production, are temporal scalars. In order to assemble all variables into a single tensor for modeling, we use the bit mask approach proposed by Badawi \& Gildin \cite{badawi2024neural}, where the spatial distribution of wells is assigned to a zero-filled tensor. That is, temporal scalar quantities become spatio-temporal tensors where the well location in the grid is a non-zero value, which is the desired scalar. The same procedure is done to extract the scalar quantities predicted in the surrogate model (for instance, the oil production rate at a given time instant). Knowing the cells corresponding to the production well, we average the predicted quantities and transform the result back into a scalar. Figure \ref{fig:bit_mask} illustrates both procedures. After assembling the input tensor, the models can be trained. We normalize the entries using the neuralop \cite{kossaifi2024neural, kovachki2023neural} implementation for the Unit Gaussian normalizer.


\begin{figure}
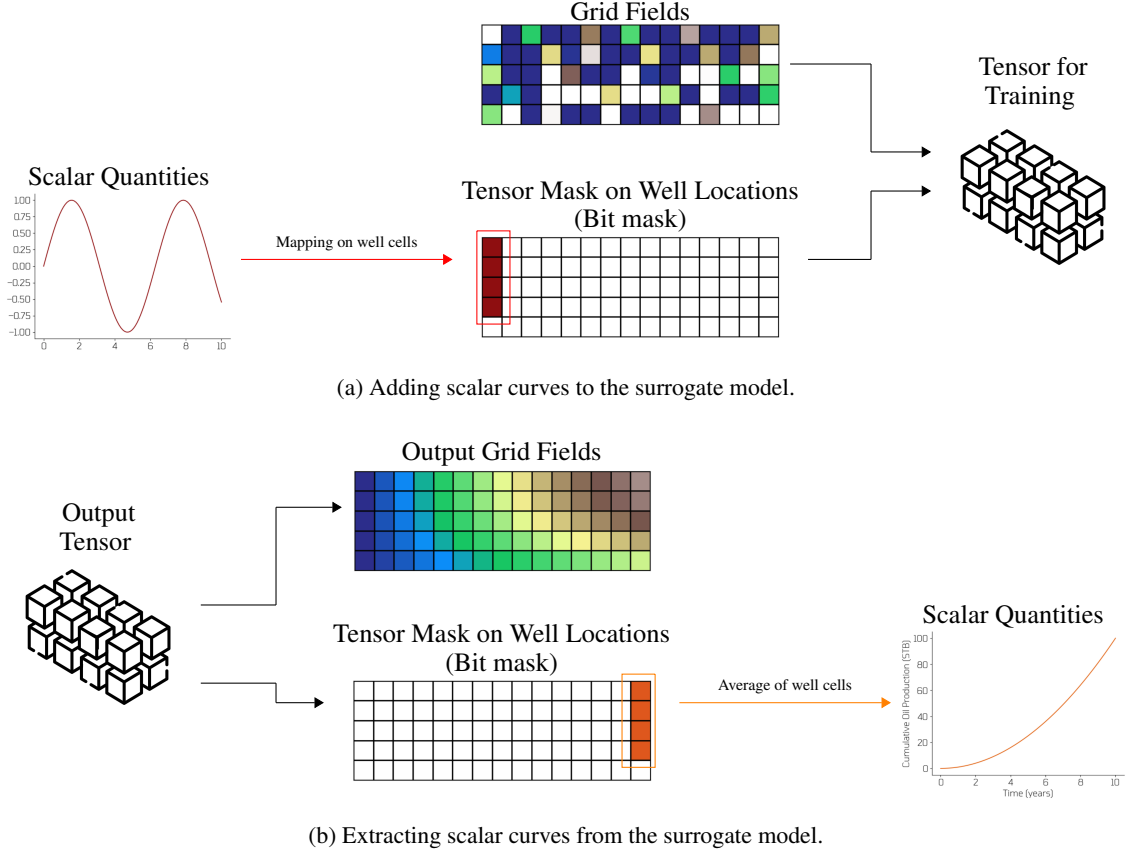

    \centering
     \include{figuras/Arquiteturas/BitMask_tikz}
    \caption{Illustrations of how scalar quantities are handled within the proposed models. Figure (a) shows the conversion of scalar quantities into tensors through bit masks around the wells, and Figure (b) shows the conversion of the output tensors into scalar quantities, which correspond to the average of the values of the cells within well locations for each time step. 
    }
    \label{fig:bit_mask}
\end{figure}

Now, we use the proposed methodology to predict $N_{c_{out}} = 1$ output channel, the gas saturation field, for a given gas injection rate value. For all hybrid configurations, the branch receives $N_{c_{in}} = 4$ channels as inputs: the $x$ and $z$ coordinates, permeability, and the tensorized mapping of the gas injection rate for that sample done through the process described in Fig. \ref{fig:bit_mask}. In this problem, the isotropic and heterogeneous permeability illustrated in Fig. \ref{fig:perm} is used in both the $x$ and $z$ directions, making it redundant to use both fields as different input channels in the model. Spatial dimensions for this problem are $N_x = 100$, $N_y = 1$ and $N_z = 20$. The trunk network is responsible for setting the time coordinates, which, in this problem, are defined as $N_t = 800$ time steps, with each $\Delta t = 10$ days. Both networks are responsible for mapping the inputs into the gas saturation field for all time steps.  For hybrid setups that use FNO in the branch network, a lifting layer of size $16$, $2$ Fourier layers with $4$ Fourier modes, and an MLP of two layers of size $16$ are used.
For MLP networks, the setup used includes $80$ input features, one hidden layer of size $16$, and $20$ output features. 
As for KANs, there are $4$ input features, $4$ layers of size $16$, $1$ output feature, and a cubic spline order. 
For all hybrid configurations, the trunk and branch networks use the Tanh and SiLU activation functions, respectively. Here, the AdamW optimizer is used with cosine learning rate decay, with an initial learning rate of $10^{-2}$ for $3,000$ epochs. These hyperparameters were defined after preliminary ablation studies that assessed the generalization quality of the hybrid schemes. 

\par 
Figure~\ref{fig:metrics_spe10m1} shows the evolution of the loss function (MSE) and the relative error $2-$norm during training for both the training (left subplots) and validation (right subplots) sets. Due to the oscillatory nature of these metrics throughout the learning process, the epoch-wise values are displayed as semi-transparent lines, while the bold curves correspond to moving averages computed over every $50$ epochs to highlight the overall convergence trend. We notice that the hybrid schemes with FNO in the branch network outperform the other models, consistently reducing errors throughout training and achieving the lowest values at the end. These trends are consistent with the final metrics reported in Table \ref{deeponet_all_hybridizations_metrics_case2}, which covers the MSE and relative $2-$norm computed over the test dataset for the gas saturation channel prediction after training the surrogate models. Both results confirm the superior performance of FNO-based hybridizations compared with hybrid schemes that used KAN and MLP on the branch network. We also assess the results for a given sample of the test set. Figures \ref{Val_spe10_dofk} to \ref{Val_spe10_dovanilla} show the ground truth, the predictions of each model, and the absolute pointwise error in space for the last time step predicted on one sample of the test set. For all tested combinations, spatial errors are more noticeable in the interface between oil and gas. As observed during training and overall metrics, the DeepONet (MLP) model represents the true field reasonably well, while the DeepONet (KAN) achieves better results than the DeepONet (MLP). Additionally, models that include FNO in DeepONet exhibit the lowest errors, both visually and quantitatively.


\begin{figure}
    \centering
    \includegraphics[width=0.95\linewidth]{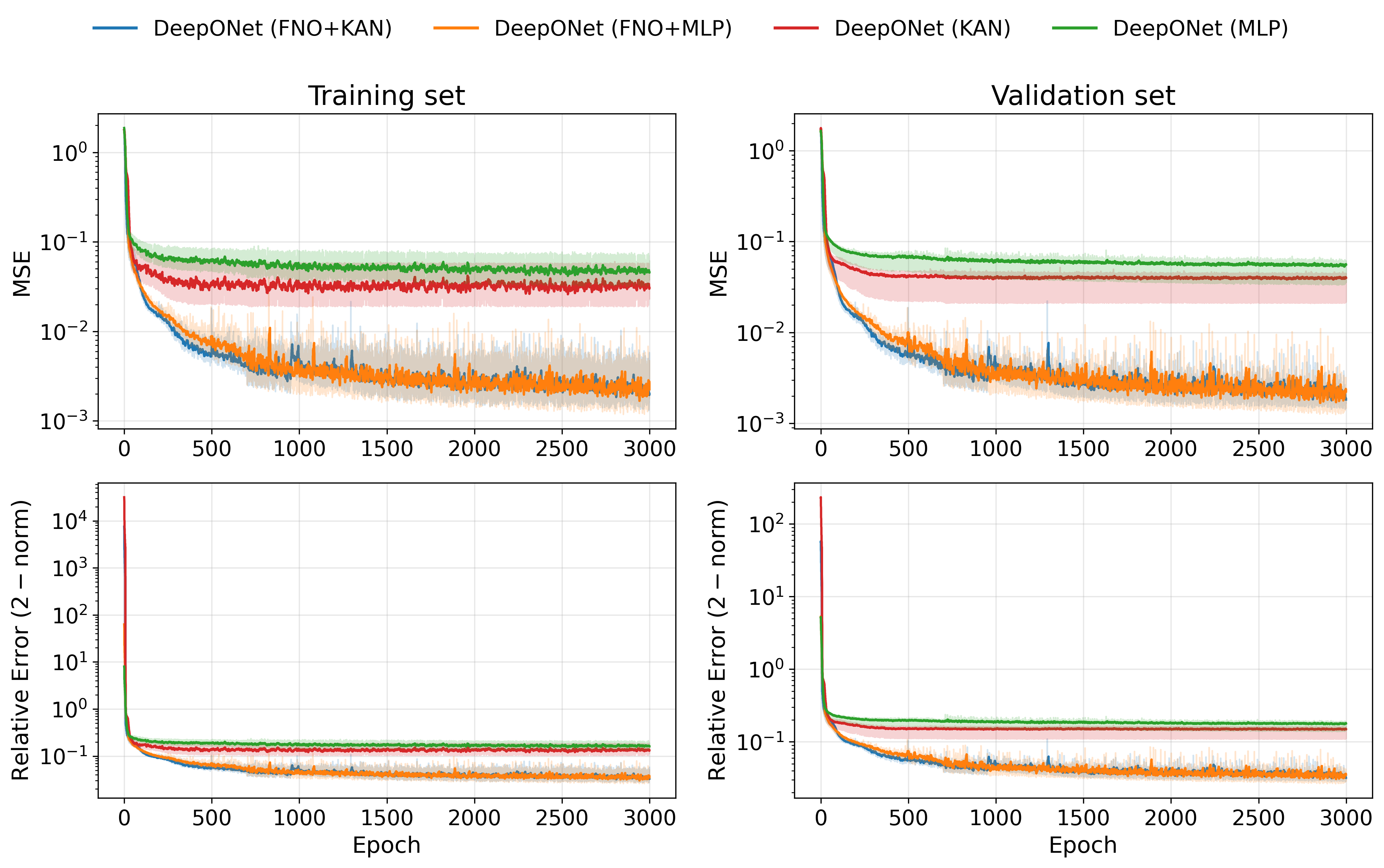}
    \caption{Evolution of the loss function (MSE) and the relative error $2-$norms during training for both the training (left) and validation (right) sets for the SPE10 Model 1 experiment. Due to the oscillatory behavior of these metrics throughout the learning process, the epoch-wise values are shown as semi-transparent lines, while the bold curves represent moving averages computed over every $50$ epochs to highlight the overall training trend.}
    \label{fig:metrics_spe10m1}
\end{figure}

\begin{figure}
    \centering
    \begin{subfigure}{1.0\textwidth}
        \includegraphics[width=\linewidth]{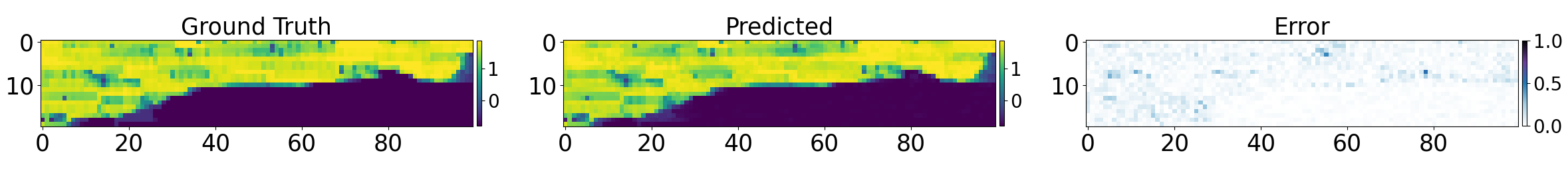}
        \caption{DeepONet (FNO+KAN).}
        \label{Val_spe10_dofk}
    \end{subfigure}
    \hfill
    \begin{subfigure}{1.0\textwidth}
        \includegraphics[width=\linewidth]{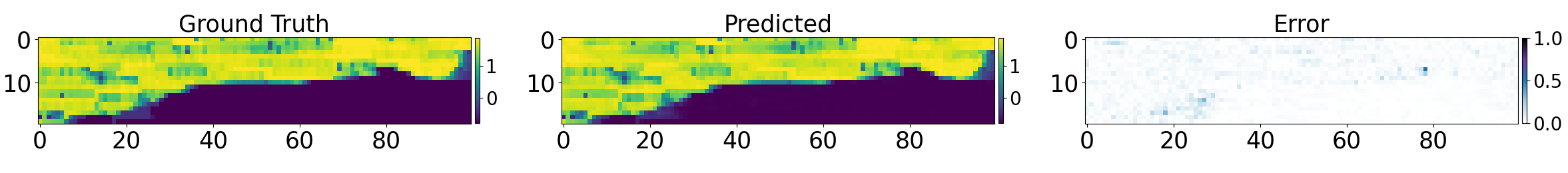}
        \caption{DeepONet (FNO+MLP).}
        \label{Val_spe10_dofmlp}
    \end{subfigure}
    \hfill
    \begin{subfigure}{1.0\textwidth}
        \includegraphics[width=\linewidth]{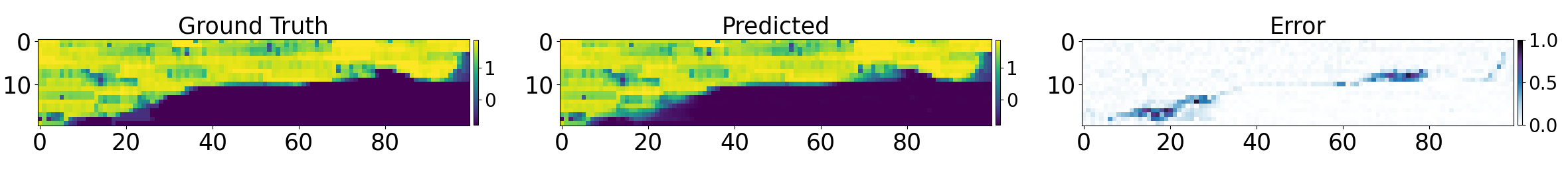}
        \caption{DeepONet (KAN).}
        \label{Val_spe10_dokan}
    \end{subfigure}
    \hfill
    \begin{subfigure}{1.0\textwidth}
        \includegraphics[width=\linewidth]{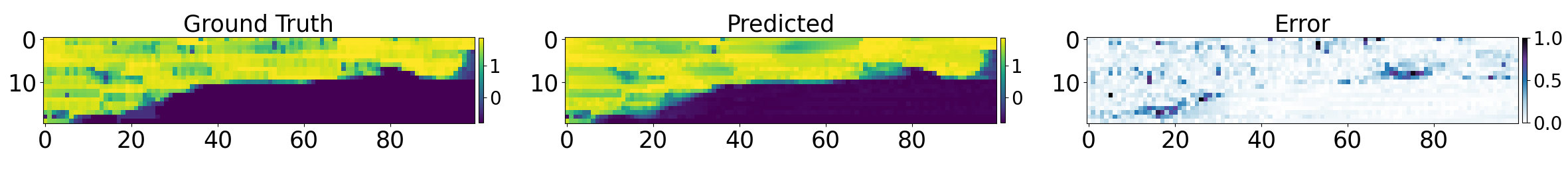}
        \caption{DeepONet (MLP).}
        \label{Val_spe10_dovanilla}
    \end{subfigure}
    \caption{Test of DeepONet models for gas saturation (SGAS) prediction over 8,000 days, at an injection rate of 0.3070 Mscf/d. Predictions are compared against simulation results at $t=8,000$ days. Left, ground truth; Center, predictions; Right, absolute pointwise error.}
    \label{spe10_val_case2}
\end{figure}

\begin{table}
\centering
\caption{Observed error values of the hybrid schemes for the predicted gas saturation on the SPE10M1 benchmark for the validation set.}
\label{deeponet_all_hybridizations_metrics_case2}
\begin{tabular}{lcc} 
\hline
\textbf{Architecture} & \textbf{MSE} & \textbf{Relative Error $2-$norm} \\
\hline 
DeepONet (FNO+KAN)     &  $1.43 \times 10^{-3}$  & $ 2.83 \times 10^{-2}$ \\
DeepONet (FNO+MLP)     &  $1.25 \times 10^{-3}$  & $ 2.65 \times 10^{-2}$ \\
DeepONet (KAN)         &  $2.09 \times 10^{-2}$  & $ 1.09 \times 10^{-1}$ \\
DeepONet (MLP)         &  $3.33 \times 10^{-2}$  & $ 1.38 \times 10^{-1}$ \\
\hline
\end{tabular}
\end{table}

\subsection{SPE10 Model 2} \label{SPE10M2Se}
The next benchmark, SPE10 Model 2, is substantially more complex than the previous one. SPE10 Model 2 is a two-phase (oil and water) model with simple geometry, no top structure, and no faults. The complexity of the problem, however, lies in the fine grid of the model and the complex permeability and porosity maps. The dimensions of the model are $1,200 \text{ft long} \times 2,200 \text{ft wide} \times 170 \text{ft}$ thick, on a $60 \times 220 \times 85$ mesh, yielding $1,122,000$ cells. The domain is initially fully saturated with oil, and water flooding starts at a fixed rate of $q_w = 5,000$ STB/day. Unlike the benchmark, we set the report steps to be defined every $10$ days, and the simulation is run for $1,000$ days. Figure \ref{fig:perm_spe10m2} shows two slices of the permeability field provided for the problem, as well as the injection and production wells' locations. For more details on the benchmark, see \cite{Christie2001}.

\par 
Similar to the SPE10 Model 1, we generate $25$ samples by varying the water injection rate of the injector well. Each value is sampled from 
\begin{equation}
    q_w = \text{Unif}(4000, 6000),
\end{equation}
being $q_w$ measured in STB/days. For each sample of $q_w$, a new simulation is performed for the SPE10 Model 2. Each parallel simulation takes around $2$ hours on $16$ cores of an Intel(R) Xeon(R) Gold 6430 processor.

\begin{figure}
    \centering
    \begin{subfigure}{0.49\textwidth}
        \includegraphics[width=\linewidth]{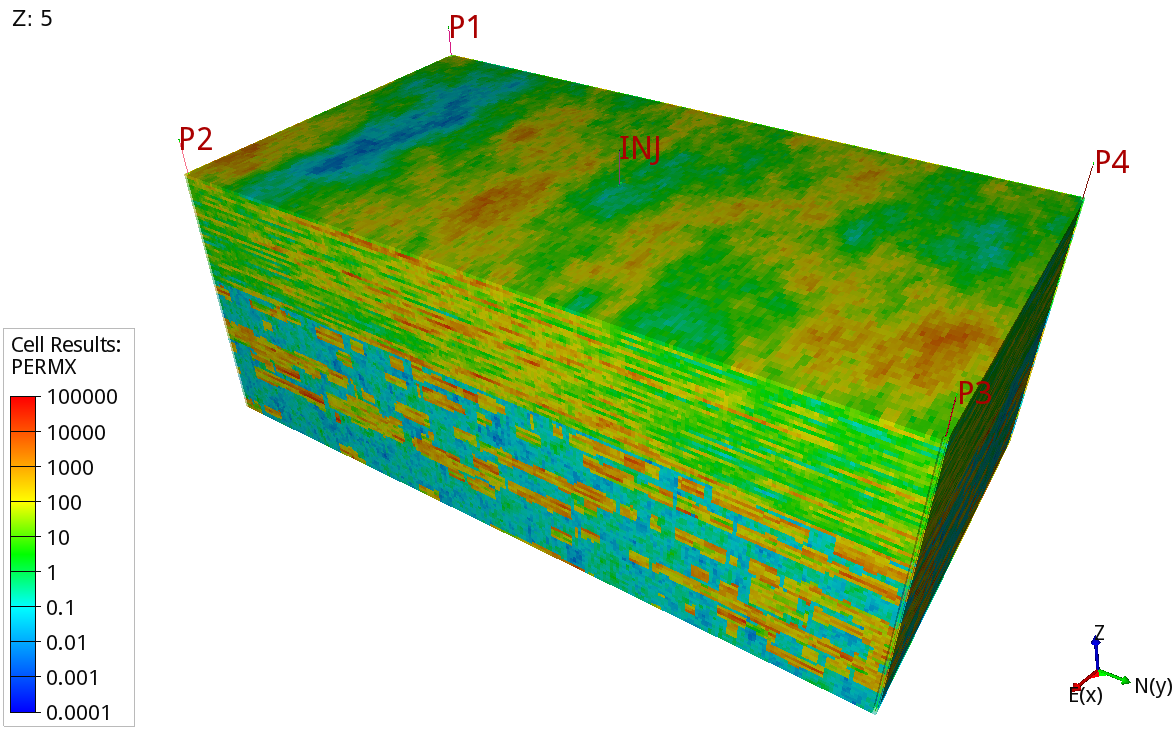}
        \caption{3D view of PERMX.}
    \end{subfigure}
    \hfill
    \begin{subfigure}{0.49\textwidth}
        \includegraphics[width=\linewidth]{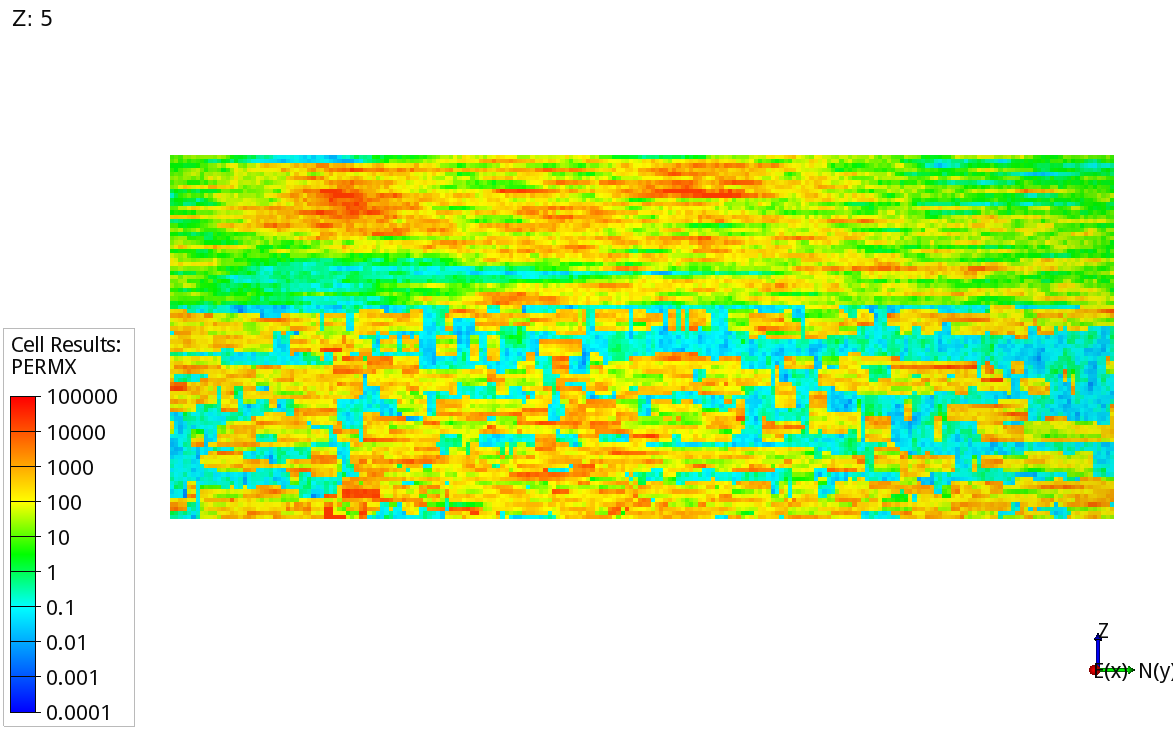}
        \caption{Central $y-z$ slice of PERMX.}
    \end{subfigure}
    \hfill
    \begin{subfigure}{0.49\textwidth}
        \includegraphics[width=\linewidth]{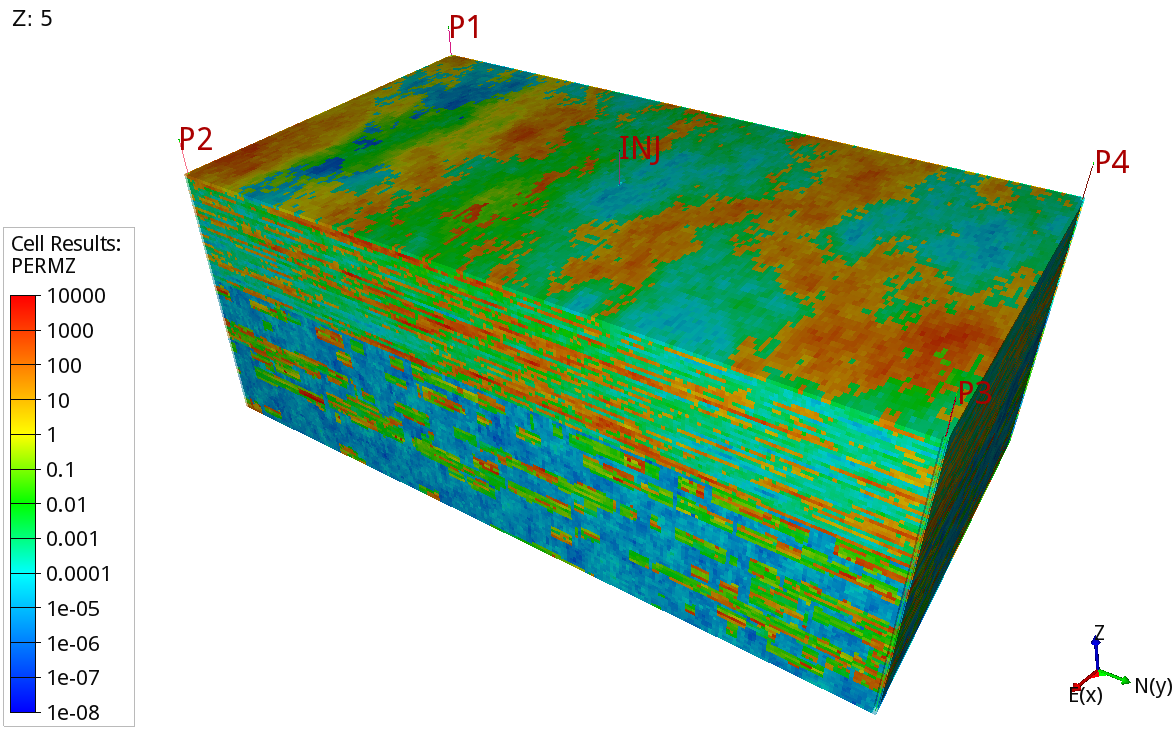}
        \caption{3D view of PERMZ.}
    \end{subfigure}
    \hfill
    \begin{subfigure}{0.49\textwidth}
        \includegraphics[width=\linewidth]{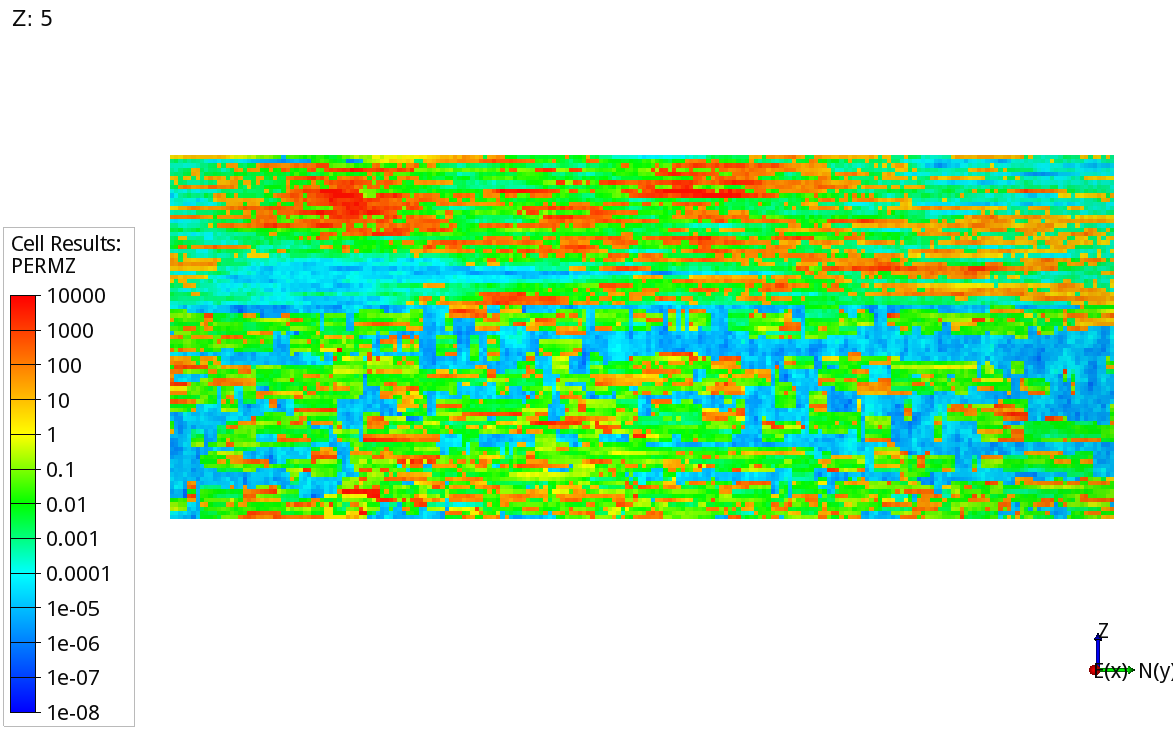}
        \caption{Central $y-z$ slice of PERMZ.}
    \end{subfigure}
    \caption{Correlated permeability field in the $x$ and $z$ direction for the SPE10 Model 2. Notice that, for this benchmark, the logarithmic scale reveals a permeability variation of eight orders of magnitude in the $x$ direction and twelve in the $z$ direction. The $z$-axis (vertical) has been exaggerated by a factor of $5$. Permeability in the $y$ direction is equal to the permeability in the $x$ direction.}
    \label{fig:perm_spe10m2}
\end{figure}

\par 
A surrogate model based on the proposed hybrid algorithm is built to predict the water and oil saturation fields at all times, as well as the oil production rate and water cut curves. For the branch network, inputs are the tensorized injection rate and static spatial fields (cell coordinates, permeability in the $x$ direction, and porosity). Again, for this problem, given that the permeability field is orthotropic, the values in directions $x$ and $y$ are identical; thus, it is redundant to add channels for permeability in all directions. For the $z$ direction, preliminary ablation studies revealed that adding a permeability channel in this direction would not affect the surrogate model accuracy. For the trunk network, time coordinates are defined. For memory reasons, instead of the available $100$ report steps, we use a smaller sample of time steps. In porous media flow simulations, it is customary to allocate a higher temporal resolution at the early stages of the process to capture the rapid dynamics emerging from phase interactions and shock development \cite{badawi2024neural, chandra2024fourier, jiang2024fourier}. Following the approach of Chandra et al. \cite{chandra2024fourier}, we employ logarithmic time sampling to select $34$ snapshots, wherein the time intervals ($\Delta t$) are smaller at the beginning of the simulation and progressively increase toward the end. Training is performed on $20$ simulations, and validation is performed on the remaining $5$ trajectories. Figure \ref{fig:swat_gt_spe10m2} shows the central $y-z$ slice of the water saturation at the last time step of the simulation for one sample of the validation set. For each hybrid scheme tested, we use three hyperparameter combinations that lead to DeepONet models with $127k$, $1M$, and $20M$ parameters, named cases A, B, and C, respectively. Appendix \ref{sec:app_c} shows in more detail how each hybrid architecture is built and how the learnable parameters are distributed between the branch and trunk networks. Our goal is to train the SPE10 Model 2 using a single NVIDIA H100 GPU with $94$ GB of VRAM. Given that the physical model exhibits high-fidelity data with substantial dimensionality, memory consumption is expected to pose a significant challenge. Under these configurations, our objective is to evaluate each hybrid architecture's ability to accommodate an increasing number of trainable parameters efficiently.

\begin{figure}
    \centering
    \includegraphics[width=0.69\linewidth]{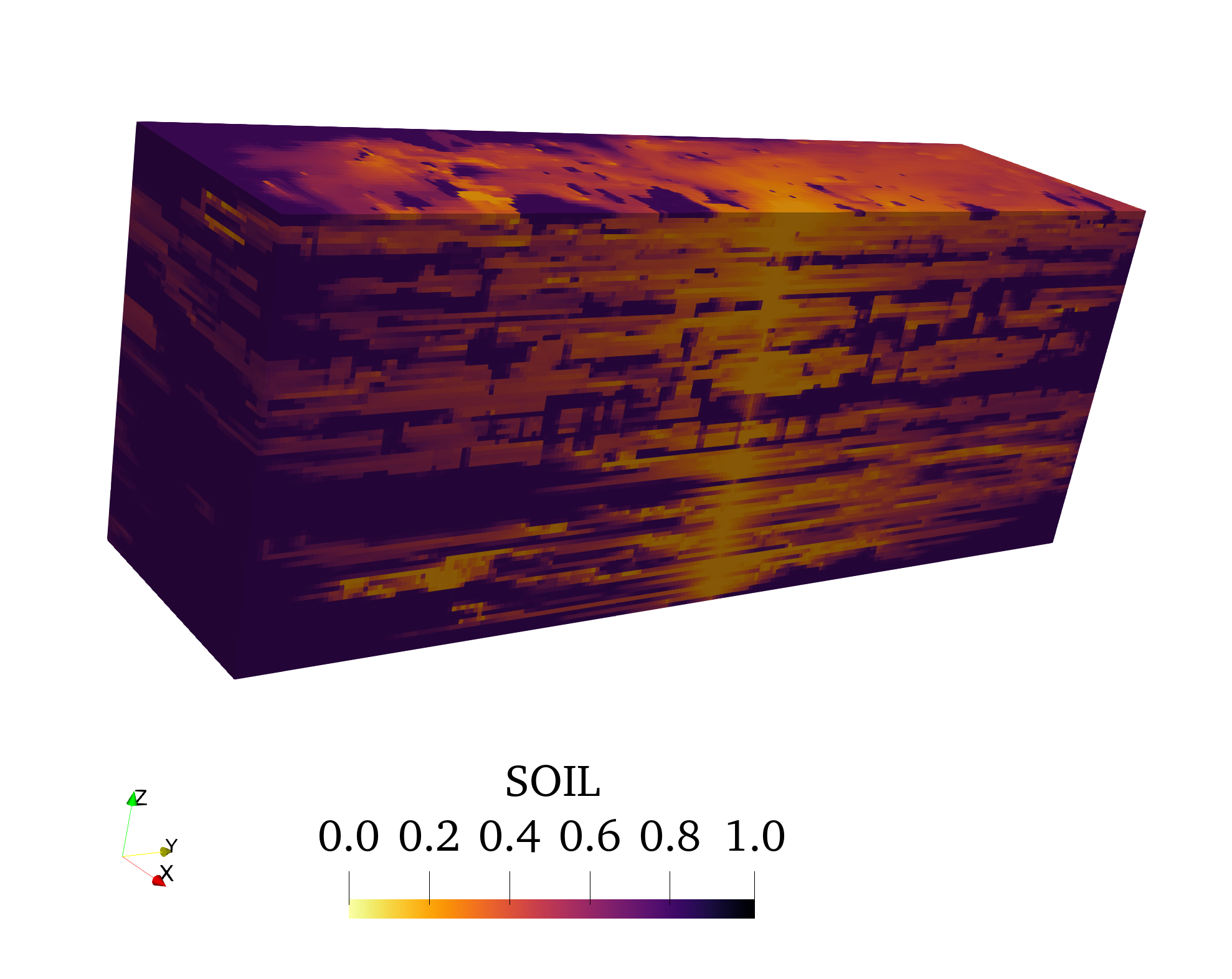}
    \caption{Central slice of the $y-z$ plane of the oil saturation field for the SPE10 Model 2 benchmark. Ground truth solution at $t=1,000$ days. SOIL stands for oil saturation, and the $z$-axis (vertical) has been exaggerated by a factor of $5$.}
    \label{fig:swat_gt_spe10m2}
\end{figure}

\par 
The surrogate models for this example predict four output channels: the water (SWAT) and oil (SOIL) saturation fields, the oil production rates (WOPR), and the well water cuts (WWCT) at each time step, given a certain injection rate. In this study, we refer to water cut as the fraction of water in the total produced fluids at a given well, expressed as the ratio of produced water to total production rate (water and oil). The output channels are very different in nature: the water and oil saturation fields are primary variables obtained in the black-oil model described in Equations (\ref{eq:blackoil} - \ref{eq:closure}) while the oil production rate and well water cut curves are post-processed quantities trained and predicted through the tensorization process described in Figure \ref{fig:bit_mask}. Regarding model training, in Case A, where the models comprise approximately $127$k parameters, all four configurations were successfully trained. In Case B, involving models with around $1$M parameters, the DeepONet (KAN) formulation exceeded the available GPU memory, whereas the DeepONet (MLP) model utilized nearly the entire GPU capacity and still completed training successfully. Finally, for Case C, comprising models with approximately $20$M parameters, training was feasible only for the hybrid configurations that employed the FNO in the branch network. Table \ref{tab:deeponet_times} shows the total time required to train each hybrid scheme. We notice that for Cases A and B, the difference in training time is small, indicating that the dominant computational effort when training $127$k and $1$M parameters networks is data size (the SPE10 Model 2 dimensions are $60 \times 220 \times 85$ in space, $34$ time steps in time with $4$ channels per sample, leading to $152.6$ million entries per sample). In this case, the difference between the learnable parameter count and the model size is of several orders of magnitude. When analyzing Case C, where the hybrid schemes have $20$M parameters, this difference is of one order of magnitude, indicating that the model size now interferes with training time.

\begin{table}
\centering
\setlength{\tabcolsep}{12pt}
\renewcommand{\arraystretch}{1.15}
\caption{Total training times for all cases and all tested hybrid configurations. Inference times for all cases are negligible.}
\begin{tabular}{@{} l l r @{}} 
\toprule
\textbf{Case} & \textbf{Model} & \textbf{Total time (h)} \\
\midrule
A & DeepONet (MLP)        & 21.63 \\
A & DeepONet (KAN)        & 24.85 \\
A & DeepONet (FNO+KAN)    & 20.85 \\
A & DeepONet (FNO+MLP)    & 25.95 \\[2pt]
B & DeepONet (MLP)        & 26.60 \\
B & DeepONet (FNO+KAN)    & 22.65 \\
B & DeepONet (FNO+MLP)    & 22.04 \\[2pt]
C & DeepONet (FNO+KAN)    & 28.52 \\
C & DeepONet (FNO+MLP)    & 36.72 \\
\bottomrule
\end{tabular}
\label{tab:deeponet_times}
\end{table}

\par 
In order to verify the model's ability to predict spatio-temporal fields and post-processed quantities, we assess them separately. First, we assess the generalization capability of the trained models by computing the relative $2$-norm errors between the ground truth and the predicted tensors across all spatial and temporal dimensions for each independent channel on the validation set. Figures \ref{fig:val_rel_error_p1}, \ref{fig:val_rel_error_p2}, and \ref{fig:val_rel_error_p3} illustrate the evolution of the relative error $2$-norms on the validation set for each output quantity across training epochs for Cases A, B, and C. The shaded regions represent the actual relative errors computed at each epoch. To enhance readability, a trend line is superimposed, corresponding to the moving average over the last $15$ epochs, highlighting the central tendency of the training dynamics. We notice that for grid field channels, as shown in the top plots of Figures \ref{fig:val_rel_error_p1}, \ref{fig:val_rel_error_p2}, and \ref{fig:val_rel_error_p3} for Cases A, B, and C, respectively, the curves are practically identical. This is expected, since the fields are highly correlated: the SPE10 Model 2 benchmark is a two-phase problem in which $s_w + s_o = 1$ for each grid cell. To avoid redundancy, we show results for the oil saturation channel. However, when analyzing both curves as we increase the number of parameters across all hybrid combinations, we observe that greater model complexity yields better performance. Relative errors for grid channels for cases A, B, and C are seen in Table \ref{tab:rel_error_spe10m2_grid}. We observe that the relative error decreases as the number of learnable parameters increases for hybrid configurations that incorporate FNO models into their branch networks, suggesting enhanced generalization. In contrast, for the DeepONet (MLP) configuration, the relative $2$-norm error on the validation set in Case B increases, possibly indicating overfitting in this scenario. This observation is corroborated by Fig. \ref{fig:hist_spe10m2}, which presents the absolute error histograms for all tested configurations and cases. For the DeepONet (MLP) model, the histogram exhibits a noticeable rightward shift in the absolute error distribution in Case B. Conversely, for the hybrid models DeepONet (FNO+KAN) and DeepONet (FNO+MLP), increasing the number of learnable parameters drives the distribution mean closer to zero and reduces the spatial error variance. Still in terms of spatial error, Figs \ref{fig:results_spe10m2_volumetric} and \ref{fig:results_spe10m2_slice} show the last time step of each model's prediction for the oil saturation field as well as the absolute error for the whole domain and for a central $y-z$ slice. A close inspection of the figures, especially in Figure~\ref{fig:results_spe10m2_volumetric}, reveals that the highest spatial prediction errors are localized within a few grid cells, reaching magnitudes of approximately $0.5$. These localized discrepancies align with the upper tail of the error distribution observed in the histogram of Figure~\ref{fig:hist_spe10m2}

\par 
For the saturation predictions, we also assess key quantities of interest in developing surrogate models for reservoir engineering applications. Given the relation $s_w + s_o = 1$, the water and oil saturation fields predicted by the hybrid schemes should satisfy this equality to ensure physically consistent solutions. In SciML, one of the primary strategies to enforce such constraints is to include physics-informed terms in the loss function \cite{raissi2019physics}, as is commonly adopted in reservoir engineering studies \cite{badawi2024neural}. In the present work, however, the models are purely data-driven, without any explicit enforcement of conservation laws or PDE-based constraints. Phase balance is implicitly maintained through the closure equations of the black-oil model evaluated by OPM Flow. Figure~\ref{fig:results_spe10m2_mass_2} presents the temporal evolution of the domain-integrated saturation for each phase, as well as their sum. This quantity is computed as $\dfrac{1}{| \Omega}| \int_{\Omega}s_\alpha d\Omega$, where $\alpha = \{w, o, w+o\}$, where $w+o$ is the sum of the cell's oil and water saturations. We observe that all hybrid schemes, depicted as dashed lines with markers, closely follow the ground-truth curves across all cases and model configurations. Together with the similarity observed in the relative error $2$-norms of the SOIL and SWAT channels in Figs.~\ref{fig:val_rel_error_p1}, \ref{fig:val_rel_error_p2}, and \ref{fig:val_rel_error_p3}, these results highlight the hybrid models’ ability to learn the strong correlation between $s_o$ and $s_w$, thereby producing physically consistent and mass-conservative predictions despite the absence of explicit physics constraints.

\begin{table}
    \centering
    \caption{Relative errors for the oil saturation (SOIL) channel for all schemes and cases. DeepONet (KAN) for cases B  and C, and DeepONet (MLP) for case C, did not fit in the NVIDIA H100 VRAM.}
    \vspace{4pt}
    \begin{tabular}{c c c}
    \toprule
         \textbf{Hybrid Scheme} & \textbf{Case} &\textbf{ Relative Error}  \\
          &  & \textbf{for the SOIL channel}  \\
         \midrule
         DeepONet (FNO+KAN) & A & $2.14 \times 10^{-1}$ \\
                          & B & $1.55 \times 10^{-1}$ \\
                          & C & $6.75 \times 10^{-2}$  \\
         \midrule
         DeepONet (FNO+MLP) & A & $2.13 \times 10^{-1}$ \\
                          & B & $1.51 \times 10^{-1}$ \\
                          & C & $6.17 \times 10^{-2}$ \\
         \midrule
         DeepONet (KAN)     & A & $1.61 \times 10^{-1}$ \\
                          & B &  - \\
                          & C & - \\
         \midrule
         DeepONet (MLP) & A & $2.89 \times 10^{-1}$ \\
                          & B & $2.51 \times 10^{-1}$ \\
                          & C & -\\
        \bottomrule
    \end{tabular}
    \label{tab:rel_error_spe10m2_grid}
\end{table}

\begin{figure}
    \centering
    \includegraphics[width=0.95\linewidth]{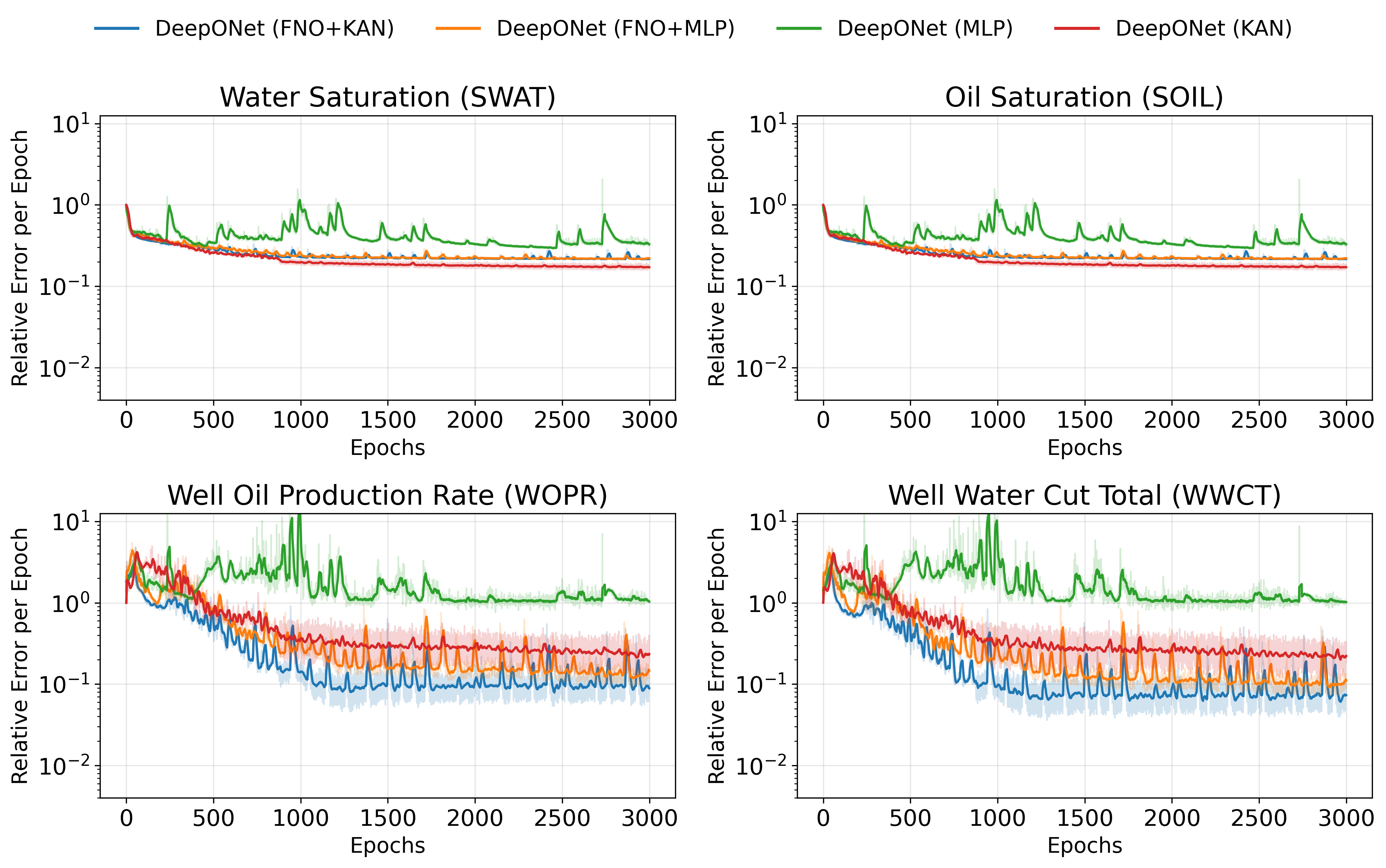}
    \caption{Evolution of the relative error $2$-norm on the validation set for the SPE10 Model 2 using surrogate models with approximately $127$k parameters. The semi-transparent regions represent epoch-wise relative errors, while the solid lines correspond to the moving averages computed over the last $15$ epochs, illustrating the overall training trend.}
    \label{fig:val_rel_error_p1}
\end{figure}

\begin{figure}
    \centering
    \includegraphics[width=0.95\linewidth]{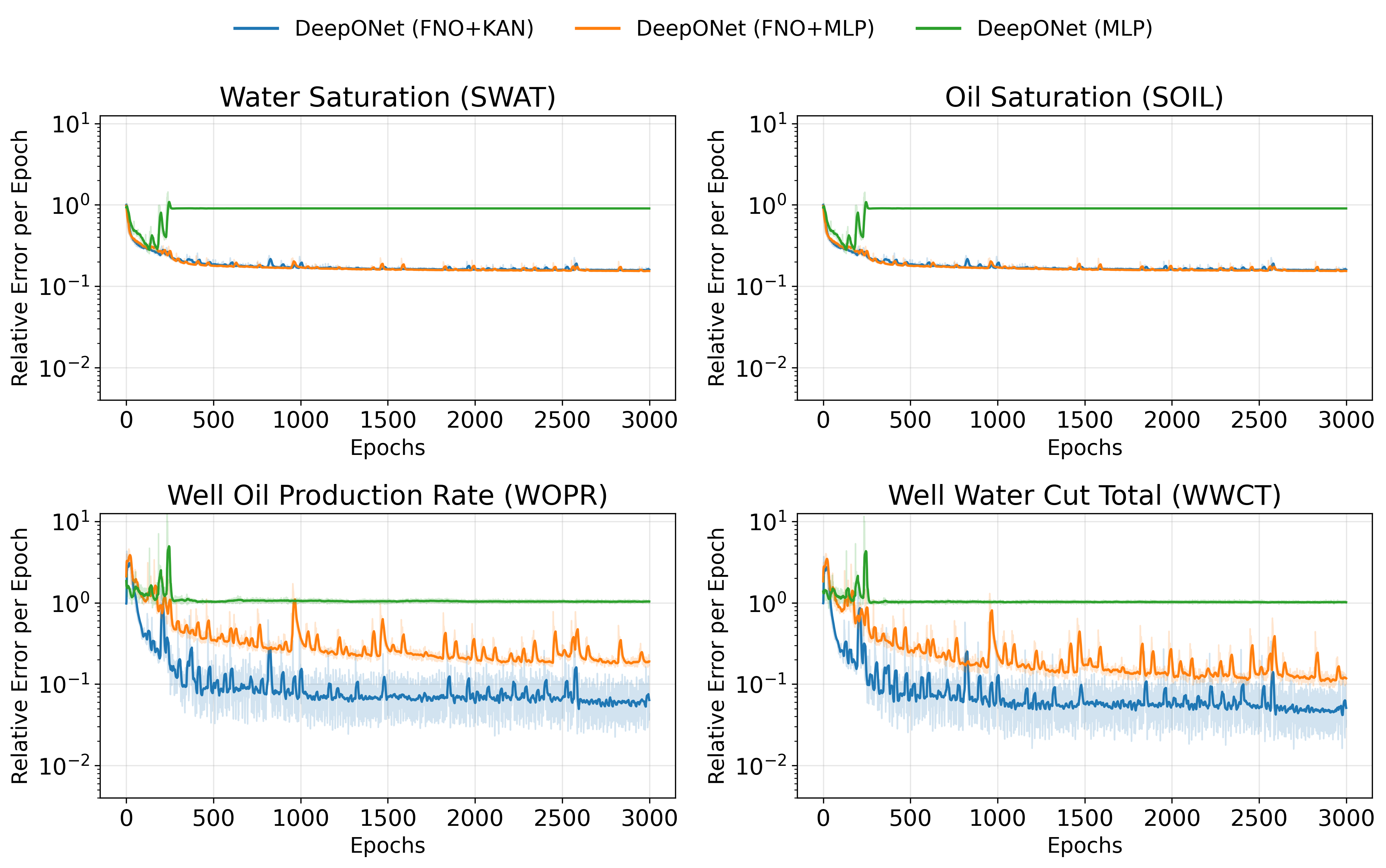}
    \caption{Evolution of the relative error $2$-norm on the validation set for the SPE10 Model 2 using surrogate models with approximately $1$M parameters. The semi-transparent regions represent epoch-wise relative errors, while the solid lines correspond to the moving averages computed over the last $15$ epochs, illustrating the overall training trend.}
    \label{fig:val_rel_error_p2}
\end{figure}

\begin{figure}
    \centering
    \includegraphics[width=0.95\linewidth]{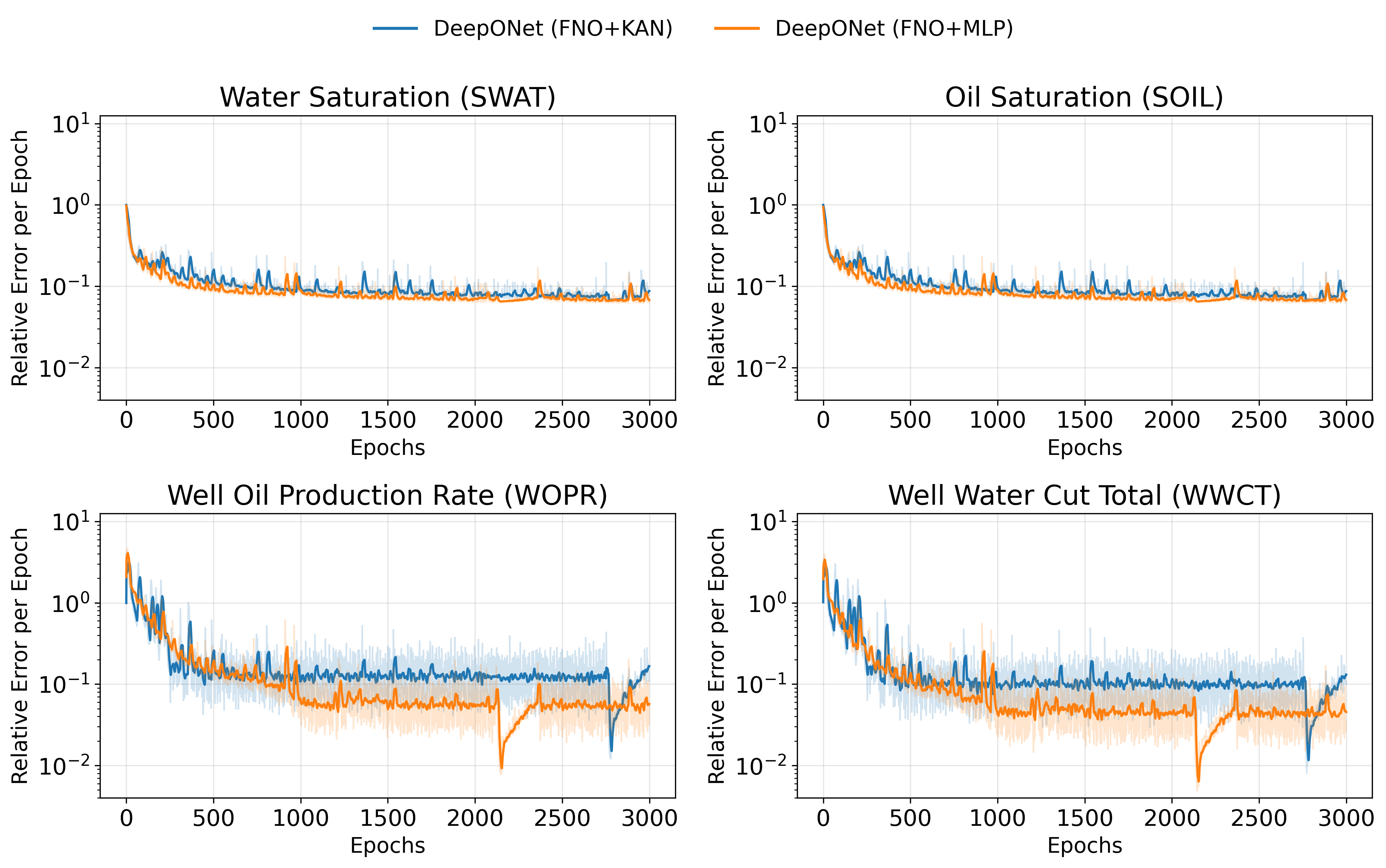}
    \caption{Evolution of the relative error $2$-norm on the validation set for the SPE10 Model 2 using surrogate models with approximately $20$M parameters. The semi-transparent regions represent epoch-wise relative error, while the solid lines correspond to the moving averages computed over the last $15$ epochs, illustrating the overall training trend.}
    \label{fig:val_rel_error_p3}
\end{figure}

\begin{figure}
    \centering
    \includegraphics[width=0.95\linewidth]{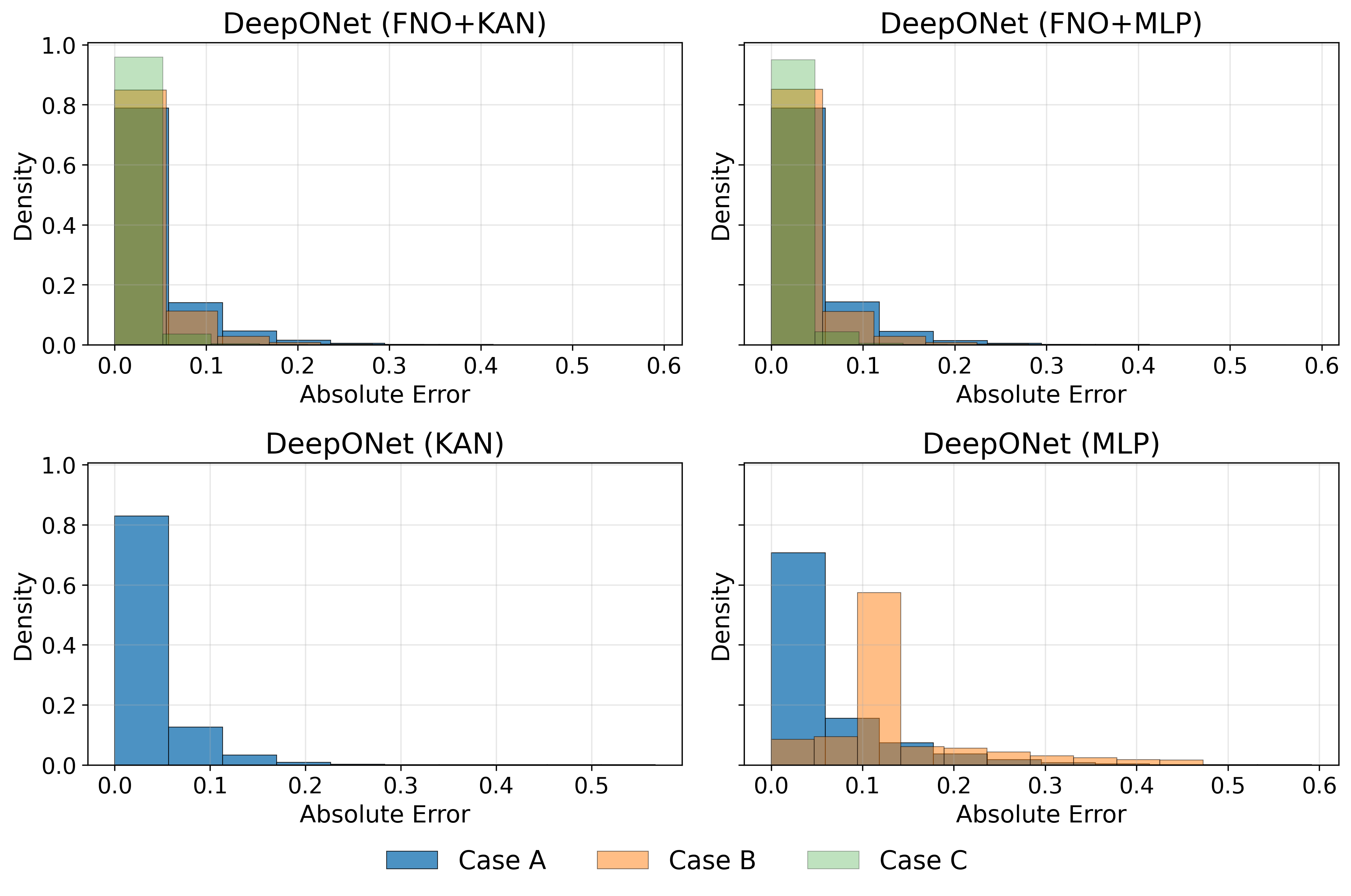}
    \caption{Histogram of the absolute errors in the oil saturation field at the final time step for one validation sample. Increasing the model capacity in the hybrid configurations where FNO was employed in the network branch reduces the spread of the error distribution. It shifts its mean toward zero, indicating improved generalization. This trend is not observed in the DeepONet (MLP) case, which suggests overfitting.}
    \label{fig:hist_spe10m2}
\end{figure}

\begin{figure}
    \centering

    \begin{subfigure}{0.6\textwidth}
        \centering
        \includegraphics[width=\linewidth, trim=0 0 0 240, clip]{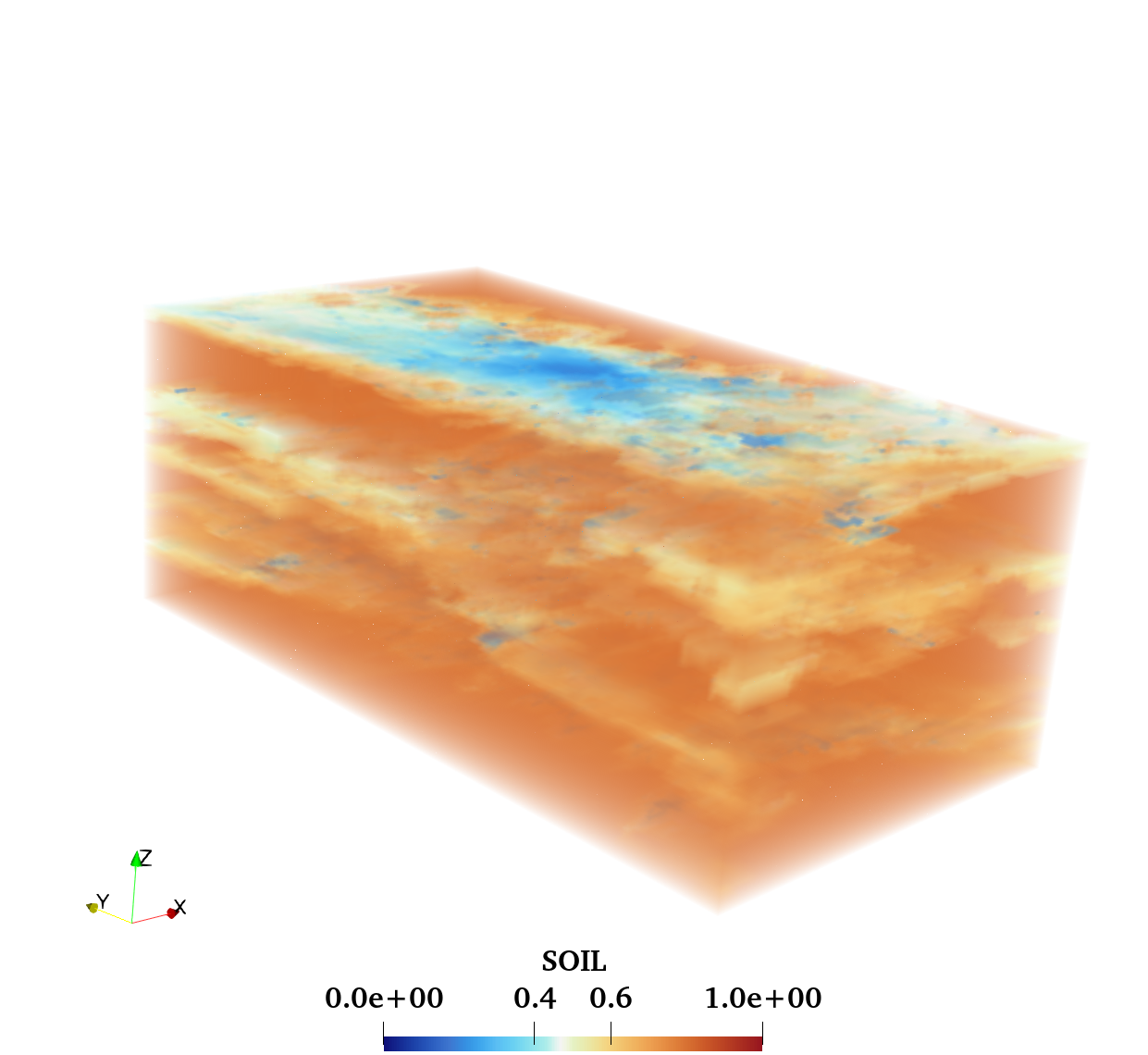}
        \caption{Ground truth}
    \end{subfigure}

    \vspace{.1em} 

    \begin{subfigure}{0.48\textwidth}
        \centering
        \includegraphics[width=\linewidth, trim=0 0 0 240, clip]{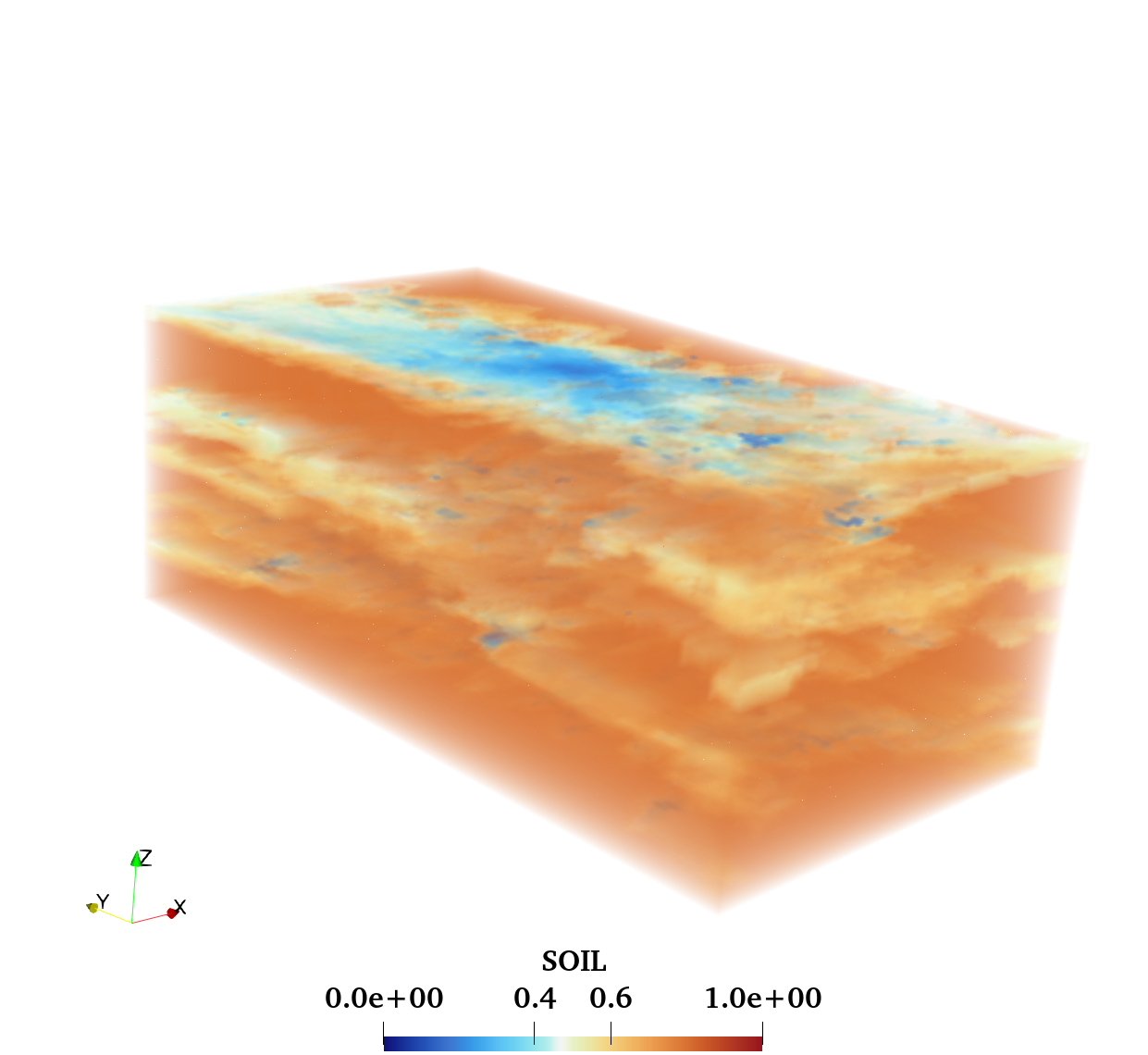}
        \caption{Prediction - DeepONet (FNO+KAN).}
    \end{subfigure}
    \hfill
    \begin{subfigure}{0.48\textwidth}
        \centering
        \includegraphics[width=\linewidth, trim=0 0 0 240, clip]{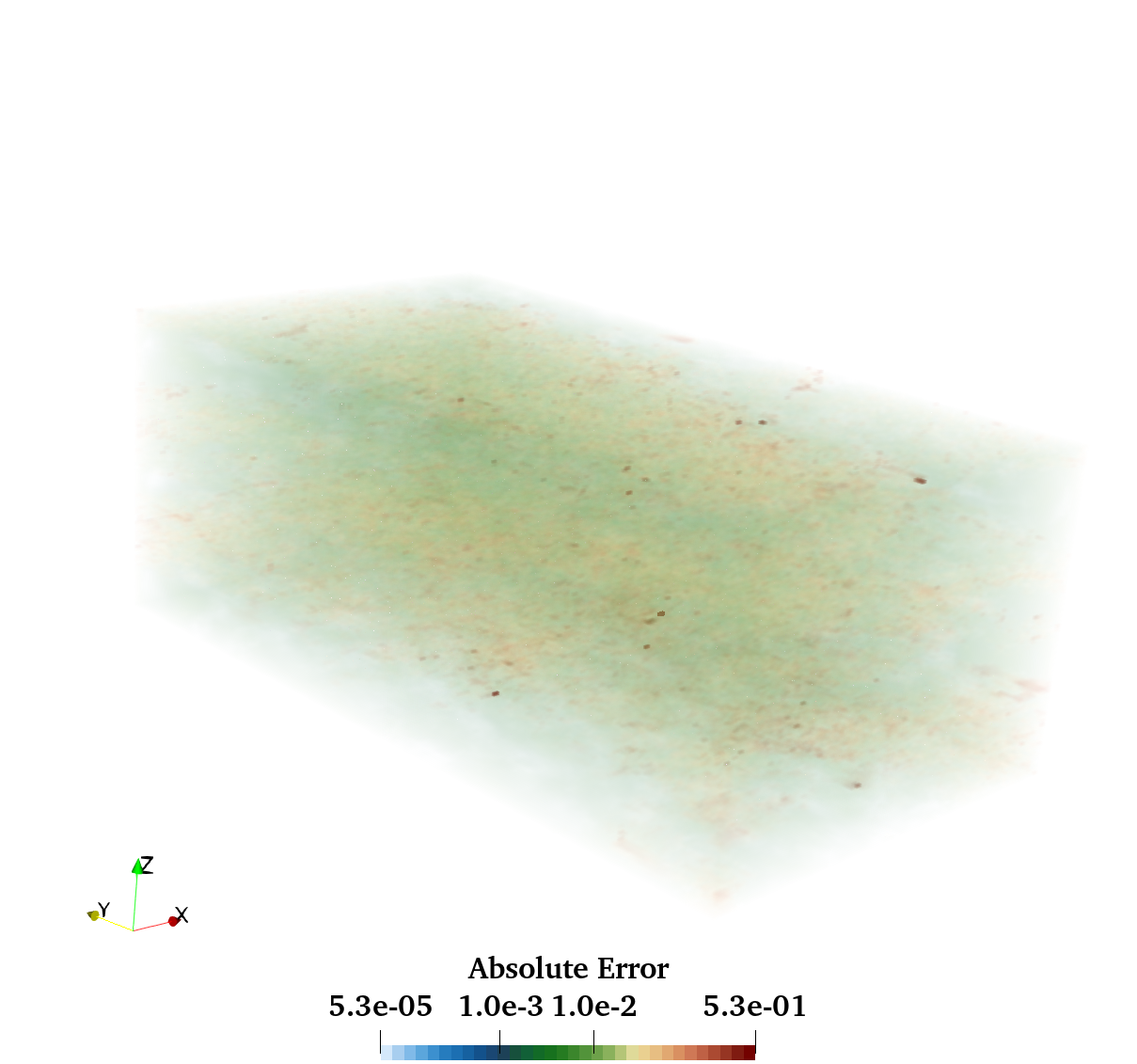}
        \caption{Error - DeepONet (FNO+KAN).}
    \end{subfigure}


    \begin{subfigure}{0.48\textwidth}
        \centering
        \includegraphics[width=\linewidth, trim=0 0 0 240, clip]{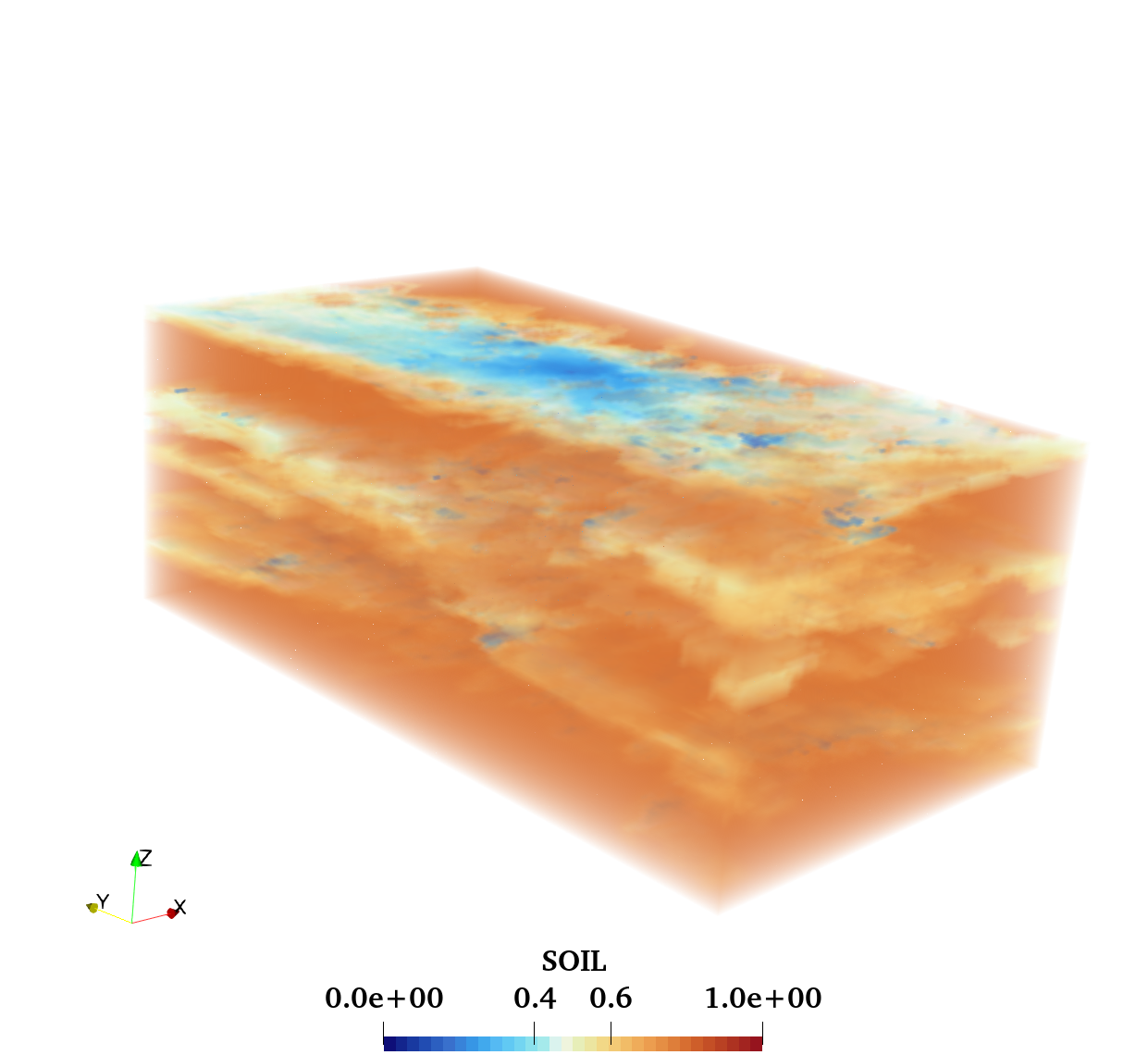}
        \caption{Prediction - DeepONet (FNO+MLP).}
    \end{subfigure}
    \hfill
    \begin{subfigure}{0.48\textwidth}
        \centering
        \includegraphics[width=\linewidth, trim=0 0 0 240, clip]{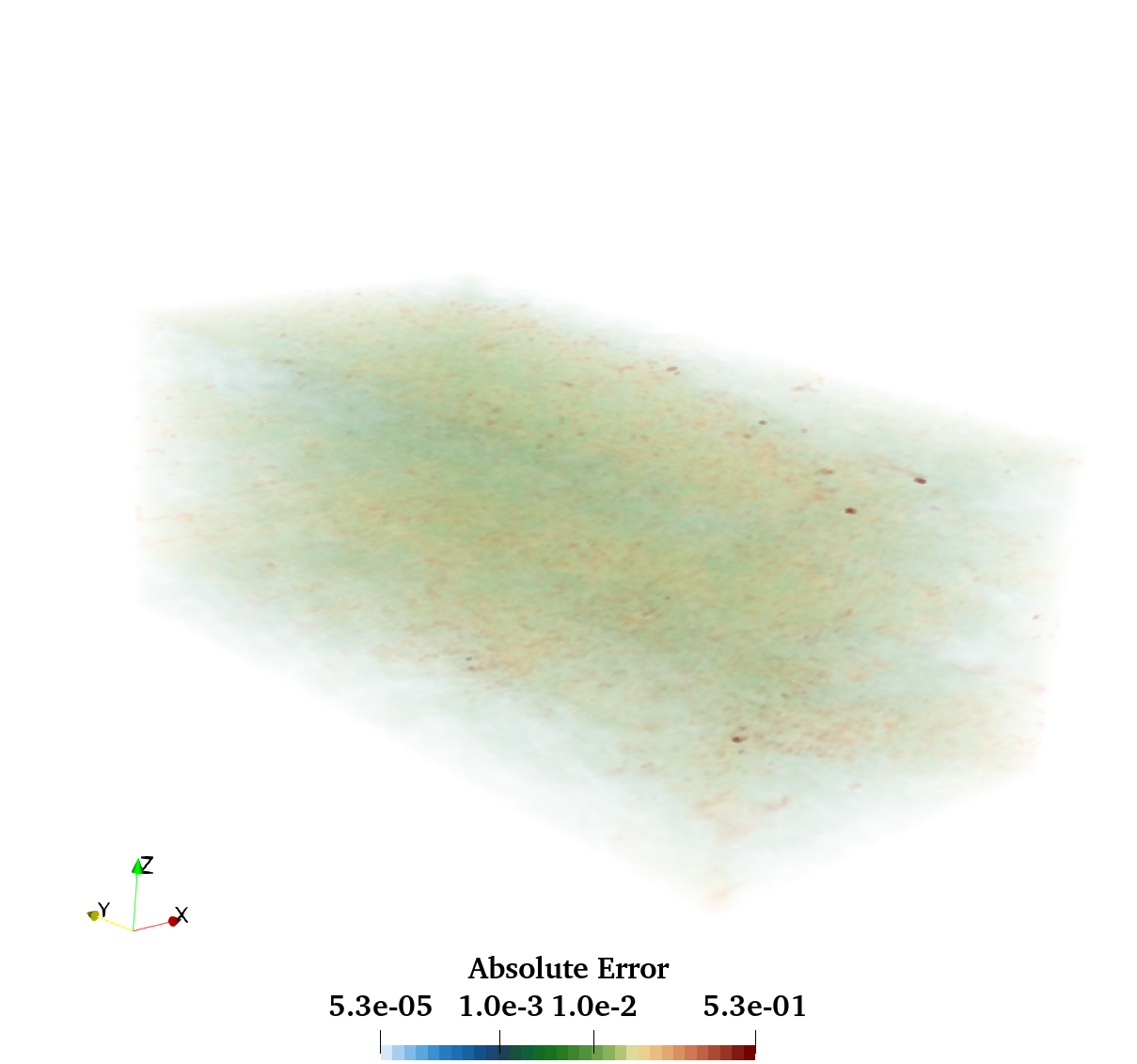}
        \caption{Error - DeepONet (FNO+MLP).}
    \end{subfigure}

    \caption{Prediction and absolute error for the proposed hybrid models at $t = 1{,}000$ days. SOIL stands for oil saturation.
    The $z$-axis (vertical) has been exaggerated by a factor of 5. We add transparency to the smaller error values in (c) and (e) so that the absolute error is visible in the interior.}
    \label{fig:results_spe10m2_volumetric}
\end{figure}


    
    

\begin{figure}
    \begin{subfigure}{0.99\textwidth}
    \includegraphics[width=\linewidth]{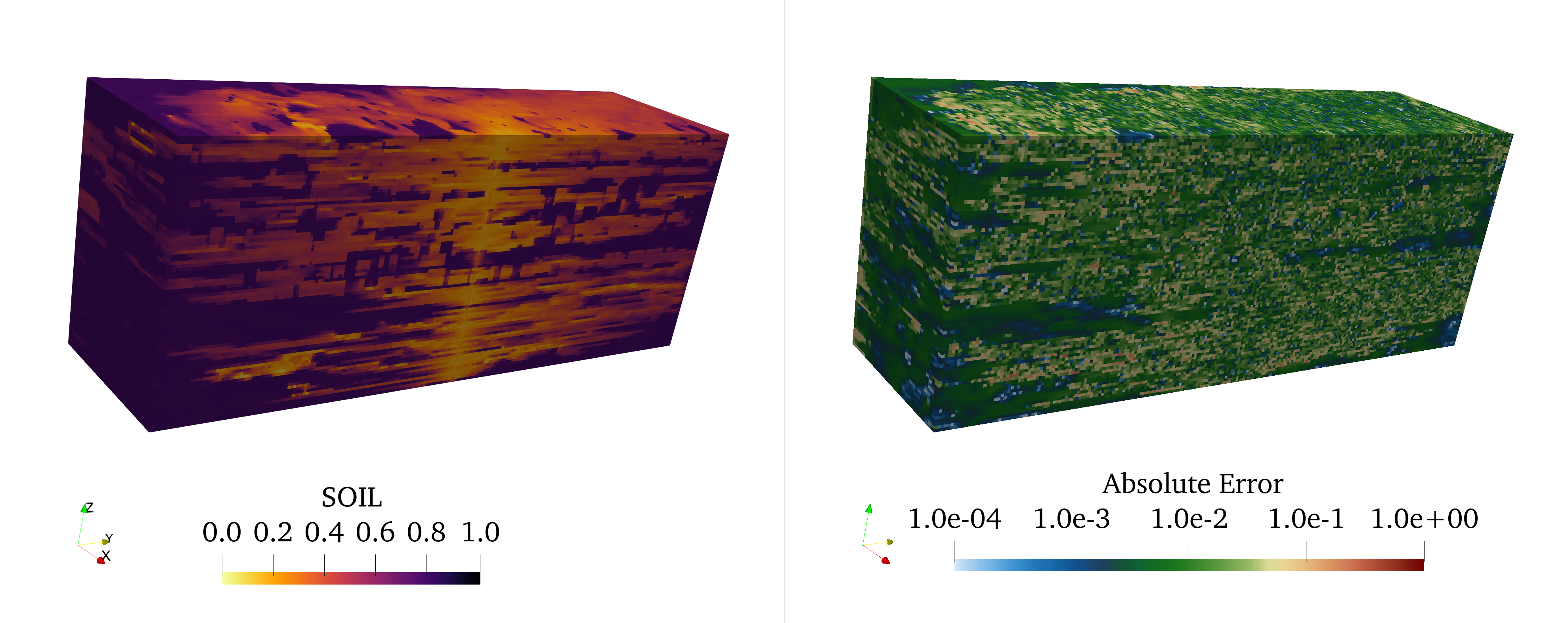}
        \caption{DeepONet (FNO+KAN) for Case C.}
    \end{subfigure}
    \hfill
    \begin{subfigure}{0.99\textwidth}
    \includegraphics[width=\linewidth]{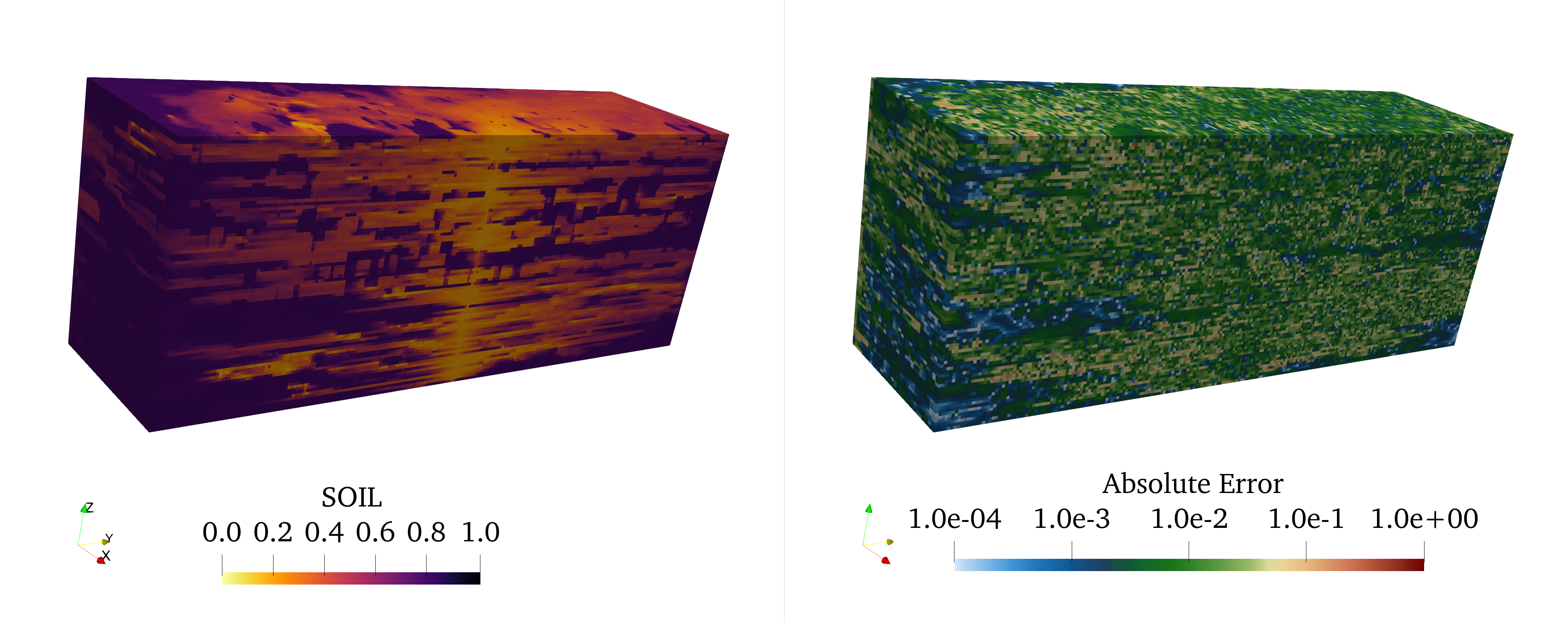}
        \caption{DeepONet (FNO+MLP) for Case C.}
    \end{subfigure}
    \centering
    \caption{Prediction and absolute error for the proposed hybrid models at $t=1,000$ days. SOIL stands for oil saturation. The $z$-axis (vertical) has been exaggerated by a factor of $5$.}
    \label{fig:results_spe10m2_slice}
\end{figure}


\begin{figure}
    \centering
    \includegraphics[width=\linewidth]{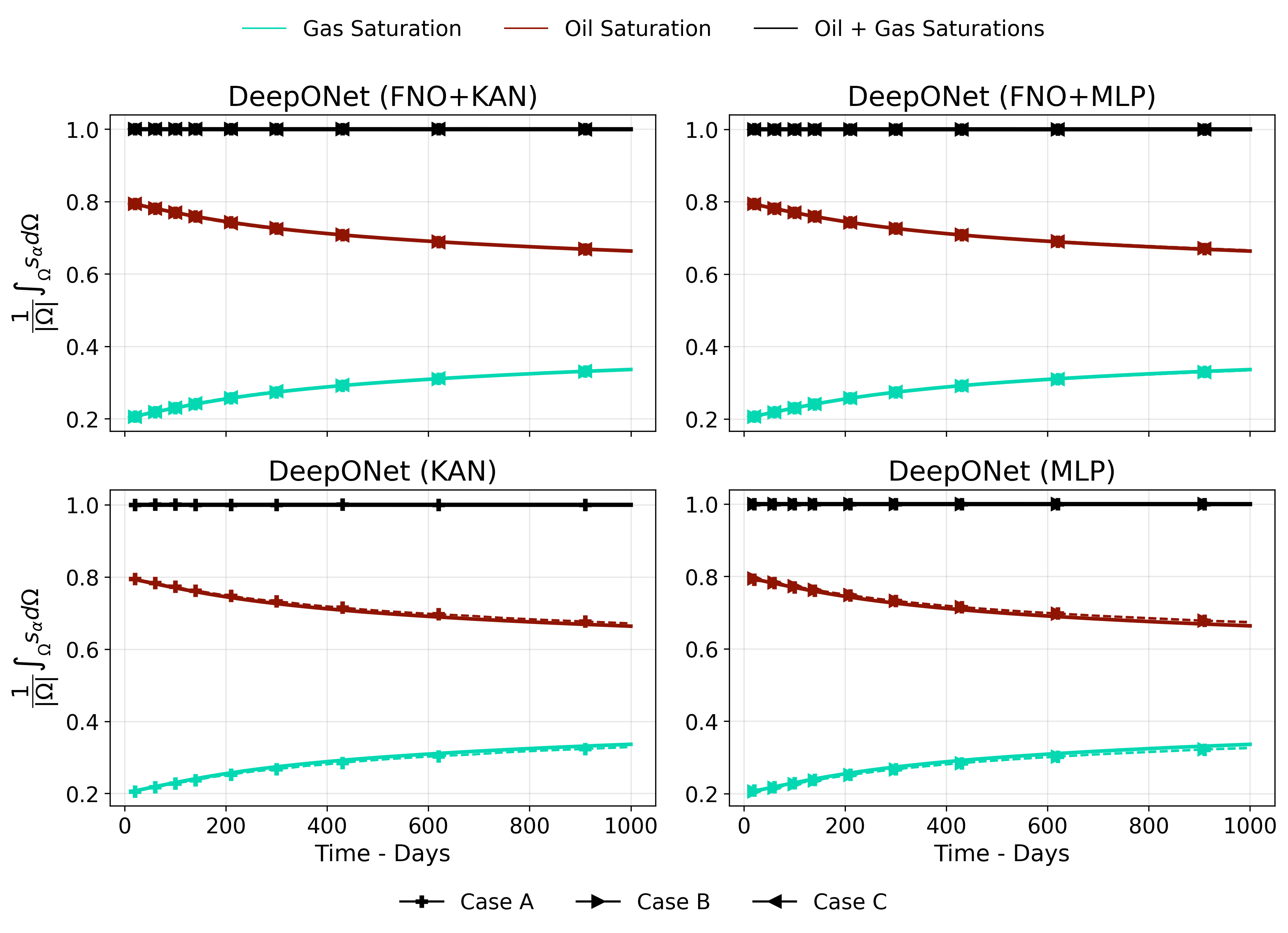}
    \caption{Phase balances conservation assessment for all hybrid configurations using the SPE10 Model 2 dataset. Dashed lines denote the surrogate model predictions, while solid bold lines represent the ground-truth results obtained from OPM Flow. Values are assessed by integrating each saturation field at each predicted time step and normalizing by the domain volume.}
    \label{fig:results_spe10m2_mass_2}
\end{figure}

\par
Regarding the well oil production rate and water cut predictions, the curves shown in Figs.~\ref{fig:val_rel_error_p1}, \ref{fig:val_rel_error_p2}, and \ref{fig:val_rel_error_p3} display larger oscillations during training, represented by the semi-transparent traces. This behavior arises from the sparsity of the ground-truth fields, which contain nonzero values only at the well locations. Consequently, any spurious nonzero predictions outside these cells can produce relatively large deviations when computing the relative error in the $2$-norm. Across all configurations, hybrid schemes incorporating FNOs in the branch network achieve lower relative errors compared to the remaining combinations. Consistent with the grid-field analyses, Table~\ref{tab:rel_error_spe10m2_scalar} reports the relative error $2$-norms for the tensorized scalar channels. For the DeepONet (FNO+KAN) configuration, the well oil production rate (WOPR) error decreases with increasing model capacity, demonstrating improved predictive accuracy. However, the corresponding improvement for the water cut (WWCT) channel is more modest. In contrast, the DeepONet (FNO+MLP) configuration exhibits substantial error reductions for both WOPR and WWCT as the number of learnable parameters increases from $1$M to $20$M, particularly when comparing the results of Case~A and Case~B. For the standard DeepONet (MLP), increasing the model complexity does not yield meaningful performance gains. The analysis presented in Table~\ref{tab:rel_error_spe10m2_scalar} is influenced by the sparsity of the tensors, which may lead to spurious error amplification. To better assess the models’ performance, we extract the post-processed scalar quantities directly from the tensor channels and compare their temporal behavior. Figure~\ref{fig:wopr_wwct_spe10m2} depicts the oil production rate and water cut curves for all production wells. The hybrid models with FNOs in the branch network show strong agreement with the reference WOPR curves, while the DeepONet (KAN) captures the general dynamics but deviates more noticeably from the ground-truth values. The DeepONet (MLP) configuration failed to capture both the WOPR and WWCT trends, producing zero predictions throughout the entire simulation period.

\begin{table}
    \centering
    \caption{Relative errors for the scalar quantities channels for all schemes and cases.DeepONet (KAN) for cases B  and C, and DeepONet (MLP) for case C, did not fit in the NVIDIA H100 VRAM.}
    \vspace{4pt}
    \begin{tabular}{c c c c}
        \toprule
         \textbf{Hybrid Scheme} & \textbf{Case}  &  \textbf{Relative Error}      & \textbf{Relative Error}           \\
                       &       & \textbf{for the WOPR channel} &   \textbf{for the WWCT channel}  \\
         \midrule
         DeepONet (FNO+KAN) & A & $4.43 \times 10^{-2}$ & $3.25 \times 10^{-2}$ \\
                          & B & $2.23 \times 10^{-2}$ & $1.59 \times 10^{-2}$\\
                          & C & $7.68 \times 10^{-3}$ &  $1.20 \times 10^{-2}$ \\
        \midrule
         DeepONet (FNO+MLP) & A & $1.69 \times 10^{-1}$ & $1.05 \times 10^{-1}$\\
                          & B & $1.61 \times 10^{-1}$ & $9.04 \times 10^{-2}$\\
                          & C & $4.87 \times 10^{-3}$ &  $8.04 \times 10^{-3}$ \\
         \midrule
         DeepONet (KAN)     & A & $1.28 \times 10^{-1}$ & $1.15 \times 10^{-1}$\\
                          & B &  -  &  - \\
                          & C & - & - \\
         \midrule
         DeepONet (MLP) & A & $7.61 \times 10^{-1}$ & $7.76 \times 10^{-1}$\\
                          & B & $8.69 \times 10^{-1}$ & $8.58 \times 10^{-1}$\\
                          & C & - & - \\
        \bottomrule
    \end{tabular}
    \label{tab:rel_error_spe10m2_scalar}
\end{table}

\begin{figure}
    \centering
    \begin{subfigure}{0.85\textwidth}
        \includegraphics[width=0.95\linewidth]{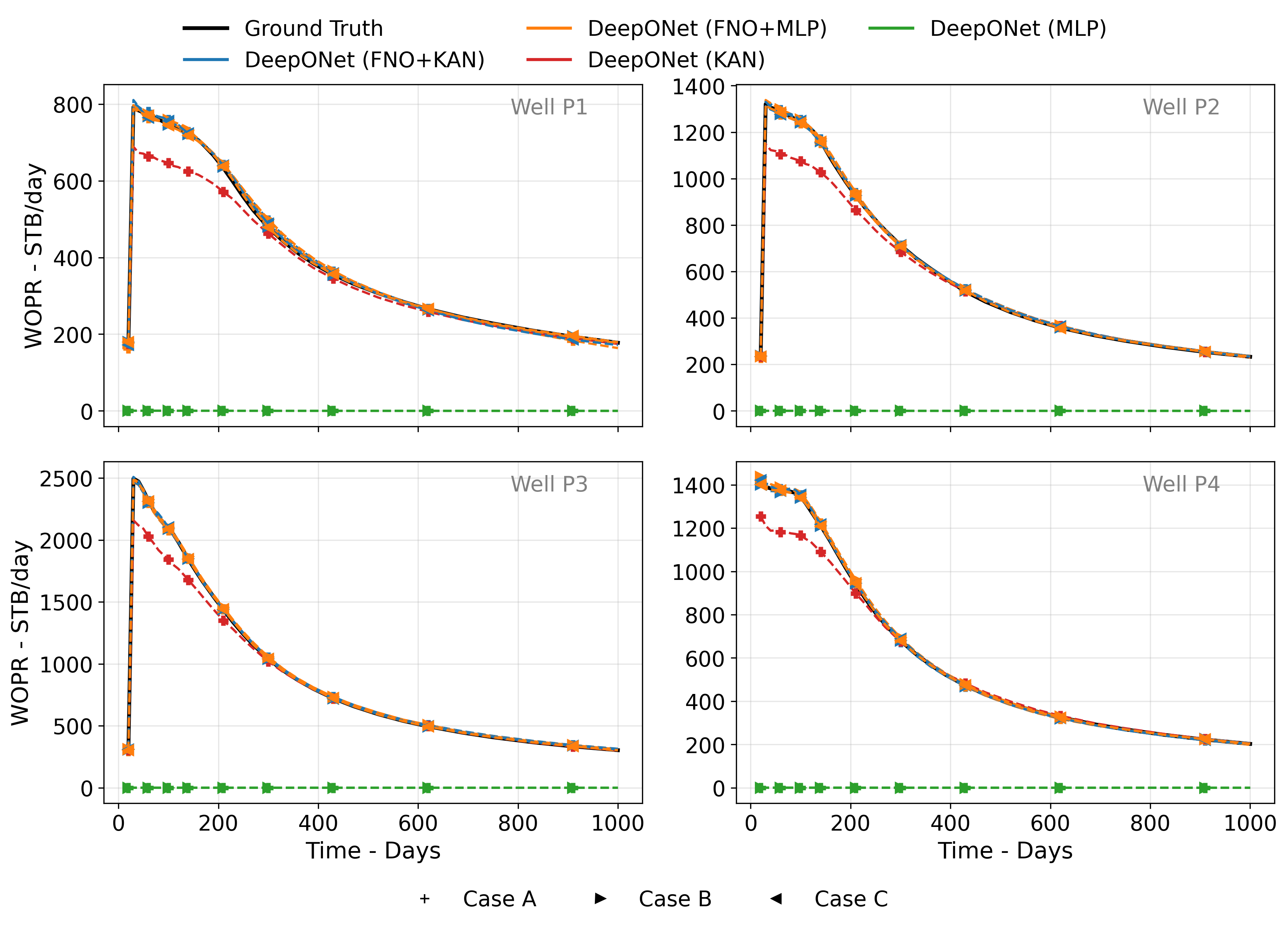}
        \caption{Predictions on WOPR curves for the four wells.}
        \label{fig:preds_wopr_spe10m2}
    \end{subfigure}
    \hfill
    \begin{subfigure}{0.85\textwidth}
        \includegraphics[width=0.95\linewidth]{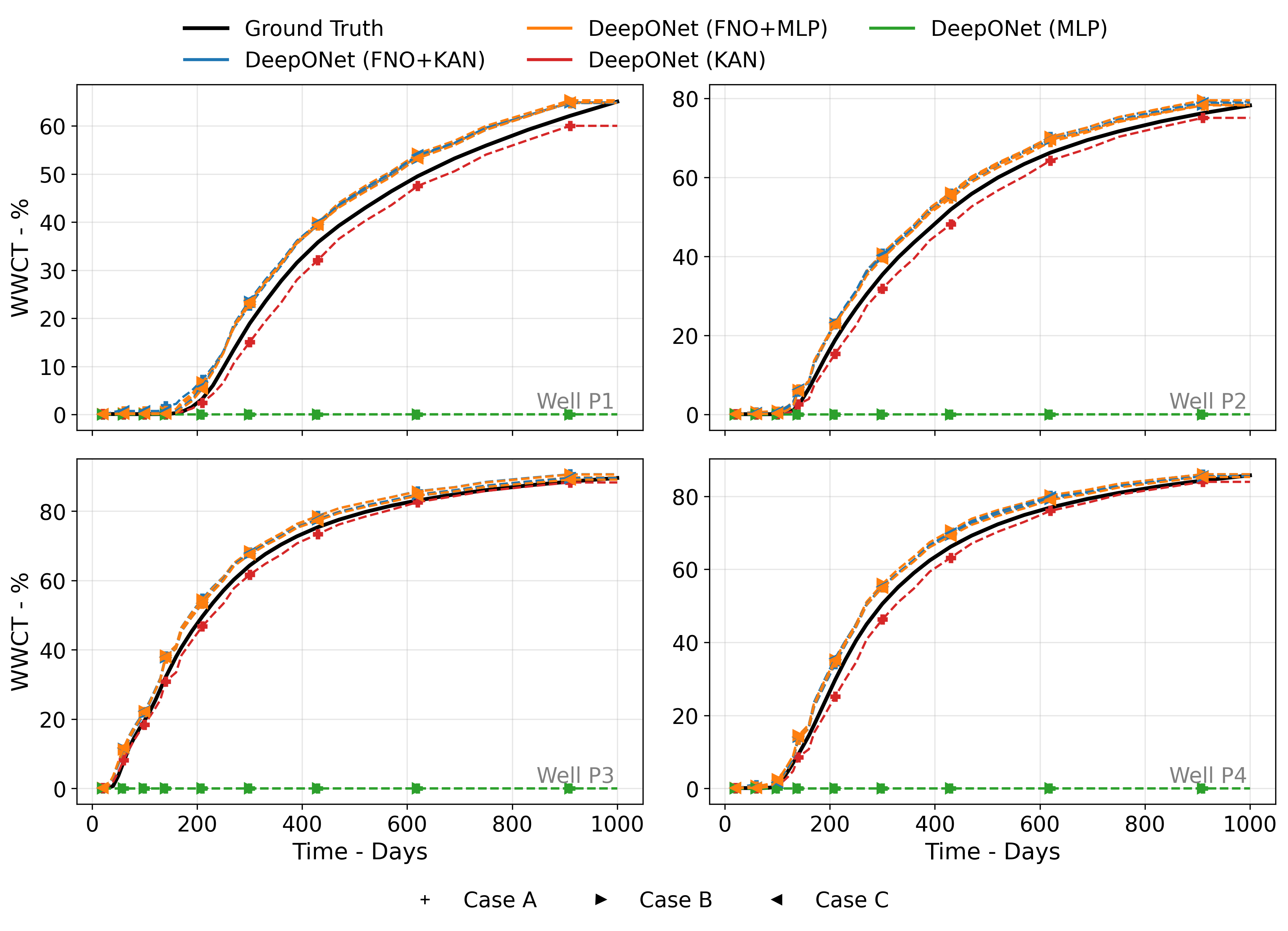}
        \caption{Predictions on WWCT curves for the four wells.}
        \label{fig:preds_wwct_spe10m2}
    \end{subfigure}
    \caption{a) Well Oil Production Rate (WOPR) and b) Well Water Cut (WWCT) curves for all production wells of the SPE10 Model 2. Hybrid configurations incorporating FNOs in the branch network show strong agreement with the reference WOPR and WWCT curves. In contrast, the DeepONet (KAN) captures the general dynamics but exhibits noticeable deviations from the ground truth. The DeepONet (MLP) configuration fails to reproduce either variable, yielding zero predictions throughout the entire simulation period.}
    \label{fig:wopr_wwct_spe10m2}
\end{figure}




\section{Conclusions}
\label{sec:conclusions}

\par
Neural Operators, particularly DeepONets and Fourier Neural Operators (FNOs), have gained increasing attention as surrogate models for porous media flow due to their ability to generalize across parameterized PDE solutions. Their generalization capability is essential for predicting unseen scenarios and reducing reliance on computationally expensive numerical simulations, which are often impractical for many-query tasks. Nevertheless, these architectures in their standard form face well-known challenges, including high memory consumption and difficulties in handling temporal dependencies. To address these limitations, this study explored hybrid neural operator schemes based on the DeepONet framework, integrating complementary architectures, including FNOs, Multi-Layer Perceptrons (MLPs), and Kolmogorov–Arnold Networks (KANs). We theoretically examined the equivalence between MLPs and KANs under specific assumptions as detailed in the appendix of this manuscript. Then, we performed a series of numerical validations of increasing complexity to evaluate the generalization capability of the proposed strategy across distinct porous media applications, ranging from the steady 2D Darcy flow to the $10$th Comparative Solution Project from the Society of Petroleum Engineers (SPE10). We chose four DeepONet configurations, two of which included FNO in the branch networks. The other possible branch and trunk networks were assessed with either KAN or MLP. All models were trained on a NVIDIA H100 GPU with $94$ GB of VRAM. First, we start validating our strategies for the 2D steady Darcy flow. All hybrid configurations demonstrated good performance, yielding comparable spatial approximations and relative errors on the order of $10^{-2}$ for predicting the pressure solution for different permeability inputs. Regarding spatial errors, the distribution across the analyzed samples in the test set showed very similar behavior across all combinations. 
\par 
Next, we assessed the SPE10 benchmark using the black-oil formulation for multiphase flow in porous media, with datasets generated using OPM Flow. For SPE10 Model 1, a smaller 2D case with isotropic yet highly heterogeneous permeability fields, we tested the hybrid models’ ability to generalize both primary variables and post-processed quantities across different injection rates. Primary variables (grid-based fields) were modeled directly, while post-processed quantities, originally represented as time-dependent scalars, were tensorized to match the wells’ spatial locations, thereby acquiring spatial meaning. We observed that incorporating an FNO into the branch network reduced spatial errors in the predicted grid fields compared with hybrid configurations lacking the FNO. For post-processed quantities, all models achieved similarly accurate predictions.

\par
Similar trends were observed for the SPE10 Model 2, a larger 3D case with over $10^6$ grid cells. DeepONet schemes that incorporate FNOs into the branch network exhibit superior generalization when predicting spatially coherent fields, such as water and oil saturations. The hybrid models were evaluated on their ability to predict both the saturation fields and post-processed quantities, including oil production rates and well water cuts, for injection rates not seen during training. Three model sizes were tested, comprising approximately $127$k, $1$M, and $20$M trainable parameters (referred to as Cases A, B, and C, respectively). Hybrid schemes employing FNOs in the branch network successfully scaled up to $20$M parameters, whereas architectures with MLP or KAN branches exceeded hardware memory limitations. In terms of efficiency, previous studies have reported that FNO-based surrogate models for reservoirs with $O(10^6)$ grid cells typically require $O(10^7$–$10^9)$ parameters. In contrast, our hybrid models, which decouple spatio-temporal structures, achieved comparable accuracy with only $O(10^5$–$10^7)$ parameters. This reduction in model complexity directly translates into significant savings in memory usage and computational cost compared with conventional Neural Operator architectures. It is interesting to note that all hybrid configurations preserved the phase balance relation $s_o + s_w = 1$, demonstrating physically consistent predictions. Furthermore, when analyzing absolute spatial error, we observed that larger hybrid models with FNO branches consistently yielded lower spatial errors. When examining post-processed quantities extracted from the output tensors, architectures featuring FNO-based branch networks again outperformed the remaining hybrid schemes, providing more accurate and smoother temporal profiles for both WOPR and WWCT.


\section*{Code availability}

The computational codes used in the simulations will be publicly available if the article is accepted for publication.

\section*{Acknowledgments}\label{ack}

This study was partially financed by the Coordenação de Aperfeiçoamento de Pessoal de Nível Superior-Brasil (CAPES)—Finance Code 001. It is also partially supported by CNPq, Brazilian Petroleum Agency, and ExxonMobil.

\section*{Credit Authorship Contribution Statement}
Ezequiel S. Santos: Conceptualization, Methodology, Software, Validation, Formal Analysis, Visualization, Writing – original draft, Writing – review \& editing. Gabriel F. Barros: Conceptualization, Methodology, Software, Validation, Formal Analysis, Visualization, Writing – original draft, Writing – review \& editing. Amanda C. N. Oliveira: Conceptualization, Methodology, Formal Analysis, Visualization, Writing – original draft, Writing – review \& editing. Romulo M. Silva: Conceptualization, Methodology, Formal Analysis, Visualization, Writing – original draft, Writing – review \& editing. Rodolfo S. M. Freitas: Conceptualization, Methodology, Formal Analysis, Visualization, Writing – original draft, Writing – review \& editing. Dakshina M. Valiveti: Conceptualization, Methodology, Validation, Formal Analysis, Visualization, Writing – original draft, Writing – review \& editing, Resources, Supervision. Xiao-Hui Wu: Conceptualization, Methodology, Validation, Formal Analysis, Visualization, Writing – original draft, Writing – review \& editing, Resources, Supervision. Fernando A. Rochinha: Conceptualization, Methodology, Validation, Formal Analysis, Visualization, Writing – original draft, Writing – review \& editing, Resources, Supervision, Funding Acquisition. Alvaro L. G. A. Coutinho: Conceptualization, Methodology, Validation, Formal Analysis, Visualization, Writing – original draft, Writing – review \& editing, Resources, Supervision, Funding Acquisition.

\section*{Conflict of Interest Statement}
The authors have no conflicts of interest to declare that are relevant to the content of this article.

\section*{Declaration of generative AI and AI-assisted technologies in the manuscript preparation process}
During the preparation of this work the author(s) used Grammarly in order to improve text readability. After using this tool/service, the author(s) reviewed and edited the content as needed and take(s) full responsibility for the content of the published article.

\bibliographystyle{unsrt}
\bibliography{Bibliography/references}

\begin{appendices}
\section{Equivalence between MLP and KAN}
\label{sec:app_a}

In this appendix, we demonstrate that MLPs and KANs are mathematically equivalent in representational capacity, such that the class of functions each architecture can represent is equivalent. Showing that a KAN is equivalent to an MLP under certain conditions allows us to understand that, under specific restrictions, the KAN is no more expressive than a classical MLP, helping to decide when its use is truly necessary and, at the same time, justifying the simplification of the model without compromising performance or the use of computational resources. Under this view, we present a mathematical demonstration of this approximation.

\par
The equivalence proves that, given a set of parameters $(w_{ij}, v_i, d_i, \sigma)$ for the MLP, it is possible to find a set of univariate functions $\psi, \varphi$ for the KAN, and vice versa, such that $f_{\text{mlp}}(\mathbf{x}) \approx f_{\text{kan}}(\mathbf{x})$ for the same class of continuous functions, ensuring that both architectures cover the same functional space under proper parameterization, as shown in Theorem \ref{thm_equivalence}. In particular, the equivalence proves that a single-layer MLP with $2n+1$ neurons can be represented by single-layer KAN with $2n+1$ neurons (which can be generalized to multiple layers), and vice versa.

\begin{lemma}[Spline Approximation]
\label{lem_splines}
Let $g \in C(K)$ be a continuous function and $K \subset \mathbb{R}^{n}$ a compact set. Then, for any $\varepsilon > 0$, there exists an adaptive spline function $S_g(\boldsymbol{\variable}; \boldsymbol{\theta}) = \sum_{k=1}^{m} c_k B_k(\boldsymbol{\variable})$ such that:

\begin{equation}
\|g - S_g(\cdot; \boldsymbol{\theta})\|_{L^\infty(K)} = \sup_{x \in K} |g(\boldsymbol{\variable}) - S_g(\boldsymbol{\variable}; \boldsymbol{\theta})| < \varepsilon
\end{equation}

where $B_k$ is a B-spline basis functions of degree $d \geq 1$, $\boldsymbol{\theta} = \{c_k\}_{k=1}^m \in \mathbb{R}^m$ are the parameters that control the spline behavior, and $m = O(\varepsilon^{-1/d})$ is the
number of parameters necessary to guarantee that the error is less than $\varepsilon$ for smooth functions.
\end{lemma}

\begin{proof}
By the density of splines in $C(K)$, consider a family of splines $\{B_k\}_{k=1}^m$ of degree $d$ defined on a set of knots $T = \{\lambda_1, \dots, \lambda_m\}$. We define the linear space generated by these splines as
\begin{equation}
\mathcal{S}_d^m = \text{span}\{B_k\}_{k=1}^m.    
\end{equation}
Fix $m \to \infty$ and the maximum mesh diameter
\begin{equation}
h = \max_i |\lambda_{i+1} - \lambda_i| \to 0.    
\end{equation}
Thus, we guarantee density in $C(K)$ in the topology of uniform convergence on compact sets. Considering the space obtained in this limit by
\begin{equation}
\mathcal{S}_d^\infty = \bigcup_{m=1}^{\infty} \mathcal{S}_d^m,  
\end{equation}
it follows that $\mathcal{S}_d^\infty$ is dense in $C(K)$. Specifically, for $g \in C^k(\mathbb{R})$ with $k \geq 1$ and any compact set $K \subset \mathbb{R}^{n}$, there exists $S \in \mathcal{S}_d^\infty$ such that
\begin{equation}
\| g - S \|_{L^\infty(K)}  =  \sup_{x \in K} | g -S | \le c\, h^{k+1} \| g^{(k+1)} \|_{L^\infty(K)},
\end{equation}
where the ($k+1$)-th derivative of function $g$, $h$ is the maximum spacing between the knots of $S$ and $c$ is a constant depending only on $d$ and $k$. For general continuous functions $g \in C(K)$, the density follows from the Stone-Weierstrass Theorem \cite{cotter1990stone}. Therefore, there exists $\boldsymbol{\theta}^* \in \mathbb{R}^m$ such that the linear combination of splines
\[
S_g(\cdot; \boldsymbol{\theta}^*) \in \mathcal{S}_d^\infty
\]
approximates $g$ on any compact set with arbitrary precision, guaranteeing the desired density.
\end{proof}


\begin{theorem}[MLP-KAN Equivalence]
\label{thm_equivalence}
Let $\mathcal{L}_{\text{MLP}} = \left\{ f_{\text{mlp}}: \mathbb{R}^n \to \mathbb{R} \mid f_{\text{mlp}}(\boldsymbol{\variable}) = \sum_{i=1}^{n_H} v_i \sigma\left(\sum_{j=1}^{n} w_{ij} \boldsymbol{\variable}_j + d_i\right)\right\}$ be the class of MLPs with non-constant activation function $\sigma: \mathbb{R} \to \mathbb{R}$ and $n_H \geq 2n+1$, and let $\mathcal{L}_{\text{KAN}} = \left\{f_{\text{kan}}: \mathbb{R}^n \to \mathbb{R} \mid f_{\text{kan}}(\boldsymbol{\variable}) = \sum_{i=1}^{2n+1} \psi_i \left( \sum_{j=1}^{n} \varphi_{ij}(\boldsymbol{\variable}_j) \right)\right\}$ be the class of KANs with univariate functions $\psi_i, \varphi_{ij} \in C(K)$. Then:

\begin{enumerate}
    \item $\forall f_{\text{mlp}} \in \mathcal{L}_{\text{MLP}}$, $\exists f_{\text{kan}} \in \mathcal{L}_{\text{KAN}}$ such that $\|f_{\text{mlp}} - f_{\text{kan}}\|_{L^\infty(K)} < \varepsilon$ for any $\varepsilon > 0$ and compact set $K \subset \mathbb{R}^n$.
    
    \item $\forall f_{\text{kan}} \in \mathcal{L}_{\text{KAN}}$ with $\psi_i, \varphi_{ij} \in C^k(\mathbb{R})$ for $k \geq 1$, $\exists f_{\text{mlp}} \in \mathcal{L}_{\text{MLP}}$ such that $\|f_{\text{kan}} - f_{\text{mlp}}\|_{L^\infty(K)} < \varepsilon$ for any $\varepsilon > 0$ and compact set $K \subset \mathbb{R}^n$.
    
    \item There exists a transformation $\Phi: \mathcal{L}_{\text{MLP}} \to \mathcal{L}_{\text{KAN}}$ that is an isomorphism, this is, it is bijective and preserves the structure of the functions.
\end{enumerate}
\end{theorem}

\begin{proof}
\textbf{(1) Representation of MLP as KAN with Adaptive Splines}

Let $f_{\text{mlp}} \in \mathcal{L}_{\text{MLP}}$ with parameters $\{v_i, w_{ij}, d_i\}_{i=1,j=1}^{n_H,n}$ and non-constant activation function $\sigma \in C(K)$. Assuming that $\max_i |v_i| > 0$ (otherwise, the function would be identically zero), by Lemma \ref{lem_splines}, for any $\varepsilon > 0$ and compact set $K \subset \mathbb{R}^{n}$, there exists a spline function $S_\sigma(u; \boldsymbol{\theta}) = \sum_{k=1}^{m} c_k B_k(u)$ such that:

\begin{equation}
\sup_{u \in K} |\sigma(u) - S_\sigma(u; \boldsymbol{\theta})| < \frac{\varepsilon}{2n_H \max_i |v_i|}
\end{equation}

where $B_k$ is a B-spline basis functions and $\boldsymbol{\theta} = \{c_k\}_{k=1}^m$ are the control parameters.

We define the equivalent KAN functions as:

\begin{align}
\varphi_{ij}(\boldsymbol{\variable}_j) &= w_{ij} \boldsymbol{\variable}_j + \frac{d_i}{n} \\
\psi_i(u) &= v_i S_\sigma(u; \boldsymbol{\theta}_i)
\end{align}

where $\boldsymbol{\theta}_i$ are specific parameters for each function $\psi_i$. Then:

\begin{align}
f_{\text{kan}}(\boldsymbol{\variable}) &= \sum_{i=1}^{n_H} \psi_i\left(\sum_{j=1}^n \varphi_{ij}(\boldsymbol{\variable}_j)\right) \\
&= \sum_{i=1}^{n_H} v_i S_\sigma\left(\sum_{j=1}^n w_{ij} \boldsymbol{\variable}_j + d_i; \boldsymbol{\theta}_i\right)
\end{align}

Therefore:

\begin{align}
|f_{\text{mlp}}(\boldsymbol{\variable}) - f_{\text{kan}}(\boldsymbol{\variable})| &= \left|\sum_{i=1}^{n_H} v_i \sigma\left(\sum_{j=1}^n w_{ij} \boldsymbol{\variable}_j + d_i\right) - \sum_{i=1}^{n_H} v_i S_\sigma\left(\sum_{j=1}^n w_{ij} \boldsymbol{\variable}_j + d_i; \boldsymbol{\theta}_i\right)\right| \\
&\leq \sum_{i=1}^{n_H} |v_i| \left|\sigma\left(\sum_{j=1}^n w_{ij} \boldsymbol{\variable}_j + d_i\right) - S_\sigma\left(\sum_{j=1}^n w_{ij} \boldsymbol{\variable}_j + d_i; \boldsymbol{\theta}_i\right)\right| \\
&< \sum_{i=1}^{n_H} |v_i| \cdot \frac{\varepsilon}{2n_H \max_i |v_i|} \leq \frac{\varepsilon}{2n_H} \sum_{i=1}^{n_H} \frac{|v_i|}{\max_i |v_i|} \leq \frac{\varepsilon}{2n_H} \cdot n_H = \frac{\varepsilon}{2} < \varepsilon
\end{align}

\textbf{(2) Representation of KAN as MLP with Non-Linear Activation}

Let $f_{\text{kan}} \in \mathcal{L}_{\text{KAN}}$ with univariate functions $\psi_i, \varphi_{ij} \in C^k(\mathbb{R})$ for $k \geq 1$. The smoothness condition $k \geq 1$ is essential to guarantee that the univariate functions are differentiable and, therefore, can be approximated by neural networks.

First, we approximate the inner functions $\varphi_{ij}$. For each $\varphi_{ij} \in C^k(\mathbb{R})$, there exist parameters $\{a_{ijk}, b_{ijk}, c_{ij}\}$ and a non-constant activation function $\sigma_1 \in C(K)$ such that:
\begin{equation}
\varphi_{ij}(\boldsymbol{\variable}_j) \approx \sum_{k=1}^{m_{ij}} a_{ijk} \sigma_1(\boldsymbol{\variable}_j - b_{ijk}) + c_{ij}
\end{equation}
with approximation error less than $\varepsilon/(4n \cdot n_H)$.

Next, we approximate the outer functions $\psi_i$. For each $\psi_i \in C^k(\mathbb{R})$, there exist parameters $\{\alpha_{ik}, \beta_{ik}, \gamma_i\}$ and a non-constant activation function $\sigma_2 \in C(K)$ such that:

\begin{equation}
\psi_i(u) \approx \sum_{k=1}^{m_i} \alpha_{ik} \sigma_2(u - \beta_{ik}) + \gamma_i
\end{equation}

with approximation error less than $\varepsilon/(2n_H)$.

Now we construct the equivalent MLP. We define the composite activation function as:

\begin{equation}
\sigma_{\text{comp}}(u) = \sum_{k=1}^{M} \alpha_k \sigma_2(u - \beta_k) + \gamma
\end{equation}

where $M = \sum_{i=1}^{n_H} m_i + \sum_{i=1}^{n_H} \sum_{j=1}^n m_{ij}$ and the parameters $\{\alpha_k, \beta_k, \gamma\}$ are organized to incorporate all approximations of the functions $\psi_i$ and $\varphi_{ij}$.

The structure of the equivalent MLP is defined as:

\begin{equation}
f_{\text{mlp}}(\boldsymbol{\variable}) = \sum_{i=1}^{n_H} v_i \sigma_{\text{comp}}\left(\sum_{j=1}^n w_{ij} \boldsymbol{\variable}_j + d_i\right)
\end{equation}

where each application of $\sigma_{\text{comp}}$ corresponds to the approximations of the original univariate functions.

The construction guarantees that:

\begin{align}
\|f_{\text{kan}} - f_{\text{mlp}}\|_{L^\infty(K)} &\leq \sum_{i=1}^{n_H} \left\|\psi_i\left(\sum_{j=1}^n \varphi_{ij}(\boldsymbol{\variable}_j)\right) - v_i \sigma_{\text{comp}}\left(\sum_{j=1}^n w_{ij} \boldsymbol{\variable}_j + d_i\right)\right\|_{L^\infty(K)} \\
&< \sum_{i=1}^{n_H} \left(\frac{\varepsilon}{2n_H} + n \cdot \frac{\varepsilon}{4n \cdot n_H}\right) \\
&< \sum_{i=1}^{n_H} \frac{\varepsilon}{n_H} = \varepsilon
\end{align}

\textbf{(3) Structure Preservation and Bijectivity}

The transformation $\Phi: \mathcal{L}_{\text{MLP}} \to \mathcal{L}_{\text{KAN}}$ is defined by:

\begin{equation}
\Phi(f_{\text{mlp}}) = \sum_{i=1}^{n_H} \psi_i\left(\sum_{j=1}^n \varphi_{ij}(\boldsymbol{\variable}_j)\right)
\end{equation}

where $\psi_i(u) = v_i S_\sigma(u; \boldsymbol{\theta}_i)$ and $\varphi_{ij}(\boldsymbol{\variable}_j) = w_{ij} \boldsymbol{\variable}_j + d_i/n$.

The injectivity of $\Phi$ follows from the uniqueness of the spline decomposition, and the surjectivity in the class of smooth functions is guaranteed by Lemma \ref{lem_splines} which establishes the density of splines in $C(K)$.
\end{proof}

\begin{corollary}
\label{cor_restrictions}
Let $\mathcal{L}_{\text{linear}} = \{f: \mathbb{R}^n \to \mathbb{R} \mid f(\boldsymbol{\variable}) = \sum_{i=1}^n a_i \boldsymbol{\variable}_i + b\}$ be the class of affine functions and $\mathcal{L}_{\text{identity}} = \{f: \mathbb{R}^n \to \mathbb{R} \mid f(\boldsymbol{\variable}) = \sum_{i=1}^{n_H} v_i \left(\sum_{j=1}^n w_{ij} \boldsymbol{\variable}_j + d_i\right)\}$ be the class of MLPs with identity activation. Then:

\begin{equation}
\mathcal{L}_{\text{MLP}} \cap \mathcal{L}_{\text{KAN}} \supsetneq \mathcal{L}_{\text{linear}} \cup \mathcal{L}_{\text{identity}}
\end{equation}

where the inclusion is strict, demonstrating that the equivalence can be established for a significantly broader class of functions, as established by Theorem \ref{thm_equivalence}.
\end{corollary}

\begin{proof}
By Theorem \ref{thm_equivalence}, we know that $\mathcal{L}_{\text{MLP}} \subseteq \mathcal{L}_{\text{KAN}}$ (Item 1) and $\mathcal{L}_{\text{KAN}} \subseteq \mathcal{L}_{\text{MLP}}$ (Item 2), therefore $\mathcal{L}_{\text{MLP}} = \mathcal{L}_{\text{KAN}}$. Hence, $\mathcal{L}_{\text{MLP}} \cap \mathcal{L}_{\text{KAN}} = \mathcal{L}_{\text{MLP}} = \mathcal{L}_{\text{KAN}}$. Since $\mathcal{L}_{\text{linear}} \cup \mathcal{L}_{\text{identity}} \subsetneq \mathcal{L}_{\text{MLP}}$ (linear functions and with identity activation are a proper subset of general MLPs), we have that $\mathcal{L}_{\text{MLP}} \cap \mathcal{L}_{\text{KAN}} \supsetneq \mathcal{L}_{\text{linear}} \cup \mathcal{L}_{\text{identity}}$, demonstrating that the MLP-KAN equivalence holds for a much larger class of functions than just linear or identity activation functions.
\end{proof}

\begin{theorem}[Uniform Convergence]
\label{thm_convergence}
Let $\{f_{\text{mlp}}^{(j)}\}_{j=1}^{\infty} \subset \mathcal{L}_{\text{MLP}}$ be a sequence of MLPs converging uniformly to $f^* \in C(\mathbb{R}^n)$ on a compact set $K \subset \mathbb{R}^n$. Then, there exists a corresponding sequence $\{f_{\text{kan}}^{(j)}\}_{j=1}^{\infty} \subset \mathcal{L}_{\text{KAN}}$ such that:

\begin{equation}
\lim_{j \to \infty} \|f_{\text{mlp}}^{(j)} - f_{\text{kan}}^{(j)}\|_{L^\infty(K)} = 0
\end{equation}

and $f_{\text{kan}}^{(j)} \to f^*$ uniformly on $K$.
\end{theorem}

\begin{proof}
By uniform convergence, for any $\varepsilon > 0$, there exists $N \in \mathbb{N}$ such that for $j \geq N$:

\begin{equation}
\|f_{\text{mlp}}^{(j)} - f^*\|_{L^\infty(K)} < \frac{\varepsilon}{2}
\end{equation}

By Theorem \ref{thm_equivalence}, for each $f_{\text{mlp}}^{(j)}$, there exists $f_{\text{kan}}^{(j)} \in \mathcal{L}_{\text{KAN}}$ such that:

\begin{equation}
\|f_{\text{mlp}}^{(j)} - f_{\text{kan}}^{(j)}\|_{L^\infty(K)} < \frac{\varepsilon}{2}
\end{equation}

Therefore, by the triangle inequality:

\begin{align}
\|f_{\text{kan}}^{(j)} - f^*\|_{L^\infty(K)} &\leq \|f_{\text{kan}}^{(j)} - f_{\text{mlp}}^{(j)}\|_{L^\infty(K)} + \|f_{\text{mlp}}^{(j)} - f^*\|_{L^\infty(K)} \\
&< \frac{\varepsilon}{2} + \frac{\varepsilon}{2} = \varepsilon
\end{align}

for $j \geq N$, establishing the uniform convergence, as guaranteed by Theorem \ref{thm_equivalence}.
\end{proof}


\numberwithin{table}{section}
\section{Hybrid architectures structure}
\label{sec:app_c}

\par 
In this appendix, we describe in more detail how our hybrid schemes are built. For cases A, B, and C, we have, respectively, $127$k, $1$M, and $20$M learnable parameters in our architectures. Tables \ref{tab:params_case_a}, \ref{tab:params_case_b}, and \ref{tab:params_case_c} show how the parameters are distributed for branch and trunk networks in each hybrid scheme. In architectures that include FNO models in their branch networks, we specify the number of FNO blocks and Fourier modes for each case. For architectures using MLP and KAN branch networks, more layers and neurons are added for each case. Table \ref{tab:models_hparams} shows the different hyperparameters for each case.

\begin{table}[htpb]
\centering
\caption{Hybrid DeepONet architectures – Case A ($127$K parameters).}
\setlength{\tabcolsep}{10pt}
\renewcommand{\arraystretch}{1.25}
\begin{tabular}{c c c c}
\toprule
\textbf{Hybrid Scheme} & \textbf{Network} & \textbf{Block} & \textbf{\# Parameters} \\
\midrule
\multirow{5}{*}{DeepONet (FNO+KAN)} 
 & \multirow{3}{*}{Branch: FNO}  & Lifting             & 224 \\
 &                               & FNO Blocks  & 111{,}136 \\
 &                               & Projection           & 6{,}088 \\
 & Trunk: KAN & KANLinear Layers & 9{,}792 \\
 \cmidrule{3-4}
 &  & \textbf{Total} & \textbf{127{,}240} \\
\midrule
\multirow{4}{*}{DeepONet (FNO+MLP)} 
 & \multirow{3}{*}{Branch: FNO}  & Lifting & 224 \\
 & & FNO Blocks  & 111{,}136 \\
  & & Projection  & 6{,}088 \\
 & Trunk: MLP & Dense Layers & 9{,}634 \\
  \cmidrule{3-4}
 &  & \textbf{Total} & \textbf{127{,}082} \\
\midrule
\multirow{2}{*}{DeepONet (KAN)} 
 & Branch: KAN & KAN Layers & 97,920 \\
 & Trunk: KAN & KAN Layers & 29,400 \\
  \cmidrule{3-4}
 &  & \textbf{Total} & \textbf{127{,}320} \\
\midrule
\multirow{2}{*}{DeepONet (MLP)} 
 & Branch: MLP & Dense Layers & 101{,}000 \\
 & Trunk: MLP & Dense Layers & 26,110 \\
   \cmidrule{3-4}
 &  & \textbf{Total} & \textbf{127{,}110} \\
\bottomrule
\end{tabular}
\label{tab:params_case_a}
\end{table}

\begin{table}
\centering
\caption{Hybrid DeepONet architectures – Case B ($1$M parameters).}
\setlength{\tabcolsep}{10pt}
\renewcommand{\arraystretch}{1.25}
\begin{tabular}{c c c c}
\toprule
\textbf{Hybrid Scheme} & \textbf{Network} & \textbf{Block} & \textbf{\# Parameters} \\
\midrule
\multirow{5}{*}{DeepONet (FNO+KAN)} 
 & \multirow{3}{*}{Branch: FNO}  & Lifting              & 704 \\
 &                               & FNO Blocks           & 1,050,688 \\
 &                               & Projection           & 7,072 \\
 & Trunk: KAN & KANLinear Layers & 21,760 \\
 \cmidrule{3-4}
 &  & \textbf{Total} & \textbf{1,080,224 } \\
\midrule
\multirow{4}{*}{DeepONet (FNO+MLP)} 
 & \multirow{3}{*}{Branch: FNO}  & Lifting              & 704 \\
 &                               & FNO Blocks           & 1,050,688 \\
 &                               & Projection           & 7,072 \\
 & Trunk: MLP & Dense Layers & 4,354 \\
  \cmidrule{3-4}
 &  & \textbf{Total} & \textbf{1,062,818} \\
\midrule
\multirow{2}{*}{DeepONet (MLP)} 
 & Branch: MLP & Dense Layers & 1,023,624 \\
 & Trunk: MLP & Dense Layers & 12,770 \\
   \cmidrule{3-4}
 &  & \textbf{Total} & \textbf{1,036,394} \\
\bottomrule
\end{tabular}
\label{tab:params_case_b}
\end{table}

\begin{table}
\centering
\caption{Hybrid DeepONet architectures – Case C ($20$M parameters).}
\setlength{\tabcolsep}{10pt}
\renewcommand{\arraystretch}{1.25}
\begin{tabular}{c c c c}
\toprule
\textbf{Hybrid Scheme} & \textbf{Network} & \textbf{Block} & \textbf{\# Parameters} \\
\midrule
\multirow{5}{*}{DeepONet (FNO+KAN)} 
 & \multirow{3}{*}{Branch: FNO}  & Lifting              & 2,432 \\
 &                               & FNO Blocks           & 20,500,800 \\
 &                               & Projection           & 27,400 \\
 & Trunk: KAN & KANLinear Layers & 21,760 \\
 \cmidrule{3-4}
 &  & \textbf{Total} & \textbf{20,552,392} \\
\midrule
\multirow{4}{*}{DeepONet (FNO+MLP)} 
 & \multirow{3}{*}{Branch: FNO}  & Lifting              & 2,432 \\
 &                               & FNO Blocks           & 20,500,800 \\
 &                               & Projection           & 27,400 \\
 & Trunk: MLP & Dense Layers & 5,410 \\
  \cmidrule{3-4}
 &  & \textbf{Total} & \textbf{20,525,802} \\
\bottomrule
\end{tabular}
\label{tab:params_case_c}
\end{table}

\newpage
\begin{landscape}
\begin{table}
\centering
\caption{Model hyperparameters for each case.}
\setlength{\tabcolsep}{10pt}
\renewcommand{\arraystretch}{1.25}
\begin{tabular}{@{} c  c  c  c  c @{}}
\toprule
\textbf{Case} & \textbf{DeepONet (FNO+KAN)} & \textbf{DeepONet (FNO+MLP)} & \textbf{DeepONet (KAN)} & \textbf{DeepONet (MLP)} \\
\midrule
A &
\begin{tabular}[c]{@{}c@{}}Fourier modes: 3\\ FNO blocks: 2\\ Projection layer size: 32 \\ KAN layers: 2\\ KAN layer size: 16 \\ KAN spline order: 2 \\ \end{tabular} &
\begin{tabular}[c]{@{}c@{}}Fourier modes: 4\\ FNO blocks: 2\\ Projection layer size: 32 \\ MLP layers: 2\\ MLP layer size: 32  \end{tabular} &
\begin{tabular}[c]{@{}c@{}} Branch KAN layers: 3\\ Branch KAN layer size: \{64, 64, 32\} \\ Branch KAN spline order : 2 \\ Trunk KAN layers: 3\\ Trunk KAN layer size: 30 \\ Trunk KAN spline order : 3 \end{tabular} &
\begin{tabular}[c]{@{}c@{}} Branch MLP layers: 6 \\ Branch MLP layer size: 128 \\ Trunk MLP layers: 4\\ Trunk MLP layer size: 82\end{tabular} \\
\midrule
B &
\begin{tabular}[c]{@{}c@{}}Fourier modes: 4\\ FNO blocks: 2\\ Projection layer size: 32 \\ KAN layers: 2\\ KAN layer size: 32 \\ KAN spline order: 3 \\ \end{tabular} &
\begin{tabular}[c]{@{}c@{}}Fourier modes: 4\\ FNO blocks: 2\\ Projection layer size: 32 \\ MLP layers: 3\\ MLP layer size: 32  \end{tabular} &
-  &
\begin{tabular}[c]{@{}c@{}} Branch MLP layers: 16 \\ Branch MLP layer size: 256 \\ Trunk MLP layers: 3\\ Trunk MLP layer size: 64\end{tabular} \\
\midrule
C &
\begin{tabular}[c]{@{}c@{}}Fourier modes: 5\\ FNO blocks: 5\\ Projection layer size: 64 \\ KAN layers: 3\\ KAN layer size: 32 \\ KAN spline order: 3 \end{tabular} &
\begin{tabular}[c]{@{}c@{}}Fourier modes: 5\\ FNO blocks: 5\\ Projection layer size: 64 \\ MLP layers: 4\\ MLP layer size: 32  \end{tabular} &
 -  &
 -  \\
\bottomrule
\end{tabular}
\label{tab:models_hparams}
\end{table}

\end{landscape}
\end{appendices}

\end{document}